\documentclass[12pt]{article}
\usepackage{graphicx}
\usepackage{amsmath,amsthm,amssymb,mathrsfs,mathtools}
\usepackage{unicode-math}
\usepackage[OT1]{fontenc}
\usepackage{etoc}
\usepackage{float}
\usepackage{titlesec}
\usepackage{multirow,array}
\usepackage{diagbox}
\usepackage[top=1.25in, bottom=1.25in, left=1.25in, right=1.25in, footskip=54pt]{geometry}
\usepackage{fancyhdr}

\titleformat{\section}
		{\scshape \large \bfseries\filcenter}
         {\thesection}
        {0.5em}
        { }
        []
\titleformat{\subsection}
        {\normalfont\bfseries}
        {\thesubsection}
        {0.5em}
        {}
        []

\allowdisplaybreaks

\def\subsectiontitle{}
\def\subsubsectiontitle{}
\usepackage{tikz}
\usetikzlibrary{automata, positioning}

\usepackage{kpfonts}
\usepackage{inconsolata}

\usepackage{makecell}
\usepackage{cellspace}
\usepackage[shortlabels]{enumitem}

\usepackage[bottom]{footmisc} 
\makeatletter
\ifdefined\@outputbox@appendfootnotes
  \count\footins=600
  \dimen\footins=75pt
  \newdimen\FootnoteHang
  \let\FootnoteAppendStandard\@outputbox@appendfootnotes
  \def\@outputbox@appendfootnotes{%
    \global\FootnoteHang\ifvoid\footins\z@\else
      \dimexpr(\ht\footins+\dp\footins)*(1000-\count\footins)/1000+\dp\footins\relax\fi
    \FootnoteAppendStandard
    \ifdim\FootnoteHang>\z@\@outputbox@append{\kern-\FootnoteHang}\fi}
\fi
\makeatother
\usepackage[colorlinks,breaklinks=true]{hyperref}
\usepackage{breakcites}
\newcommand{\abs}[1]{\left| #1 \right|} 

\usepackage{xcolor}
\usepackage{framed}
\makeatletter
\newif\ifEg@after 
\ifdefined\directlua
  \AddToHook{para/begin}{\ifnum\directlua{tex.sprint(tex.nest.ptr)}=\@ne \global\Eg@afterfalse\fi}
\fi
\newenvironment{egblock}{%
  \ifEg@after\OuterFrameSep=\z@ \else\OuterFrameSep=10pt plus 2pt minus 2pt\relax\fi 
  \MakeFramed{\advance\hsize-\width \FrameRestore \@nobreaktrue}}%
 {\global\OuterFrameSep=10pt plus 2pt minus 2pt\relax\endMakeFramed\ifdefined\directlua\global\Eg@aftertrue\fi}
\makeatother
\AtBeginDocument{%
	\hypersetup{%
    linkcolor={red!50!black},%
    citecolor={blue!50!black},%
    urlcolor={blue!80!black}%
}%
}%

\usepackage[nameinlink,noabbrev,sort,capitalise]{cleveref}
\makeatletter
\def\ps@pprintTitle{%
 \let\@oddhead\@empty
 \let\@evenhead\@empty
 \def\@oddfoot{\emph{Preliminary draft -- comments welcome}\hfill\emph{This draft: \today}}%
 \let\@evenfoot\@oddfoot}
\makeatother
\usepackage{etoolbox}
\patchcmd{\pprintMaketitle}
  {\fi\hrule}
  {\fi\ifvoid\extrainfobox\else\unvbox\extrainfobox\par\vskip10pt\fi\hrule}
  {}{}

\newsavebox\extrainfobox

\newtheorem{proposition}{Proposition}
\crefname{proposition}{Proposition}{Propositions}

\newtheorem{theorem}{Theorem}
\crefname{theorem}{Theorem}{Theorems}

\newtheorem{corollary}{Corollary}
\crefname{corollary}{Corollary}{Corollaries}

\newtheorem{lemma}{Lemma}
\crefname{lemma}{Lemma}{Lemmas}

\crefname{assumption}{Assumption}{Assumptions}

\crefname{axiom}{Axiom}{Axioms}

\newtheorem{definition}{Definition}
\crefname{definition}{Definition}{Definitions}

\theoremstyle{remark}
\newtheorem{remark}{Remark}
\crefname{remark}{Remark}{Remarks}

\theoremstyle{remark}

\crefname{claim}{Claim}{Claims}

\theoremstyle{definition}
\newtheorem{example}{Example}
\crefname{example}{Example}{Examples}

\crefname{problem}{Problem}{Problems}

\newcommand{\norm}[1]{\left\lVert #1 \right\rVert}
\usepackage{indentfirst}

\allowdisplaybreaks

\let\oldfootnote\footnote
\renewcommand\footnote[1]{\oldfootnote{\hspace{.4mm}#1}}
\renewcommand{\footnotesize}{\scriptsize}

\makeatletter
\renewenvironment{proof}[1][\proofname] {\par\pushQED{\qed}\normalfont\topsep6\p@\@plus2\p@\@minus2\p@\relax\trivlist\item[\hskip\labelsep\bfseries#1\@addpunct{.}]\ignorespaces}{\popQED\endtrivlist\@endpefalse}
\makeatother
\let\oldFootnote\footnote
\newcommand\nextToken\relax

\renewcommand\footnote[1]{%
    \oldFootnote{#1}\futurelet\nextToken\isFootnote}

\newcommand\isFootnote{%
    \ifx\footnote\nextToken\textsuperscript{,}\fi}

\usepackage[american]{babel}
\usepackage[authoryear,round]{natbib}
\renewcommand\d{\mathrm{d}}
\newcommand\E{\mathbb{E}}
\newcommand\R{\mathbb{R}}

\newcommand\bm{\symbf}

\titlespacing\section{0pt}{6pt minus 2pt}{3pt}
\titlespacing\subsection{0pt}{3pt minus 1pt}{1.5pt}
\titlespacing\subsubsection{0pt}{1.5pt}{1.5pt}
 \titlespacing*{\paragraph}{0pt}{1ex plus .5ex minus .2ex}{0.5em}
\newcommand\LayoutDisplaySkips{%
  \setlength{\abovedisplayskip}{9pt plus 3pt minus 5pt}%
  \setlength{\belowdisplayskip}{9pt plus 3pt minus 5pt}%
  \setlength{\abovedisplayshortskip}{0pt plus 3pt}%
  \setlength{\belowdisplayshortskip}{5pt plus 3pt minus 3pt}}
\expandafter\gappto\csname normalsize\space\endcsname{\LayoutDisplaySkips}
\makeatletter
\g@addto@macro\appendix{%
  \def\@listI{\leftmargin\leftmargini \parsep 5\p@ \@plus2.5\p@ \@minus\p@
    \topsep 6\p@ \@plus2\p@ \@minus2\p@ \itemsep5\p@ \@plus2.5\p@ \@minus\p@}%
  \let\@listi\@listI \@listi}
\makeatother

\begin{document}

\title{Bundling Complements\thanks{Stanford Graduate School of Business. Email: weijie.zhong@stanford.edu. I thank Yi-Chun Chen, Roberto Corrao and Laura Doval for valuable comments and suggestions.}}
\date{}
\author{\makebox[.25\linewidth]{Weijie Zhong}}

\maketitle
\thispagestyle{empty}
\vspace{-1.5em}

\begin{abstract}
\noindent I study a monopolist selling goods to a buyer with combinatorial preferences. The dual of the seller's problem is a minimum-cost flow of the differential virtual value. I characterize when the type distribution alone determines the differential virtual value. The distribution then determines the \emph{required directions}, along which the dual transports mass, and the preference determines the \emph{covered directions} of each bundle; optimality is full coverage. This separation enables comparative statics in the preference: raising a complementarity premium on bundles containing a core of goods widens their covered directions. Two premium thresholds describe the optimum: above a lower threshold the grand bundle must be offered; above an upper threshold a \emph{core-peripheral} menu --- a bundled core with optional add-ons not sold standalone --- is optimal, and every optimal menu is core-peripheral. The upper threshold is finite if and only if the menu is \emph{inclusive}: it excludes no near-top buyer.
\end{abstract}

\bigskip

\noindent\textit{Keywords:} multidimensional screening, bundling, combinatorial preferences, complementarity, flow duality, optimal transport, core-peripheral menus, mechanism design.

\smallskip

\noindent\textit{JEL Classification:} D42, D82, D86.

\newpage
\pagenumbering{arabic}
\setcounter{page}{1}

\section{Introduction}\label{sec:intro}

Many products and services are sold with a core component paired with optional add-ons unavailable standalone: AI assistants bundle advanced reasoning models with a base subscription, enterprise vendors license security and compliance modules only atop their core platform, and cable operators have long sold specialty channels only as add-ons to a basic tier. In each, an add-on is worth more paired with the core than alone --- an advanced model inside a workflow that also runs faster base models for routine steps, a security module integrated with the platform whose data it monitors. Such cross-good interactions generate \emph{combinatorial preferences}: a bundle's value is not the sum of its parts. The bundling decision depends on the interaction, which in these examples is complementarity.

This paper asks how the optimal menu changes as goods become stronger complements. A monopolist sells $N$ indivisible goods to a buyer with a multidimensional private type. A \emph{valuation map} assigns to each type a \emph{valuation vector} --- a value for every bundle. The paper makes two contributions: a method for screening with combinatorial preferences, and the comparative statics of the optimal menu with respect to complementarity.

The method is a duality followed by two reductions. I write revenue as a linear functional of the buyer's utility, whose weight at each valuation vector is the \emph{differential virtual value}. I show that the maximal revenue equals the minimum cost of a flow of the differential virtual value across valuation vectors; I call this result the \emph{flow duality}. I then derive two reductions. The first reduction represents the differential virtual value \emph{on types}. This representation exists if and only if scaling up a buyer's valuation vector gives, to first order, another buyer's valuation vector. I call this property \emph{tangency}. Under it, the flow reduces to transport between types, as in \citet{DDT2017} and \citet{KM2019}. The second reduction makes the differential virtual value on types independent of the preference. For a family of preferences that includes the additive one, this independence holds if and only if the family's valuation maps are linear. Under this \emph{separation}, the distribution and the preference enter through two sets of directions. The type distribution determines the \emph{required directions}: the displacements of a transport plan for the differential virtual value on types. The preference determines the \emph{covered directions}: each bundle covers the directions along which it extracts the most surplus. These directions form its \emph{bundle cone}. A menu is optimal when the cone of the bundle each buyer chooses covers the required directions through that buyer.

The two sets of directions address the central question of multidimensional screening: along which directions the seller screens. The required directions are the directions along which screening is most profitable. Along them the seller screens harder until feasibility binds: no allocation creates a larger utility difference. The preference, in turn, determines the \emph{screening power} of each bundle: along its covered directions, selling the bundle maximizes the utility difference. In these terms, a menu is optimal when the bundle each buyer chooses exercises screening power along the required directions through that buyer. Wider cones therefore make a menu optimal for more distributions. The comparative statics with respect to complementarity apply this insight.

I consider a model of complementary goods in which a \emph{premium} $\alpha$ multiplies the value of every bundle that contains a fixed set of \emph{core} goods. The premium keeps the valuation map linear, which leaves the required directions unchanged, and widens the covered directions of bundles that contain the core. Call a menu \emph{core-peripheral} if it sells the core only as a block and each add-on only with the core, and \emph{inclusive} if it serves every buyer who values some good at its maximum. Two thresholds in the premium describe the optimum. Above a \emph{lower threshold} $\alpha_*$, the grand bundle must be sold (\Cref{thm:grand-bundle-nec}): in every optimal menu, a set of buyers of positive measure buys it. Above an \emph{upper threshold} $\alpha^*$, an inclusive core-peripheral menu is optimal under a regularity condition, and every optimal menu is core-peripheral (\Cref{thm:core-bundle}). Inclusivity is also necessary: no premium makes a menu optimal if it does not serve buyers who value some good at its maximum (\Cref{thm:incl-necessity}).

In two cases the regularity condition follows from primitives. For pure bundling it reduces to Myersonian regularity of the bundle's value. Pure bundling is then uniquely optimal at every high premium if it gives a positive surplus to every buyer who values some good at its maximum, and suboptimal at every premium otherwise (\Cref{thm:pure-bundle}). For \emph{two-tier} menus --- the core alone and the core with one add-on --- the condition holds for a class of product distributions. The two-tier menu is then uniquely optimal at every high premium (\Cref{thm:twotier}). 

\Cref{sec:model} describes the model and the premium. \Cref{sec:method} develops the method: the flow duality and its two reductions. \Cref{sec:main} applies the method to complementary goods: the threshold results, each opened by a two-good example; two cases in which the regularity condition follows from primitives; and the three regimes in practice. \Cref{sec:discussion} explains why the premium family takes the core-peripheral form and extends the analysis to valuation maps that are not tangent. Proofs are in the Appendix and the Online Supplement. 

\subsection{Literature review}\label{sec:lit}

\paragraph{Optimal menus under additive multi-good preferences.} \citet{McAfeeMcMillan1988}, \citet{Armstrong1996}, and \citet{RochetChone1998} studied exclusion, bunching, and convexity in multidimensional screening, where the buyer has several dimensions of private information. A central feature of the additive multi-good optimum, emphasized by \citet{ManelliVincent2007} and analyzed for two goods by \citet{Pavlov2011}, is that randomization (lotteries over bundles) is intrinsic: deterministic posted prices are generically suboptimal. \citet{manelli2006bundling} characterize when deterministic \emph{bundle} pricing nevertheless solves the multi-good monopoly problem. \citet{Rochet2024} surveys recent developments.\footnote{The additive optimum has also been studied through approximation: \citet{Armstrong1999}, \citet{HartNisan2017,HartNisan2019}, and \citet{BriestEtAl2015}.}

\paragraph{Optimal-transport duality.} A strand of this literature characterizes the additive multi-good optimum through a transport dual. \citet{DDT2017} establish the duality and the saddle-point characterization of the optimum against a dual cost $c$, applying the framework to several canonical settings. \citet{KM2019} give a short linear-programming proof of the duality. In these papers the cost $c$ is a primitive. Hence the preference structure of the screening problem is implicit in its definition. Here I characterize how the cost changes with the preference. The characterization gives comparative statics of the covered directions in the preference. The transport plan of \citet{KM2019} is the multiplier on the pairwise feasibility constraints between types. The flow of \Cref{thm:flow-duality} is the local counterpart of this multiplier. A second strand writes the dual as a Beckmann flow problem in two settings: additive multi-bidder auctions with a positive type density \citep{KSTZ2022}, and the monopolist problem with the differential virtual value on types as the data \citep{BK2023}. By contrast, the flow dual of \Cref{thm:flow-duality} holds for any distribution of the valuation vector, including distributions whose differential virtual value is not a measure.\footnote{In both, valuations are additive and the differential virtual value is a measure.}

\paragraph{Sufficient conditions for pure bundling.} A separate strand identifies sufficient conditions on the type distribution for the additive optimum to be pure bundling or separate sales. \citet{McAfeeMcMillanWhinston1989} show that mixed bundling dominates separate sales when two non-complementary goods have stochastically independent values. For two goods with independent values on bounded intervals, \citet{GiannakopoulosKoutsoupias2018} give sufficient conditions under which the optimal menu has a closed form, with pure bundling as a special case. \citet{menicucci2015bundling} establish pure-bundling optimality for two additive, independently distributed goods under non-negative virtual valuations. A parallel literature on robust mechanism design characterizes the optimum under ambiguity: separate sales are robust when only the marginals are known \citep{Carroll2017}, pure bundling when the prior is exchangeable across goods and the seller does not know what the buyer learns \citep{DebRoesler2023}, and categorical bundling under moment-based ambiguity sets \citep{CheZhong2025}. As the seller's information about the buyer becomes more precise, pure bundling approaches the first-best revenue at the same rate as the optimal menu \citep{FrickIijimaIshii2026}. These papers vary the type distribution, or the seller's information about it, at fixed additive preferences. By contrast, I fix an arbitrary type distribution and vary complementarity. Under sufficient complementarity and a regularity condition, a core-peripheral menu is optimal. With every good in the core, this menu is pure bundling.

\paragraph{Combinatorial preferences and bundling extensions.} The closer comparators are papers on nested menus and papers that allow non-additive bundle valuations. \citet{Yang2023nested} develops a theory of nested bundling in a one-dimensional type-space setting amenable to the Myersonian approach. In the same setting, \citet{Feng2025} gives sufficient conditions for the optimality of tree bundling, in which several upgrade paths branch from one base bundle. \citet{Ghili2023} and \citet{Kwak2026} also assume one-dimensional types. \citet{BBHS2021} assume additive values and order finite multidimensional types by their values or by their marginal rates of substitution. In these three papers a nested menu --- a base bundle with one upgrade path --- follows from an ordering of types. In \citet{Kwak2026}, with additive values, this menu is equivalent to separate sales. The two-tier menus of \Cref{sec:apps-twotier} are the multidimensional counterpart, obtained under a product density with no ordering across types. \citet{haghpanah2021pure} study multidimensional types with general non-additive bundle valuations. Through a path decomposition that reduces the problem to a family of one-dimensional restrictions, they identify a monotone-ratio condition under which pure bundling is optimal.\footnote{The sufficient conditions of \citet{Yang2023nested} and \citet{haghpanah2021pure} are invariant to the premium (see footnote~\ref{fn:rays}). Neither yields the comparative statics in complementarity developed here. Yang's setting is nested in the Myersonian case of \Cref{sec:typespace}.} \citet{Thanassoulis2004} studies substitutes, and \citet{BikhchandaniMishra2024} show that selling heterogeneous goods under symmetric additive values is equivalent to selling identical units under decreasing marginal values. \citet{ChenGui2026} apply primal marginal-revenue conditions to the same premium specification: with two goods and a symmetric density they characterize the optimal menu, and with more goods their conditions are necessary.\footnote{I cite the latest version of their paper. Where their results and mine overlap, I cite its 2025 version \citep{ChenGui2025}, which precedes this paper.} By contrast, the dual certificates of \Cref{thm:core-bundle} give sufficient conditions for any number of goods and any core, and the same dual proves the necessity of the grand bundle and of inclusivity (\Cref{thm:grand-bundle-nec,thm:incl-necessity}). Relative to the additive and the combinatorial literatures, this paper makes complementarity a geometric primitive --- the directions each bundle covers --- that varies independently of the distribution.

\section{Model}\label{sec:model}

This section describes the buyer, the seller's menus and revenue, and the core-peripheral preferences in which the comparative statics are posed. \Cref{sec:method} develops the method in this environment.

\subsection{Preferences}\label{sec:model-prims}

\paragraph{Goods, types, and valuations.} Consider $N$ goods and a single buyer. Let $[N] = \{1, \dots, N\}$ index goods and let $\mathcal{B} := 2^{[N]} \setminus \{\varnothing\}$ denote the full family of nonempty bundles. The buyer's primitive type $x$ lies in the type space $D \subset \R^d$ and has a density $f$. The seller offers a subfamily $\mathcal{F} \subseteq \mathcal{B}$ of $K := \abs{\mathcal{F}}$ bundles. For each $B \in \mathcal{F}$, the buyer's valuation $v_B(x)$ is the value of bundle $B$ to type $x$. Collect these into the valuation map
\[
v(x) := \bigl(v_B(x)\bigr)_{B \in \mathcal{F}} \in \R^K,
\]
which assigns to each type its valuation vector. Its image $Y := v(D)$ is the \emph{valuation manifold}.

I separate the type space $D$ from the valuation manifold $Y$ to accommodate the classical analysis of additive preferences, in which $D$ is the space of values of individual goods. The first reduction (\Cref{sec:typespace}) shows when working on $D$ is without loss.\footnote{The distinction is unnecessary when types are valuation vectors ($d = K$). In most applications $d < K$, for instance $d = N$ with additive values and $d = 1$ with one-dimensional types. The valuation manifold is then a lower-dimensional subset of $\R^K$, which the flow duality of \Cref{sec:duality} treats as the type space.}

\paragraph{Standing assumptions.} The type space $D$ is compact, is the closure of its interior, and has a piecewise-smooth boundary. It contains a \emph{bottom type} $x_0$ with $x_0 \preceq x$ for every $x \in D$. The density $f \in C^2(D)$ is strictly positive on $D$.\footnote{The density conditions are assumed for the necessity theorems (\Cref{thm:incl-necessity,thm:grand-bundle-nec}). The sufficiency certificates --- \Cref{thm:core-bundle} and its applications --- extend to $C^1$ densities positive on $\operatorname{int} D$ with integrable boundary behavior. This extension applies to the singular Beta family of \Cref{sec:apps-pure} (\Cref{os:single-cell} of the Online Supplement).} Each valuation $v_B$ extends to a smooth function on a neighborhood of $D$ and is nondecreasing and nonconstant on $D$. Then $v(x_0)$ is the coordinatewise-least point of $Y$, because $x_0 \preceq x$ for every $x \in D$ and each $v_B$ is nondecreasing. I call $y_0 := v(x_0)$ the \emph{bottom valuation vector}. The Jacobian $Jv(x) \in \R^{K \times d}$ of $v$ has rank $d$ at every $x \in D$ (the \emph{rank condition}): every direction of motion in the type space changes some bundle value to first order.\footnote{In \Cref{sec:main}, $D$ is a box in $\R^N$, $x_0 = 0$ and $\mathcal{F} = \mathcal{B}$. There the rows of $Jv$ for the singleton bundles form an invertible diagonal matrix. Hence $Jv$ has rank $d = N$.}

\subsection{Menus and revenue}\label{sec:menus}

\paragraph{Menus.} Write $\mathcal{F}_0 := \mathcal{F} \cup \{\varnothing\}$, with $v_\varnothing \equiv 0$. An \emph{option} $(\sigma, p)$ offers a bundle lottery $\sigma \in \Delta(\mathcal{F}_0)$ at a price $p \ge 0$. Let $\bar p := \max_B \max_D v_B$ be the highest value of any bundle. An option priced above $\bar p$ is never taken. I therefore consider prices in $[0, \bar p]$. The seller posts a \emph{menu} $M$: a compact set of options that contains the outside option $(\delta_\varnothing, 0)$. Its options are indexed by $\ell$, with $\ell = 0$ the outside option. A \emph{finite menu} lists its options as $M = \{(\sigma_\ell, p_\ell)\}_{\ell=0}^L$.

\paragraph{Ambient and indirect utility.} Let $\mathcal{Q} := \{q \in \R_+^K : \mathbf{1}^\top q \le 1\}$ be the simplex of bundle lotteries, with $q_B$ the probability of bundle $B$. Write $q^\ell := (\sigma_\ell(B))_{B \in \mathcal{F}} \in \mathcal{Q}$ for the lottery of option $\ell$. A menu's utility $u(y) := \max_\ell \{q^\ell \cdot y - p_\ell\}$ is defined on all of $\R^K$. The maximum is attained because the menu is compact. I call $u$ the \emph{ambient utility}. It is convex and coordinatewise nondecreasing, with subgradients in $\mathcal{Q}$. Conversely, a \emph{direct mechanism} --- an allocation $q(\cdot) \in \mathcal{Q}$ and a payment on $Y$ --- is incentive compatible and individually rational if and only if its utility is the restriction to $Y$ of a convex function $u$ on $\R^K$ with subgradients in $\mathcal{Q}$, $q(y) \in \partial u(y)$ and $u(y_0) \ge 0$ \citep{RochetChone1998}. The ambient utility defines the canonical \emph{indirect utility} $U(x) := u(v(x))$ on the type space. The same indirect utility may correspond to different ambient utilities.\footnote{The indirect utility determines $u$ only on $Y$. Its extension to $\R^K$ as a convex function with subgradients in $\mathcal{Q}$ is not unique. Under tangency, revenue is a function of $U$ alone, whatever the extension (\Cref{lem:revenue-invariance}).}

\paragraph{Revenue.} The payment is given by the seller-favorable tie-breaking rule $p(y) := \max\{p_\ell : q^\ell \cdot y - p_\ell = u(y)\}$. The allocation $q(y)$ of the selected option satisfies $q(y) \in \partial u(y)$ and $p(y) = q(y) \cdot y - u(y)$. Let $\Pi := v_\#(f\, \d x)$ be the distribution of the valuation vector $y = v(x)$. The seller chooses the menu to maximize its \emph{revenue}
\[
R(M) := \int p(y)\, \d \Pi(y) = \int \bigl[q(y) \cdot y - u(y)\bigr]\, \d \Pi(y).
\]
The normalization $u(y_0) = 0$ is without loss: raising the price of every option $\ell \ge 1$ by $u(y_0)$ keeps every type's participation and raises revenue by $u(y_0)$.

\subsection{Core-peripheral preferences and the question}\label{sec:premium}

The model so far allows any combinatorial preference. The application of this paper is to the products of the introduction: a core sold with add-ons, where an add-on is worth more paired with the core than alone. A \emph{core-peripheral preference} models this complementarity with one parameter: a premium that multiplies the value of every bundle containing the core.

\paragraph{The premium.} Fix a nonempty set $C \subseteq [N]$ of \emph{core goods} and a premium $\alpha > 0$. The core-peripheral preference values each bundle $B$ at
\begin{equation}\label{eq:standing-v}
v_B(x) = \alpha^{\bm{1}\{C \subseteq B\}} \sum_{i \in B} x_i.
\end{equation}
The goods outside the core form the set $P := [N] \setminus C$ of \emph{peripheral goods} --- the add-ons of the introduction. I call \eqref{eq:standing-v} the \emph{core-peripheral form} of the valuation map. The \emph{premium family} is the set of bundles that earn the premium: every bundle containing the core. When $C = [N]$, the family is the grand bundle $[N]$ alone. The premium measures complementarity: valuations are complementary when $\alpha > 1$, additive when $\alpha = 1$, and substitute-like when $0 < \alpha < 1$. Each valuation map of the core-peripheral form is linear in the type.\footnote{Experimental studies support this specification: they elicit willingness to pay for individual goods and for their bundles, and in several settings find bundle values proportional to the sum of component values \citep{CrossWebbODoherty2026, PopkowskiHaubl2010, RousuBeachCorrigan2008}. \Cref{sec:main-segments} summarizes the first two. The elicited values identify $\alpha$ as a ratio: for $C = [N]$, the ratio of bundle value to the sum of component values; for a general core, the ratio of an add-on's incremental value on top of the core to its standalone value, $v_{C \cup S} - v_C = \alpha\, v_S$ for every $S \subseteq P$.}

\paragraph{The question.} A menu is \emph{core-peripheral} if it sells the core only as a block and each peripheral good only with the core. I ask whether complementarity justifies core-peripheral bundling: with the type distribution fixed, is there a premium above which a core-peripheral menu maximizes revenue? The heuristic answer is yes: an add-on worth more with the core should be sold with the core. A screening argument supports the opposite answer: selling an inferior product --- here a peripheral good on its own --- sorts buyers by their values and can raise revenue even when every buyer values the good more with the core.\footnote{With additive values, a menu that sells the goods separately alongside the bundle earns more than the bundle alone \citep{AdamsYellen1976, McAfeeMcMillanWhinston1989}.} The question is which force prevails, and at which premium. \Cref{sec:method,sec:main} use the following example throughout.

\begin{example}[Two-good uniform benchmark]\label{eg:uniform}
Let $N = 2$, $D = [0,1]^2$, $f \equiv 1$ and $\mathcal{F} = \mathcal{B}$, with the coordinates of $v(x)$ ordered as $\{1\}$, $\{2\}$, $\{1,2\}$. With the core $C = \{1,2\}$, only the grand bundle earns the premium: $v(x) = (x_1,\ x_2,\ \alpha(x_1 + x_2))$. The valuation manifold $Y$ is then the parallelogram with vertices $0$, $(1,0,\alpha)$, $(0,1,\alpha)$ and $(1,1,2\alpha)$ in the plane $z_{12} = \alpha(z_1 + z_2)$ of $\R^3$. With the core $C = \{1\}$, the bundles $\{1\}$ and $\{1,2\}$ earn it: $v(x) = (\alpha x_1,\ x_2,\ \alpha(x_1 + x_2))$. Then $Y$ is the parallelogram with vertices $0$, $(\alpha,0,\alpha)$, $(0,1,\alpha)$ and $(\alpha,1,2\alpha)$.
\end{example}

\section{The Method: Flow Duality and Two Reductions}\label{sec:method}

The question of \Cref{sec:premium} calls for a method that handles combinatorial preferences in screening. This section develops one: a duality (\Cref{sec:duality}) and two reductions (\Cref{sec:typespace,sec:bundle-to-cone}). For linear valuation maps, the two reductions lead to the bundle-cone certificate (\Cref{prop:saddle-D}), posed on the type space in terms of the differential virtual value on types and the bundle cones. \Cref{sec:main} works with this certificate. A reader interested in the comparative statics may go directly to \Cref{sec:bundle-to-cone}.

\subsection{Flow duality}\label{sec:duality}

\textbf{Differential virtual value.} Revenue is a linear functional of the ambient utility. To write the functional, I work on the box $X := \prod_{B} [\min_D v_B, \max_D v_B]$, which contains $Y$ and has nonempty interior. Write $C^1(X)$ for the functions that are continuously differentiable on a neighborhood of $X$, and let $\mathcal{U}^1$ denote the convex functions in $C^1(X)$ with gradients in $\mathcal{Q}$ and $u(y_0) = 0$. For $\phi \in C^1(X)$, define
\begin{equation}\label{eq:TP}
\langle T_\Pi, \phi \rangle := \phi(y_0) + \int_X \bigl[y \cdot \nabla \phi(y) - \phi(y)\bigr]\, \d \Pi(y),
\end{equation}
that is, $T_\Pi = \delta_{y_0} - \Pi - \nabla \cdot (y\Pi)$, with measures extended by zero outside $X$. For $u \in \mathcal{U}^1$, the mechanism with allocation $\nabla u$ earns $\langle T_\Pi, u \rangle$. Revenue thus integrates the buyer's utility, which is the integral of the allocation, against $T_\Pi$. Integrated by parts, it is the integral of the allocation against a virtual-value density whose differential is $T_\Pi$, as in Myerson's virtual surplus.\footnote{With one good valued at the type $t \in [0, 1]$, $T_\Pi([\theta, 1]) = \varphi(\theta) f(\theta)$ with $\varphi$ Myerson's virtual valuation. Hence $T_\Pi$ is the differential of the virtual-value density $\varphi f$. The revenue $\langle T_\Pi, u \rangle = \int \varphi\, u'\, \d F$ is then Myerson's virtual surplus.} I therefore call $T_\Pi$ the \emph{differential virtual value}.

\textbf{Primal problem and multipliers.} The seller's primal problem is
\begin{equation}\label{eq:flow-primal}\tag{P}
\sup_{u\in C^1(X)}\ \langle T_\Pi, u \rangle
\qquad \text{s.t.} \qquad
\begin{aligned}[t]
&\nabla u \in \mathcal{Q} && \text{(feasibility)},\\
&u \text{ convex} && \text{(IC)}.
\end{aligned}
\end{equation}
Feasibility restricts the gradient at every point of $X$. Its multiplier therefore has a direction at every point: it is a \emph{flow}, a finite Radon measure $J$ on $X$ with values in $\R^K$. The test inequality of feasibility bounds the \emph{gain} $\int_X \nabla u \cdot \d J$ of the allocation along the flow by the largest gain $\mathcal{A}(J) := \sup_q \int_X q \cdot \d J$ over continuous allocations $q$ with values in $\mathcal{Q}$. Feasibility holds if and only if the test inequality holds for every flow. The multiplier of incentive compatibility is the \emph{ironing multiplier}: a linear functional $\Lambda$ on $C^1(X)$ that is nonnegative at every convex coordinatewise nondecreasing function, written $\Lambda \succeq_{\mathrm{mcvx}} 0$. Under feasibility, which makes $u$ nondecreasing, its test inequality, $\langle \Lambda, u \rangle \ge 0$ for every such $\Lambda$, holds if and only if $u$ is convex.

\textbf{Lagrangian and dual problem.} Adding to the objective the slack of each test inequality at its multiplier gives the Lagrangian for $u\in C^1(X)$
\begin{equation}\label{eq:lagrangian}\tag{L}
\mathcal{L}(u; J, \Lambda) := \langle T_\Pi, u \rangle + \langle \Lambda, u \rangle + \Bigl[\mathcal{A}(J) - \int_X \nabla u \cdot \d J\Bigr].
\end{equation}
The primal \eqref{eq:flow-primal} is the max-min $\sup_u \inf_{J, \Lambda} \mathcal{L}$. The dual is the min-max $\inf_{J, \Lambda} \sup_u \mathcal{L}$. Weak duality follows from the max-min inequality. The inner supremum is finite only when the part of $\mathcal{L}$ that is linear in $u$ vanishes. The multipliers are then \emph{balanced}: $T_\Pi = -\nabla \cdot J - \Lambda$, where $\langle -\nabla \cdot J, \phi \rangle := \int_X \nabla \phi \cdot \d J$. Hence the dual problem is
\begin{equation}\label{eq:flow-dual}\tag{D}
\inf_{J, \Lambda}\ \mathcal{A}(J)
\qquad \text{s.t.} \qquad
T_\Pi = -\nabla \cdot J - \Lambda, \quad \Lambda \succeq_{\mathrm{mcvx}} 0.
\end{equation}
Substituting $\Lambda = -\nabla \cdot J - T_\Pi$ reduces the constraint of \eqref{eq:flow-dual} to the \emph{dominance} $-\nabla \cdot J \succeq_{\mathrm{mcvx}} T_\Pi$. A flow that satisfies it is \emph{admissible}. The objective $\mathcal{A}(J)$ of \eqref{eq:flow-dual} is the \emph{dual cost} of the flow. A unit of flow in direction $z$ costs $c(z) := \max_{q \in \mathcal{Q}} q \cdot z = \max\bigl\{0,\ \max_{B \in \mathcal{F}} z_B\bigr\}$, the largest surplus differential that any allocation extracts along $z$.

\begin{theorem}[Flow duality]\label{thm:flow-duality}
For every probability measure $\Pi$ on $X$, the maximal revenue over menus equals the values of \eqref{eq:flow-primal} and \eqref{eq:flow-dual}. The supremum over menus and the infimum in \eqref{eq:flow-dual} are attained.
\end{theorem}

\noindent The proof is in \Cref{app:duality-flow}.

The duality is posed on the valuation manifold. It certifies a menu by an admissible flow whose dual cost equals the menu's revenue. By contrast, the comparative statics hold the type distribution on $D$ fixed and vary the preference. To analyze them, I need a representation of the problem on the type space: on the valuation manifold, the distribution $\Pi$ of valuation vectors and the differential virtual value $T_\Pi$ move with the valuation map. Two reductions provide this representation: \Cref{sec:typespace} represents the problem on primitive types through the indirect utility, and \Cref{sec:bundle-to-cone} makes this representation independent of the preference.

\subsection{First reduction: representation through the indirect utility}\label{sec:typespace}

The first reduction returns the analysis from $Y$ to $D$. On $D$ the buyer's utility is the indirect utility $U$ of \Cref{sec:menus}. The reduction asks whether $U$ is sufficient for revenue: whether two mechanisms with the same indirect utility earn the same revenue. In general they do not. One incentive constraint determines each payment: the comparison with the scaled-up valuation vector $(1 + \varepsilon) y$, which raises every bundle value at the same rate. When the scaled-up vector leaves the valuation manifold, it is, to first order, the valuation vector of no buyer. The comparison is then with an imaginary buyer. The indirect utility is defined on real buyers only. Hence it contains no information about this comparison. I introduce a tangency condition on the valuation map under which the comparison is between real buyers. Tangency is necessary and sufficient for a representation of $T_\Pi$ on types through $U$.

\textbf{Radial field and tangency.} Under the rank condition the type space has at each $x$ one direction that scales the valuation vector to first order: the \emph{radial field} $a_v(x) := (Jv(x)^\top Jv(x))^{-1} Jv(x)^\top v(x)$, which is the least-squares solution of $Jv(x)\, a = v(x)$. Its \emph{residual} $r_v := v - Jv\, a_v$ is the component of the valuation vector normal to the manifold. For a linear map, $a_v(x) = x$.

\begin{definition}[Tangency condition]\label{def:tangent}
The valuation map $v$ is \emph{tangent} if $Jv(x)\, a_v(x) = v(x)$ for every $x \in D$.
\end{definition}

\noindent In other words, the radial field points to real buyers: the scaled-up valuation vector $(1 + \varepsilon) v(x)$ is, to first order, the valuation vector of the type $x + \varepsilon a_v(x)$. The comparison that determines the payment is then a comparison between buyers. Hence the indirect utility determines the payment. Formally, with $\nabla U = Jv^\top q$ at $\lambda_D$-almost every $x$, where $\lambda_D$ is Lebesgue measure on $D$, the value of the allocation decomposes as $q \cdot v = a_v \cdot \nabla U + q \cdot r_v$. Under tangency the residual term vanishes.\footnote{For any vector field $a : D \to \R^d$, $q \cdot v = a \cdot \nabla U + q \cdot r_a$ with $r_a := v - Jv\, a$. \Cref{sec:apps-curvature} uses this decomposition.} The payment of any mechanism implementing $u$ is then $t(x) = a_v(x) \cdot \nabla U(x) - U(x)$ at $\lambda_D$-a.e.\ $x$ (\Cref{lem:revenue-invariance}).

\textbf{Differential virtual value on types.} Under tangency, revenue depends on the mechanism only through the indirect utility: it equals $\int_D [a_v \cdot \nabla U - U]\, f\, \d x$. The divergence theorem produces, for a smooth field $a$, the signed measure
\[
\mu(A) := \bm{1}_A(x_0) + \int_{\partial D \cap A} f\, a \cdot n\, \d\sigma - \int_A \bigl[\nabla \cdot (fa) + f\bigr]\, \d x,
\]
where $n$ is the outward normal and $\sigma$ the surface measure on $\partial D$; write $\mu := \mu_{a_v}$.\footnote{When $D = [\underline t, \bar t]$ and $v$ is not tangent, the Myerson field $a := (1 - F)/f$, with $F$ the distribution function, gives $\mu_a = 0$. The flow $J := v_\#(f \psi\, \d t)$ with $\psi := v - a v'$ yields Myerson's virtual-surplus bound. \citet{Yang2023nested} and \Cref{sec:apps-curvature} study this case.} With the normalization $U(x_0) = u(y_0) = 0$ of \Cref{sec:menus}, revenue is $\int_D U\, \d\mu$. The same divergence theorem applied to $\phi \circ v$, for which $a_v \cdot \nabla(\phi \circ v) = v \cdot \nabla\phi(v)$ under tangency, gives $\langle T_\Pi, \phi \rangle = \int_D \phi \circ v\, \d\mu$: $T_\Pi = v_\# \mu =: \nu$ is a signed measure on $Y$. I call $\mu$ the \emph{differential virtual value on types}: the representation of $T_\Pi$ on the type space. The positive and negative parts of $\mu$ are called its \emph{sources} and \emph{sinks}: the types at which leaving the buyer more utility raises or lowers revenue.

I next show that tangency is also necessary for this representation.

\begin{proposition}[Characterization of tangency]\label{prop:tangency-char}
The following are equivalent:
\begin{enumerate}[label=\textnormal{(\roman*)}]
\item $v$ satisfies the tangency condition;
\item $\langle T_\Pi, u \rangle$ depends on $u \in \mathcal{U}^1$ only through $U = u \circ v$;
\item $\langle T_\Pi, u \rangle = \int_D U\, \d\mu_{a_v}$ for every $u \in \mathcal{U}^1$.
\end{enumerate}
\end{proposition}

\noindent The proof is in \Cref{app:reductions}.

\textbf{Transport duality on $D$.} Under tangency the seller's problem is posed on types. Let $\mathcal{U}_D$ be the set of indirect utilities $U = u \circ v$ with $U(x_0) = 0$, where $u$ is convex on $\R^K$ with subgradients in $\mathcal{Q}$. By the revenue formula, the seller solves the primal
\begin{equation}\label{eq:primal-D}
\sup_{U \in \mathcal{U}_D}\ \int_D U\, \d\mu .
\end{equation}
Its Lagrangian is that of \Cref{sec:duality}, with feasibility stated between pairs of types: a buyer's utility exceeds another's by at most the largest surplus differential between their valuation vectors, $U(x) - U(x') \le c(v(x) - v(x'))$. The multiplier of these constraints is a \emph{transport plan}, a finite positive measure $\gamma$ on $D \times D$ with marginals $\gamma_1$ and $\gamma_2$. Balance makes the gap $\gamma_1 - \gamma_2 - \mu$ the ironing multiplier, nonnegative against every $U \in \mathcal{U}_D$, written $\gamma_1 - \gamma_2 \succeq_{\mathrm{mcvx}} \mu$. The dual is the transport problem of \citet{DDT2017} and \citet{KM2019},
\begin{equation}\label{eq:dual-D}
\inf_{\gamma} \int_{D \times D} c\bigl(v(x) - v(x')\bigr)\, \d\gamma(x, x')
\quad \text{s.t.} \quad
\gamma_1 - \gamma_2 \succeq_{\mathrm{mcvx}} \mu .
\end{equation}
The plan moves the sources of $\mu$ to its sinks.

\begin{corollary}[Transport duality]\label{thm:strong-duality}
Under \hyperref[def:tangent]{tangency}, with $Y$ convex, the transport dual \eqref{eq:dual-D} has the value of the flow dual \eqref{eq:flow-dual}, the optimal revenue, and attains it.
\end{corollary}

\noindent The proof, in \Cref{os:transport-proof} of the Online Supplement, combines a strong duality between \eqref{eq:primal-D} and \eqref{eq:dual-D} with \Cref{thm:flow-duality}. Under tangency, the flow of the differential virtual value across valuation vectors is a transport of its representation on types.

\subsection{Second reduction: separation and the bundle cones}\label{sec:bundle-to-cone}

In general, the differential virtual value on types depends on the preference through the radial field. The second reduction asks when it is independent of the preference. Linearity suffices. A linear valuation map is $v(x) = Vx$ for a constant matrix $V \in \R^{K \times d}$, whose rows $V_B$ encode the bundle valuations $v_B(x) = V_B \cdot x$. Its radial field is $a_v(x) = x$. Hence the differential virtual value on types $\mu := \mu_x$ does not depend on $V$. I call this property \emph{separation}: the distribution determines $\mu$, and the preference enters only through the dual cost. The next proposition shows that, for a family of valuation maps that contains the additive one, linearity is also necessary.

\begin{proposition}[Separation]\label{prop:separation}
\begin{enumerate}[label=\textnormal{(\roman*)}]
\item Two tangent valuation maps have the same differential virtual value on types for every positive $C^2$ density if and only if they have the same radial field.
\item Let $D \subset \R^N$ be a box with bottom type $x_0 = 0$, and let a family of tangent valuation maps contain the additive map $v_B(x) = \sum_{i \in B} x_i$.\footnote{Part (ii) holds for every family that contains a linear map, with $D$ a box in $\R^d$.} All members of the family have the same differential virtual value on types for every positive $C^2$ density if and only if every member is linear.
\end{enumerate}
\end{proposition}

\noindent The proof is in \Cref{app:reductions}.

\textbf{Bundle cones.} Under separation the preference enters only through the dual cost, which the bundle cones characterize geometrically. Each bundle valuation $v_B(x) = V_B \cdot x$ has the constant gradient $V_B^\top$, which I call the \emph{gradient of bundle} $B$. The gradients $\nabla U = V^\top q$ of indirect utilities lie in the polytope $V^\top \mathcal{Q} = \mathrm{conv}\{0,\, V_B^\top : B \in \mathcal{F}\} \subset \R^d$. The dual cost of moving between two types depends only on their displacement $h := x - x'$: $c(v(x) - v(x')) = c_V(h) := \sup_{p \in V^\top \mathcal{Q}} p \cdot h = \max_{B \in \mathcal{F}_0} v_B(h)$. Bundle $B$ \emph{attains the dual cost} at $h$, its gradient $V_B^\top$ attaining the supremum, if and only if $h$ lies in the \emph{bundle cone}
\begin{equation}\label{eq:bundle-cone}
K_B^{\mathcal{F}} := \bigl\{h \in \R^d : v_B(h) \ge v_{B'}(h) \text{ for every } B' \in \mathcal{F}_0\bigr\}.
\end{equation}
Both claims follow from support-function calculus on the polytope $V^\top \mathcal{Q}$.\footnote{Formally, $c(v(x) - v(x')) = \sup_{q \in \mathcal{Q}} q \cdot Vh = \sup_{p \in V^\top \mathcal{Q}} p \cdot h$ is the largest value of $p \cdot h$ at the points $0$ and $V_B^\top$ ($B \in \mathcal{F}$) whose convex hull is the polytope. These values are $v_B(h)$ for $B \in \mathcal{F}_0$. Hence the gradient $V_B^\top$ attains the supremum at $h$ if and only if $v_B(h) = \max_{B' \in \mathcal{F}_0} v_{B'}(h)$, which is the condition $h \in K_B^{\mathcal{F}}$.} The cone $K_B^{\mathcal{F}}$ collects the \emph{covered directions} of $B$: the displacements along which $B$ attains the dual cost, extracting the most surplus. Bundle $B$ \emph{covers} each of them. The same formula with $v_\varnothing \equiv 0$ defines $K_\varnothing^{\mathcal{F}}$: the displacements along which no bundle gains. In the additive case $v_B(x) = \sum_{i \in B} x_i$ over the full family $\mathcal{F} = \mathcal{B}$, every bundle cone is an orthant, and $c_V(h) = \sum_i (h_i)_+$ is the separable cost of \citet{DDT2017} and \citet{KM2019}.

\textbf{Certificate.} The cones give a certificate of optimality. For a menu $M$ with options $(\sigma_\ell, p_\ell)$ and ambient utility $u_M$, let the \emph{cell} $\Omega_\ell \subseteq D$ be the set of types taking option $\ell$. The cell $\Omega_0$ of the outside option is called the \emph{exclusion cell}. The cells of the other options are called \emph{selling cells}. On $\Omega_\ell$ the indirect utility $U_M := u_M \circ v$ is affine with the gradient $V^\top q^\ell$ of the active lottery: $U_M(x) = (V^\top q^\ell) \cdot x - p_\ell$.

\begin{proposition}[Bundle-cone certificate]\label{prop:saddle}\label{prop:saddle-D}
Let $v(x) = Vx$ with $Y$ convex, let $M$ be a menu with cells $\Omega_\ell$, and let $\gamma$ be a transport plan that satisfies:
\begin{enumerate}[label=\textnormal{(\roman*)}]
\item for $\gamma$-almost every $(x, x')$, both $x$ and $x'$ lie in the closure of one cell $\Omega_\ell$;\footnote{At an optimal pair the condition is automatic for every plan with absolutely continuous sink marginal (\Cref{rmk:within-cell-support}).}
\item for $\gamma$-almost every $(x, x')$ and any cell $\ell$ with $x, x' \in \overline{\Omega_\ell}$, the displacement $x - x'$ lies in the bundle cone $K_B^{\mathcal{F}}$ of every $B \in \operatorname{supp} \sigma_\ell$;
\item $\gamma_1 - \gamma_2 = \mu$.
\end{enumerate}
Then $(U_M - U_M(x_0), \gamma)$ is optimal for the primal \eqref{eq:primal-D} and the dual \eqref{eq:dual-D}. In particular, the menu $M$ with the price of every option $\ell \ge 1$ raised by $U_M(x_0)$ is revenue-optimal.
\end{proposition}

\noindent When $x_0 = 0$, as in \Cref{sec:main}, $U_M(x_0) = 0$. Hence $M$ itself is revenue-optimal. The proof, in \Cref{os:saddle-proofs} of the Online Supplement, equates the cost of $\gamma$ with the primal value and applies weak duality on $D$.

The proposition answers the central question of multidimensional screening: along which directions the seller screens. Screening creates utility differences between types. With one good there is one direction, the type line. With several goods the direction is the seller's choice: the primitives do not identify which comparisons between types the optimal mechanism uses. The three conditions answer the question together. Condition (iii) gives the \emph{required directions} in terms of the primitive distribution: the plan transports the differential virtual value from its sources to its sinks. Revenue is the total of the utility differences between the sources and the sinks. Raising the utility of a source relative to its sink therefore raises revenue. Hence the seller screens along the required directions, from each sink to its source, until a constraint binds. The constraints enter through conditions (i) and (ii), which give the covered directions: a displacement lies in the cone of a bundle if that bundle creates the largest utility difference among all bundles. Coverage is the joint content of the three conditions: along every required direction, the shadow cost of the feasibility constraint offsets the screening gain. In these terms, a menu is optimal when the bundle each buyer chooses exercises screening power along the required directions through that buyer. The distribution alone determines the required directions, and the preference alone determines the covered ones (\Cref{prop:separation}). Condition (iii) also states that the ironing multiplier $\gamma_1 - \gamma_2 - \mu$ is zero. Only the feasibility constraints then bind: global incentive compatibility is slack. This is the counterpart of no ironing in Myerson's problem.\footnote{\label{fn:one-good-regularity}With one good valued at the type $t \in [0, \bar t]$, $\mu([t, \bar t]) = \varphi(t) f(t)$ with $\varphi$ Myerson's virtual valuation. Under a posted price $p$, a plan with conditions (i)--(iii) exists if and only if $\varphi$ is nonpositive below $p$ and nonnegative above it. Above the price, this sign condition states that the sources dominate the sinks in first-order stochastic dominance. \Cref{sec:main-core} gives its counterpart with several goods.} The comparative static of \Cref{cor:premium} below follows. A premium on the bundles that a menu sells widens their cones and makes the menu optimal for more distributions.

\textbf{A direct comparative static.} Separation also gives a comparative static in the preference. A premium $\alpha$ on a family $\mathcal{S} \subseteq \mathcal{F}$ of bundles multiplies the value of every bundle in $\mathcal{S}$ by $\alpha$ and gives the valuation map $v^\alpha$. The menu $\alpha M$ is $M$ with every price multiplied by $\alpha$. A menu \emph{sells only} bundles in a family $\mathcal{S}$ if almost every type takes the outside option or an option whose lottery puts probability one on $\mathcal{S}$. With types as per-good values ($d = N$), every good \emph{can be sold alone} if every singleton $\{i\}$ is in $\mathcal{F}$, with $v_{\{i\}}(x) = a_i x_i$ for some $a_i > 0$. The type space $D \subseteq \R^N_+$ is \emph{downward closed} if $x \in D$ and $0 \le x' \le x$ imply $x' \in D$.\footnote{The sink-density condition of part (ii) is the strict version of the regularity condition $(N+1)f + x \cdot \nabla f \ge 0$ of \citet{McAfeeMcMillan1988}. Under the weak condition, Proposition~2 of \citet{Pavlov2011b} implies that for some optimal menu, no buyer is partially excluded. Under the strict condition, Proposition~4 of \citet{ChenGui2026} implies this for every optimal menu. Under a premium, part (ii) strengthens this: every optimal menu sells only bundles in the restricted family $\mathcal{S}$.}

\begin{proposition}[Optimality and selection under a premium on the bundles sold]\label{cor:premium}
Let $v(x) = Vx$ with $Y$ convex, let $M$ be an optimal menu, and let $\mathcal{S}$ contain every nonempty bundle in the lotteries of $M$. At every premium $\alpha \ge 1$ on $\mathcal{S}$:
\begin{enumerate}[label=\textnormal{(\roman*)}]
\item the menu $\alpha M$ is optimal under $v^\alpha$, and its revenue is $\alpha R(M)$;
\item if $\alpha > 1$, every good can be sold alone, $D$ is downward closed, and $\mu$ has a strictly positive sink density, then every optimal menu under $v^\alpha$ sells only bundles in $\mathcal{S}$.
\end{enumerate}
\end{proposition}

\noindent The premium does not change the distribution. Hence $\mu$ and the plan that certifies $M$ do not change either. The premium multiplies the gradients of the bundles in $\mathcal{S}$ and widens their cones relative to the other bundles. The displacements of the plan are therefore still covered, and the scaled prices keep every buyer's choice. Part (ii) holds for every optimal menu because one optimal plan of the dual constrains all of them. The proof is in \Cref{os:premium-proof} of the Online Supplement.

\Cref{cor:premium} starts from a menu that is optimal at some premium. \Cref{sec:main} studies when a high premium makes a menu optimal, by comparing the required directions of candidate menus with the bundle cones. \Cref{sec:main-algebra} explains which cones widen.

\begin{egblock}
\paragraph{Example 1, continued.} For the core $\{1,2\}$, \Cref{fig:ex0}(a) draws the parallelogram $Y$ at $\alpha = 1$ and at $\alpha = 2$, with the differential virtual value $T_\Pi$ on each: an atom at the origin, uniform mass on the two far edges --- the images of the top faces $F_1 := \{x_1 = 1\}$ and $F_2 := \{x_2 = 1\}$ --- and a negative density on the interior. The premium tilts the manifold and moves $T_\Pi$ with it. At every premium the parallelogram lies in a plane through the origin, so the radial direction $y \mapsto (1 + \varepsilon) y$ stays in that plane. The valuation map is therefore tangent, with $a_v(x) = x$. By the first reduction, $T_\Pi$ is then represented on the type square (\Cref{fig:ex0}(b)) by
\[
\mu = \delta_0 + \lambda_{F_1} + \lambda_{F_2} - 3\lambda_D ,
\]
whose sources are the origin and the top faces and whose sink is the interior. The map is also linear. Hence, by separation, this representation does not change with the premium: $\mu$ is the same at both premia and for the core $\{1\}$.
\end{egblock}

\begin{figure}[htbp]
\centering
\tikzset{every picture/.style={line width=0.75pt}}
\definecolor{exblue}{RGB}{0,122,255}
\begin{tikzpicture}[x=1cm,y=1cm]
\def\prj#1#2#3{({2.2*(#1)+1.4*(#2)},{-0.6*(#1)+0.9*(#2)+0.5*(#3)})}
\begin{scope}
\begin{scope}[shift={(-1.7,-1.25)}]
\draw[->,gray] (0,0) -- (0.66,-0.18) node[right] {\footnotesize $z_1$};
\draw[->,gray] (0,0) -- (0.42,0.27) node[right] {\footnotesize $z_2$};
\draw[->,gray] (0,0) -- (0,0.5) node[above] {\footnotesize $z_{12}$};
\end{scope}
\draw[gray!70,dashed] \prj{0}{0}{0} -- \prj{1}{0}{0} -- \prj{1}{1}{0} -- \prj{0}{1}{0} -- cycle;
\draw[gray!70,dashed] \prj{1}{0}{0} -- \prj{1}{0}{2};
\draw[gray!70,dashed] \prj{0}{1}{0} -- \prj{0}{1}{2};
\draw[gray!70,dashed] \prj{1}{1}{0} -- \prj{1}{1}{4};
\fill[black,opacity=0.08] \prj{0}{0}{0} -- \prj{1}{0}{1} -- \prj{1}{1}{2} -- \prj{0}{1}{1} -- cycle;
\draw[dashed] \prj{0}{0}{0} -- \prj{1}{0}{1} -- \prj{1}{1}{2} -- \prj{0}{1}{1} -- cycle;
\draw[line width=1.6pt] \prj{1}{0}{1} -- \prj{1}{1}{2};
\draw[line width=1.6pt] \prj{0}{1}{1} -- \prj{1}{1}{2};
\fill[black,opacity=0.12] \prj{0}{0}{0} -- \prj{1}{0}{2} -- \prj{1}{1}{4} -- \prj{0}{1}{2} -- cycle;
\draw \prj{0}{0}{0} -- \prj{1}{0}{2} -- \prj{1}{1}{4} -- \prj{0}{1}{2} -- cycle;
\draw[line width=2pt] \prj{1}{0}{2} -- \prj{1}{1}{4};
\draw[line width=2pt] \prj{0}{1}{2} -- \prj{1}{1}{4};
\fill \prj{0}{0}{0} circle (2.2pt);
\draw[->,exblue,line width=1pt] \prj{0.20}{0.55}{1.50} -- \prj{0.27}{0.7425}{2.025};
\draw[->,exblue,line width=1pt] \prj{0.65}{0.60}{2.50} -- \prj{0.8775}{0.81}{3.375};
\node[below left] at \prj{0}{0}{0} {$b$};
\node[above] at \prj{1}{1}{4} {$\alpha = 2$};
\node[right] at \prj{1}{1}{2} {$\alpha = 1$};
\node[above] at \prj{0.4}{1}{2.8} {$v(F_2)$};
\node[right] at \prj{1}{0.8}{3.6} {$v(F_1)$};
\node at \prj{0.33}{0.33}{1.32} {\footnotesize sink};
\node[exblue,anchor=east] at \prj{0.2}{0.62}{1.64} {\footnotesize radial direction};
\node at (1.9,-1.05) {(a) Valuation manifold};
\end{scope}
\begin{scope}[xshift=6.6cm,yshift=-0.3cm]
\draw[->,gray] (0,0) -- (3.0,0) node[below] {$x_1$};
\draw[->,gray] (0,0) -- (0,3.0) node[left] {$x_2$};
\fill[black,opacity=0.12] (0,0) rectangle (2.4,2.4);
\draw (0,0) rectangle (2.4,2.4);
\draw[line width=2pt] (2.4,0) -- (2.4,2.4);
\draw[line width=2pt] (0,2.4) -- (2.4,2.4);
\fill (0,0) circle (2.2pt);
\draw[->,exblue,line width=1pt] (1.08,0.48) -- (1.46,0.65);
\draw[->,exblue,line width=1pt] (0.48,1.20) -- (0.65,1.62);
\draw[->,exblue,line width=1pt] (1.44,1.44) -- (1.94,1.94);
\node[below left] at (0,0) {$0$};
\node[right] at (2.4,1.3) {$F_1$};
\node[above] at (1.1,2.4) {$F_2$};
\node at (1.05,1.0) {\footnotesize $-3\lambda_D$};
\node[exblue,anchor=south west] at (0.05,1.9) {\footnotesize $a_v(x) = x$};
\node at (1.2,-0.75) {(b) Differential virtual value on types};
\end{scope}
\end{tikzpicture}
\caption{Example 1 with the core $\{1,2\}$.}
\label{fig:ex0}
\end{figure}

\section{Core-Peripheral Bundling}\label{sec:main}

This section answers the question of \Cref{sec:premium} through two thresholds in the premium: an upper threshold, above which a core-peripheral menu is optimal and every optimal menu is core-peripheral (\Cref{sec:main-core}), and a lower threshold, above which the grand bundle must be offered (\Cref{sec:main-grand}). The upper threshold is finite only if the menu is inclusive (\Cref{sec:main-incl}). \Cref{sec:apps} gives primitive sufficient conditions for the optimality of pure bundling and of two-tier menus. \Cref{sec:main-segments} describes the three regimes in practice.

\paragraph{Assumptions and notation.}

Throughout this section the type space is the box $D = \prod_{i \in [N]} [0, \bar v_i]$. The density $f$ is $C^2$ and strictly positive on $D$, with \emph{radial score} $\Phi(x) := x \cdot \nabla \log f(x) + N + 1$ strictly positive. Valuations take the core-peripheral form \eqref{eq:standing-v} for a fixed core $C$. Each option $\ell \ge 1$ of a menu has a strictly positive price $p_\ell$ and a \emph{marginal-inclusion vector} $w_\ell := \E_{B \sim \sigma_\ell}[\mathbf{1}_B]$. The differential virtual value on types (\Cref{sec:typespace}) splits into its sources --- the origin atom and the top faces $F_j := \{x_j = \bar v_j\}$ --- and its sinks in the interior:
\[
\mu = \mu^+ - \mu^-, \qquad \mu^+ = \delta_0 + \sum_{j} \bar v_j\, f|_{F_j}, \qquad \mu^- = \Phi f \cdot \lambda_D .
\]

\subsection{Upper threshold: optimality and necessity of core-peripheral bundling}\label{sec:main-core}

In this subsection the candidate is a core-peripheral menu --- a menu that sells only the core-containing bundles $\mathcal{F}_C := \{B \in \mathcal{B} : C \subseteq B\}$. The question is when such a menu is optimal against \emph{every} mechanism, including mechanisms that sell bundles outside $\mathcal{F}_C$. \Cref{thm:core-bundle} below gives conditions on the density $f$ and the candidate menu under which, at every large enough premium, the menu is optimal and every optimal menu sells only bundles in $\mathcal{F}_C$. I begin with two examples, which trace the required and the covered directions of a core-peripheral menu as the premium increases.

The two-good examples of this section take $D = [0,1]^2$. Their cost $c_\alpha(h) = \max\{0,\, v_1(h),\, v_2(h),\, v_{\{1,2\}}(h)\}$ and bundle cones (\Cref{sec:bundle-to-cone}) are computed against the full family $\mathcal{B} = \{\{1\}, \{2\}, \{1,2\}\}$. Each figure draws the cells of a candidate menu $M$ with its required directions $h = x_\mathrm{src} - x_\mathrm{snk}$, and the cones at two values of $\alpha$. The question is whether \Cref{prop:saddle-D} certifies $M$ as optimal. A cell and the cone of its option share a color: blue for the grand bundle, green for good $1$, yellow for good $2$ and grey for exclusion. Red marks the required directions that the options of $M$ fail to cover at one or both of the premia drawn. Each arrow runs from a source to a sink, opposite to its required direction $h$.

\begin{egblock}
\paragraph{Pure bundling in the uniform benchmark.}\label{sec:eg-pure}
With a pure-bundling menu and the core $\{1,2\}$, \Cref{eg:uniform} is the simplest case of the upper threshold. The example that ends \Cref{sec:bundle-to-cone} computes its differential virtual value on types. The pure-bundling candidate menu $M$ offers the grand bundle $\{1,2\}$ alone, at price $\alpha\,\hat s$ with $\hat s = \sqrt{2/3}$. Its cells $\Omega_0 = \{x \in D : x_1 + x_2 < \hat s\}$ and $\Omega_{\{1,2\}} = \{x \in D : x_1 + x_2 \ge \hat s\}$ are separated by the anti-diagonal $\{x_1 + x_2 = \hat s\}$.

\paragraph{Required transport.} Conditions (i) and (iii) of \Cref{prop:saddle-D} pair each source with sinks of its own cell. Each cell therefore balances its masses: the origin atom against the sinks of $\Omega_0$, and the two top faces against the sinks of $\Omega_{\{1,2\}}$. This balance is the first-order condition that defines $\hat s$. The origin atom serves $\Omega_0$ along required directions $h \le 0$. A source on $F_1$ that serves a sink below it does so along $h \ge 0$. Near $(1, 0)$ such sinks are too few: below a small height, the narrow strip of $\Omega_{\{1,2\}}$ along $F_1$ holds less sink mass than $F_1$ holds source mass. Some $F_1$-mass therefore serves sinks above it, along fourth-quadrant directions ($h_1 > 0,\ h_2 < 0$). Symmetrically, $F_2$-mass near $(0, 1)$ needs second-quadrant directions. Panel (a) of \Cref{fig:ex1} draws the fourth-quadrant directions in red.

\paragraph{The cones.} At $\alpha = 1$ the bundle's cone is the first quadrant. The oblique required directions (second and fourth quadrants) lie in the singleton cones, outside the menu's coverage. The certificate of \Cref{prop:saddle-D} therefore fails. At $\alpha = 2$ the bundle cone has widened into both mixed-sign quadrants and now covers them (panels (b), (c) of \Cref{fig:ex1}). The dual-transport line construction of \Cref{os:dual-transport} certifies pure bundling once $\alpha \ge \bar\alpha_{\mathrm{unif}} := (9+4\sqrt 6)/15 \approx 1.253$.\footnotemark
\end{egblock}%
\footnotetext{\hspace{.4mm}\label{fn:rays}Recast in this duality, \citet{Yang2023nested} and \citet{haghpanah2021pure} study settings where transport along \emph{rays} suffices. Pure bundling covers such transport at every premium of at least one. By contrast, the uniform square requires transport across rays (the red arrows of \Cref{fig:ex1}(a)), which only a sufficiently high premium covers by widening the grand bundle's cone.}

\begin{figure}[htbp]
\centering
\scalebox{0.85}{\input{figures/examples_v11/ex1}}

\caption{Pure bundling in the uniform benchmark.}
\label{fig:ex1}
\end{figure}

\begin{egblock}
\paragraph{Core-peripheral bundling: a one-good core.}\label{sec:eg-core}
With the core $C = \{1\}$ in \Cref{eg:uniform}, good $2$ is peripheral and earns the premium only together with good $1$. The candidate menu offers the core good and the bundle at prices $\alpha \hat s_1$ and $\alpha \hat s_{\{1,2\}}$ that scale with the premium. Mass from the top face $F_1$ splits between the singleton cell, along fourth-quadrant displacements, and the bundle cell, along first-quadrant ones. Mass from $F_2$ must reach the bundle cell along second-quadrant displacements. At $\alpha = 1$ no option of the menu covers them: this is the obstruction in \Cref{fig:ex2}a. Good $1$ covers the fourth-quadrant displacements at every premium. For $\alpha > 1$ the bundle cone widens into the second quadrant (\Cref{sec:bundle-to-cone}; \Cref{fig:ex2}b,c). Once the widened bundle cone reaches the $F_2$ directions, the obstruction disappears. \Cref{thm:core-bundle,thm:twotier} below show that the upper threshold is finite. For this menu, \Cref{thm:twotier} verifies the regularity condition that \Cref{thm:core-bundle} requires.
\end{egblock}

\begin{figure}[htbp]
\centering
\scalebox{0.85}{\input{figures/examples_v11/ex2}}

\caption{Core-peripheral bundling with a one-good core.}
\label{fig:ex2}
\end{figure}

The examples show the premium widening the covered directions until they contain the required ones. I next introduce the conditions under which this happens for a core-peripheral menu with $N$ goods. \Cref{thm:core-bundle} requires two conditions: inclusivity (on the menu's prices) and \emph{core regularity} (on the menu--density pair, in a weak and a strict version). Both are stated on the menu's cells: option $\ell$ is chosen on a cell $\Omega_\ell \subseteq D$. Prices scale with the premium: option $\ell$ is priced at $\alpha \hat s_\ell$, where the \emph{crossing point} $\hat s_\ell$ does not depend on $\alpha$. The cells $\{\Omega_\ell\}_{\ell \ge 0}$ then do not depend on $\alpha$ either (\Cref{app:core}).

\begin{definition}[Inclusivity]\label{def:incl}
A candidate menu $M$ on a family $\mathcal{F} \subseteq \mathcal{B}$, with exclusion cell $\Omega_0$ (the no-purchase region), is \emph{inclusive} if no top face meets the exclusion cell: $F_j \cap \Omega_0 = \varnothing$ for every $j \in [N]$.
\end{definition}

\noindent Inclusivity is an economic condition: the menu's prices are low enough to serve every buyer who values some good at its maximum. For pure bundling, inclusivity reduces to the single price inequality $\hat s \le \min_i \bar v_i$. In \Cref{eg:uniform} the pure-bundling menu is inclusive: $\hat s = \sqrt{2/3} < 1$.

The second condition requires the sources of each selling cell to dominate its sinks in two statistics of a type: the option's own value $s_\ell(x) := w_\ell \cdot x$ and the peripheral coordinates weighted by the \emph{peripheral weight matrix}
\[
\mathsf{D}_\ell \;:=\; \begin{pmatrix} \operatorname{diag}\bigl(w_\ell|_{P}\bigr) \\[1mm] -\operatorname{diag}\bigl((\mathbf{1} - w_\ell)|_{P}\bigr) \end{pmatrix},
\]
with the diagonal weights taken over the peripheral coordinates. The core goods enter only the first statistic. The matrix $\mathsf{D}_\ell$ weights each peripheral coordinate twice: \emph{upward} by the probability that the option includes the good, and \emph{downward} by the probability that it excludes the good. Dominance in these weighted coordinates puts a source above its sink in each peripheral good that the option always includes, below it in each good that the option never includes, and level with it in each good that the option includes with interior probability. The \emph{option score} $Q_\ell(x) := \bigl(s_\ell(x),\ \mathsf{D}_\ell\, x|_{P}\bigr)$ consists of these statistics, where $x|_{P}$ is the restriction of $x \in \R^N$ to the peripheral coordinates. Let $\nu^\pm_{\ell,0}$ be the pushforward of $\mu^\pm_\ell := \mu^\pm|_{\Omega_\ell}$ --- the cell's sources and sinks --- under $Q_\ell$.

\begin{definition}[Weak core regularity]\label{def:weak-core}
The pair $(f, M)$ is \emph{weakly core-regular} if, at every selling cell $\ell$, $\nu^-_{\ell,0} \preceq_{\mathrm{FOSD}} \nu^+_{\ell,0}$.\footnotemark
\end{definition}

\footnotetext{That is, writing $\nu^\pm_{\ell,0}(A) := \mu^\pm_\ell\{x : Q_\ell(x) \in A\}$: $\int \varphi\, \d\nu^-_{\ell,0} \le \int \varphi\, \d\nu^+_{\ell,0}$ for every bounded Borel coordinatewise nondecreasing $\varphi$. Taking $\varphi$ constant, the dominance implies that the two cell masses are equal.}

\noindent Weak core regularity is the counterpart of Myersonian regularity. With one good, the first-order condition for the optimal price is necessary: the virtual value is zero at the price. Myersonian regularity requires an increasing virtual value. This zero is then the only crossing: the virtual value is negative below the price and positive above it. The first-order condition is therefore also sufficient. With several goods, the dominance implies equal source and sink mass on each selling cell, which is the first-order condition for the price of the cell's option. The dominance itself, in both statistics of the option score, plays the role of the single crossing. With one good, it states that the virtual value is nonnegative above the cutoff type (\cref{fn:one-good-regularity}). By \Cref{lem:restricted-opt}, inclusivity and weak core regularity make the first-order conditions sufficient in the restricted problem.

\begin{lemma}[Restricted optimality]\label{lem:restricted-opt}
Fix a core $C \subseteq [N]$ and let $M = \{(\sigma_\ell, p_\ell)\}_{\ell=1}^L$ be a finite menu of lotteries on $\mathcal{F}_C$ with positive prices $p_\ell = \alpha \hat s_\ell$. If $M$ is \hyperref[def:incl]{inclusive} and $(f, M)$ is \hyperref[def:weak-core]{weakly core-regular}, then $M$ is optimal in the restricted problem on $\mathcal{F}_C$ at every premium.
\end{lemma}

\noindent The proof is in \Cref{app:core}.

Strict core regularity tests the order of weak core regularity against a small amount of \emph{partial unbundling} of the option. For an unbundling probability $\varepsilon \in [0,1]$ and each nonempty $T \subseteq C$, the \emph{unbundling score} $q_{\ell,T}^{\varepsilon}(x) := \varepsilon \sum_{i \in T} x_i + (1-\varepsilon)\, s_\ell(x)$ is the value of the lottery that replaces the cell's own bundle lottery $\sigma_\ell$ with the core subbundle $T$ with probability $\varepsilon$. Replacing the option's value by the unbundled family changes the option score to $Q_\ell^{\varepsilon}(x) := \bigl((q_{\ell,T}^{\varepsilon}(x))_{\varnothing \ne T \subseteq C},\ \mathsf{D}_\ell\, x|_{P}\bigr)$, with $\nu^\pm_{\ell,\varepsilon}$ the pushforward of $\mu^\pm_\ell$ under $Q_\ell^{\varepsilon}$.

\begin{definition}[Strict core regularity]\label{def:align-reg}
A \hyperref[def:weak-core]{weakly core-regular} pair $(f, M)$ is \emph{strictly core-regular} if the dominance holds with a common strictly positive \emph{unbundling slack}: for some $\varepsilon > 0$, $\nu^-_{\ell,\varepsilon} \preceq_{\mathrm{FOSD}} \nu^+_{\ell,\varepsilon}$ at every selling cell $\ell$.\footnotemark
\end{definition}

\footnotetext{Dominance at a positive $\varepsilon$ implies the order at $\varepsilon = 0$ (\Cref{os:align-transfer}). Strict core regularity therefore strengthens weak core regularity.}

\noindent Strictness requires the order to hold when the option is replaced by \emph{any} subbundle of the core with a strictly positive probability that is common to every selling cell and every subbundle. \Cref{sec:apps-pure} relates weak and strict core regularity to Myersonian regularity at $C = [N]$.

\begin{theorem}[Optimality and necessity of core-peripheral bundling]\label{thm:core-bundle}
Fix a core $C \subseteq [N]$ and let $M = \{(\sigma_\ell, p_\ell)\}_{\ell=1}^L$ be a finite menu of lotteries on $\mathcal{F}_C$ with positive prices $p_\ell = \alpha \hat s_\ell$. If $M$ is \emph{\hyperref[def:incl]{inclusive}} and $(f, M)$ is \emph{\hyperref[def:align-reg]{strictly core-regular}}, then there is a finite premium $\alpha^*$ such that, at every premium $\alpha > \alpha^*$, $M$ is optimal in the unrestricted problem and every optimal menu is core-peripheral.
\end{theorem}

\noindent The least such $\alpha^*$ is the \emph{upper threshold} $\alpha^*(C, M)$ of the menu; by \Cref{cor:premium}, it is the infimum of the premia at which $M$ is optimal, and $+\infty$ if there is none.

As the premium increases, the covered directions of the menu widen and its required directions do not change. Weak core regularity balances mass on the option's selling cells. Under inclusivity, no top face meets the exclusion cell. The required transport therefore moves mass only within cells. Weak core regularity yields a plan whose required directions point upward: it matches each source with a sink below it in the option's value and in the weighted peripheral coordinates. The option covers these directions relative to every other bundle that contains the core: along them, changing which peripheral goods are included extracts no more surplus. A bundle that leaves out part of the core can still extract more. In \Cref{fig:ex2}, good $2$ alone does so along the directions from $F_2$ at $\alpha = 1$. Strict core regularity bounds this excess: such a bundle gains at most a fixed multiple of the option's gain. The premium multiplies the option's gain and does not change the gain of such bundles. By separation, the required directions do not depend on the premium. At a large enough premium, the option's covered directions therefore contain every required direction. The proof is in \Cref{app:core}.

\subsection{Lower threshold: necessity of the grand bundle}\label{sec:main-grand}

\Cref{thm:core-bundle} showed that the upper threshold $\alpha^*$ is finite. The next question is when a menu without the grand bundle is suboptimal. I first answer it in the uniform benchmark with separate sales. I then define free disposal for $N$ goods and the lower threshold $\alpha_*$ above which a set of positive measure of types has free disposal. \Cref{thm:grand-bundle-nec} states that above this threshold every optimal menu includes the grand bundle.

\begin{egblock}
\paragraph{Separate sales in the uniform benchmark.}\label{sec:eg-grand}
Consider the separate-sales menu $M$ offering the two singletons $\{1\}, \{2\}$ in \Cref{eg:uniform} with the core $\{1,2\}$. By symmetry the optimal prices satisfy $p_1 = p_2 = p^\star$. The indirect utility is $U(x) = \max\{0,\, x_1 - p^\star,\, x_2 - p^\star\}$. The cells $\Omega_0 = \{x_1 < p^\star,\ x_2 < p^\star\}$, $\Omega_{\{1\}} = \{x_1 \ge p^\star,\ x_1 \ge x_2\}$ and $\Omega_{\{2\}} = \{x_2 \ge p^\star,\ x_2 \ge x_1\}$ are separated along the diagonal segment $E = \{x \in D : x_1 = x_2 \ge p^\star\}$ (\Cref{fig:ex4}a).

\paragraph{Required transport along the diagonal.} The cells $\Omega_{\{1\}}$ and $\Omega_{\{2\}}$ each contain top-face mass (on $F_1$ and $F_2$) that must move into their interiors. By symmetry the transport plan can be taken symmetric across $E$. Near the top corner $(1,1)$ the two faces together hold source mass of order $\delta$ in a $\delta$-neighborhood against sink mass of order $\delta^2$. Most of the corner's source mass must therefore travel an order-one distance, to the sinks around the bottom vertex $(p^\star, p^\star)$ of $E$. The displacement $h = x_{\mathrm{src}} - x_{\mathrm{snk}}$ of such a pair is proportional to $(1,1)$ --- a $+45^\circ$ required direction.

\paragraph{The cones.} The cost at $h = (1,1)$ is $c_\alpha(1,1) = \max\{0,\ 1,\ 1,\ 2\alpha\}$: it equals $1$ for $\alpha \le 1/2$ and $2\alpha$ for $\alpha > 1/2$. For $\alpha \le 1/2$ the singletons $\{1\}$ and $\{2\}$, with gradients $(1,0)$ and $(0,1)$, both cover the $+45^\circ$ direction. It lies on the shared boundary of their cones (panel (b), $\alpha = 0.4$: the bundle cone degenerates to the origin). For $\alpha > 1/2$ only the grand bundle covers it: the $+45^\circ$ direction now lies in the \emph{interior} of the bundle cone and in neither singleton cone (panel (c), $\alpha = 0.6$: the bundle cone is a narrow cone around $+45^\circ$). No bundle that $M$ offers covers the required direction. The certificate of \Cref{prop:saddle-D} therefore fails.
\end{egblock}

\begin{figure}[htbp]
\centering
\scalebox{0.85}{\tikzset{every picture/.style={line width=0.75pt}} 

\begin{tikzpicture}[x=0.75pt,y=0.75pt,yscale=-1,xscale=1]

\draw  [fill={rgb, 255:red, 248; green, 231; blue, 28 }  ,fill opacity=0.36 ] (226.96,55.49) -- (152.6,129.84) -- (46.05,129.84) -- (46.05,55.49) -- cycle ;
\draw  [fill={rgb, 255:red, 126; green, 211; blue, 33 }  ,fill opacity=0.43 ] (226.96,55.49) -- (226.96,236.39) -- (152.6,236.39) -- (152.6,129.84) -- cycle ;
\draw   (46.05,55.49) -- (226.96,55.49) -- (226.96,236.39) -- (46.05,236.39) -- cycle ;
\draw    (98.77,189.02) -- (49.9,230.79) ;
\draw [shift={(100.29,187.72)}, rotate = 139.48] [fill={rgb, 255:red, 0; green, 0; blue, 0 }  ][line width=0.08]  [draw opacity=0] (7.2,-1.8) -- (0,0) -- (7.2,1.8) -- cycle    ;
\draw [color={rgb, 255:red, 0; green, 0; blue, 0 }  ,draw opacity=1 ]   (165.38,231.31) -- (218.6,231.31) ;
\draw [shift={(163.38,231.31)}, rotate = 0] [fill={rgb, 255:red, 0; green, 0; blue, 0 }  ,fill opacity=1 ][line width=0.08]  [draw opacity=0] (7.2,-1.8) -- (0,0) -- (7.2,1.8) -- cycle    ;
\draw [color={rgb, 255:red, 0; green, 0; blue, 0 }  ,draw opacity=1 ]   (165.97,220.7) -- (218.2,213.4) ;
\draw [shift={(163.99,220.97)}, rotate = 352.05] [fill={rgb, 255:red, 0; green, 0; blue, 0 }  ,fill opacity=1 ][line width=0.08]  [draw opacity=0] (7.2,-1.8) -- (0,0) -- (7.2,1.8) -- cycle    ;
\draw [color={rgb, 255:red, 0; green, 0; blue, 0 }  ,draw opacity=1 ]   (165.82,211.18) -- (219,193) ;
\draw [shift={(163.93,211.82)}, rotate = 341.13] [fill={rgb, 255:red, 0; green, 0; blue, 0 }  ,fill opacity=1 ][line width=0.08]  [draw opacity=0] (7.2,-1.8) -- (0,0) -- (7.2,1.8) -- cycle    ;
\draw [color={rgb, 255:red, 0; green, 0; blue, 0 }  ,draw opacity=1 ]   (164.73,195.32) -- (220.2,161.4) ;
\draw [shift={(163.03,196.36)}, rotate = 328.56] [fill={rgb, 255:red, 0; green, 0; blue, 0 }  ,fill opacity=1 ][line width=0.08]  [draw opacity=0] (7.2,-1.8) -- (0,0) -- (7.2,1.8) -- cycle    ;
\draw [color={rgb, 255:red, 208; green, 2; blue, 27 }  ,draw opacity=1 ][line width=1.5]    (159.52,128.57) -- (224.19,63.91) ;
\draw [shift={(156.69,131.4)}, rotate = 315] [fill={rgb, 255:red, 208; green, 2; blue, 27 }  ,fill opacity=1 ][line width=0.08]  [draw opacity=0] (9.36,-2.34) -- (0,0) -- (9.36,2.34) -- cycle    ;
\draw  [color={rgb, 255:red, 248; green, 231; blue, 28 }  ,draw opacity=1 ][fill={rgb, 255:red, 248; green, 231; blue, 28 }  ,fill opacity=1 ] (305.9,76.21) .. controls (305.9,76.04) and (305.9,75.87) .. (305.9,75.7) .. controls (305.9,49.25) and (327.35,27.8) .. (353.8,27.8) .. controls (367.35,27.8) and (379.58,33.43) .. (388.3,42.47) -- (367.4,62.6) .. controls (363.96,59.03) and (359.14,56.82) .. (353.8,56.82) .. controls (343.37,56.82) and (334.92,65.27) .. (334.92,75.7) .. controls (334.92,75.77) and (334.92,75.83) .. (334.92,75.9) -- cycle ;
\draw  [color={rgb, 255:red, 126; green, 211; blue, 33 }  ,draw opacity=1 ][fill={rgb, 255:red, 126; green, 211; blue, 33 }  ,fill opacity=1 ] (387.61,42.02) .. controls (396.31,50.69) and (401.7,62.69) .. (401.7,75.95) .. controls (401.7,102.32) and (380.39,123.72) .. (354.05,123.85) -- (353.9,94.8) .. controls (364.26,94.75) and (372.65,86.33) .. (372.65,75.95) .. controls (372.65,70.73) and (370.53,66.01) .. (367.11,62.6) -- cycle ;
\draw  [color={rgb, 255:red, 155; green, 155; blue, 155 }  ,draw opacity=1 ][fill={rgb, 255:red, 155; green, 155; blue, 155 }  ,fill opacity=1 ] (354.34,123.91) .. controls (354.13,123.91) and (353.91,123.91) .. (353.7,123.91) .. controls (327.25,123.91) and (305.8,102.47) .. (305.8,76.01) .. controls (305.8,76) and (305.8,75.99) .. (305.8,75.99) -- (335.1,76) .. controls (335.1,76) and (335.1,76.01) .. (335.1,76.01) .. controls (335.1,86.28) and (343.43,94.61) .. (353.7,94.61) .. controls (353.78,94.61) and (353.87,94.61) .. (353.95,94.61) -- cycle ;
\draw [color={rgb, 255:red, 208; green, 2; blue, 27 }  ,draw opacity=1 ]   (355.12,74.6) -- (387.91,41.8) ;
\draw [shift={(353.7,76.01)}, rotate = 315] [fill={rgb, 255:red, 208; green, 2; blue, 27 }  ,fill opacity=1 ][line width=0.08]  [draw opacity=0] (7.2,-1.8) -- (0,0) -- (7.2,1.8) -- cycle    ;
\draw  [color={rgb, 255:red, 248; green, 231; blue, 28 }  ,draw opacity=1 ][fill={rgb, 255:red, 248; green, 231; blue, 28 }  ,fill opacity=1 ] (302.9,210.21) .. controls (302.9,210.04) and (302.9,209.87) .. (302.9,209.7) .. controls (302.9,183.25) and (324.35,161.8) .. (350.8,161.8) .. controls (362.38,161.8) and (373.01,165.91) .. (381.29,172.75) -- (362.6,195.4) .. controls (359.39,192.75) and (355.28,191.16) .. (350.8,191.16) .. controls (340.56,191.16) and (332.26,199.46) .. (332.26,209.7) .. controls (332.26,209.77) and (332.26,209.83) .. (332.26,209.9) -- cycle ;
\draw  [color={rgb, 255:red, 0; green, 122; blue, 255 }  ,draw opacity=1 ][fill={rgb, 255:red, 0; green, 122; blue, 255 }  ,fill opacity=1 ] (380.59,171.99) .. controls (383.66,174.4) and (386.42,177.17) .. (388.82,180.24) -- (365.8,198.2) .. controls (364.86,197) and (363.78,195.92) .. (362.58,194.98) -- cycle ;
\draw  [color={rgb, 255:red, 126; green, 211; blue, 33 }  ,draw opacity=1 ][fill={rgb, 255:red, 126; green, 211; blue, 33 }  ,fill opacity=1 ] (388.7,180.65) .. controls (394.97,188.75) and (398.7,198.91) .. (398.7,209.95) .. controls (398.7,236.32) and (377.39,257.72) .. (351.05,257.85) -- (350.9,228.8) .. controls (361.26,228.75) and (369.65,220.33) .. (369.65,209.95) .. controls (369.65,205.61) and (368.18,201.61) .. (365.71,198.42) -- cycle ;
\draw  [color={rgb, 255:red, 155; green, 155; blue, 155 }  ,draw opacity=1 ][fill={rgb, 255:red, 155; green, 155; blue, 155 }  ,fill opacity=1 ] (351.34,257.91) .. controls (351.13,257.91) and (350.91,257.91) .. (350.7,257.91) .. controls (324.25,257.91) and (302.8,236.47) .. (302.8,210.01) .. controls (302.8,210) and (302.8,209.99) .. (302.8,209.99) -- (332.1,210) .. controls (332.1,210) and (332.1,210.01) .. (332.1,210.01) .. controls (332.1,220.28) and (340.43,228.61) .. (350.7,228.61) .. controls (350.78,228.61) and (350.87,228.61) .. (350.95,228.61) -- cycle ;
\draw [color={rgb, 255:red, 208; green, 2; blue, 27 }  ,draw opacity=1 ]   (352.12,208.6) -- (385,175.71) ;
\draw [shift={(350.7,210.01)}, rotate = 315] [fill={rgb, 255:red, 208; green, 2; blue, 27 }  ,fill opacity=1 ][line width=0.08]  [draw opacity=0] (7.2,-1.8) -- (0,0) -- (7.2,1.8) -- cycle    ;
\draw   (46.05,129.84) -- (152.6,129.84) -- (152.6,236.39) -- (46.05,236.39) -- cycle ;

\draw (19.7,243.18) node [anchor=north west][inner sep=0.75pt]  [font=\small,color={rgb, 255:red, 155; green, 155; blue, 155 }  ,opacity=1 ]  {$\mu =+1$};
\draw (239.38,187.61) node [anchor=north west][inner sep=0.75pt]  [font=\small,color={rgb, 255:red, 155; green, 155; blue, 155 }  ,opacity=1 ]  {$\mu =+1$};
\draw (96.1,215.95) node [anchor=north west][inner sep=0.75pt]  [font=\small,color={rgb, 255:red, 155; green, 155; blue, 155 }  ,opacity=1 ]  {$\mu =-3$};
\draw (66.02,66.65) node [anchor=north west][inner sep=0.75pt]  [font=\small] [align=left] {$\displaystyle \Omega_{\{2\}}$};
\draw (183.5,168.5) node [anchor=north west][inner sep=0.75pt]  [font=\small] [align=left] {$\displaystyle \Omega_{\{1\}}$};
\draw (61.96,165.96) node [anchor=north west][inner sep=0.75pt]  [font=\small] [align=left] {$\displaystyle \Omega_{0}$};
\draw (238.76,133.39) node [anchor=north west][inner sep=0.75pt]  [font=\small]  {$F_{1}$};
\draw (107.88,31.73) node [anchor=north west][inner sep=0.75pt]  [font=\small]  {$F_{2}$};
\draw (231.79,242.82) node [anchor=north west][inner sep=0.75pt]  [font=\small]  {$x_{1}$};
\draw (26.97,33.86) node [anchor=north west][inner sep=0.75pt]  [font=\small]  {$x_{2}$};
\draw (167.66,29.69) node [anchor=north west][inner sep=0.75pt]  [font=\small,color={rgb, 255:red, 155; green, 155; blue, 155 }  ,opacity=1 ]  {$\mu =+1$};
\draw (294,131.5) node [anchor=north west][inner sep=0.75pt]  [font=\small] [align=left] {(b) Cones at $\displaystyle \alpha =0.4$};
\draw (293,265.6) node [anchor=north west][inner sep=0.75pt]  [font=\small] [align=left] {(c) Cones at $\displaystyle \alpha =0.6$};
\draw (54.62,268.84) node [anchor=north west][inner sep=0.75pt]  [font=\small] [align=left] {(a) Cells and required directions in $\displaystyle D$};

\end{tikzpicture}}

\caption{Separate sales in the uniform benchmark.}
\label{fig:ex4}
\end{figure}

At every premium above $1/2$, a set of positive measure of buyers around the diagonal has free disposal: each of them values the grand bundle at least as much as every other bundle. Below $1/2$, by contrast, the grand bundle is pointwise dominated, and no interior buyer has free disposal.

I now formally define free disposal for $N$ goods.

\begin{definition}[Free disposal]\label{def:free-disposal}
A type $x$ has \emph{free disposal} if the grand bundle is worth at least as much to $x$ as every other bundle. Let $\mathcal{Z}_\alpha$ denote the set of types in $\operatorname{int} D$ with free disposal.
\end{definition}

\noindent At an interior type, the most valuable bundle other than the grand bundle omits the core good the type values least. Hence
\[
\mathcal{Z}_\alpha = \Bigl\{x \in \operatorname{int} D : \min_{i \in C} x_i \ge (1 - \alpha) \sum_{i \in [N]} x_i\Bigr\}.
\]
Equivalently, an interior type has free disposal if and only if adding goods to any bundle does not reduce its value. The set $\mathcal{Z}_\alpha$ increases with the premium. It has positive measure if and only if the premium exceeds the \emph{lower threshold} $\alpha_* := \inf\{\alpha : \mathcal{Z}_\alpha \ne \varnothing\}$.\footnote{In the two-good benchmark with the core $\{1,2\}$, $\alpha_* = 1/2$. \citet[Example~6]{ChenGui2025} identify the same threshold.} It equals $\operatorname{int} D$ when $\alpha \ge 1$.

\begin{theorem}[Necessity of the grand bundle]\label{thm:grand-bundle-nec}
Suppose a set of positive measure of types has \hyperref[def:free-disposal]{free disposal}, that is, $\alpha > \alpha_*$. Then in every menu that is optimal in the unrestricted problem, the grand bundle, offered as a deterministic option, has a cell of positive measure. Equivalently, a menu in which this option is absent or has a null cell is not optimal in the unrestricted problem.
\end{theorem}

\noindent The lower threshold depends only on the type space and the core. Some interior type strictly prefers the grand bundle to every other bundle if and only if the premium exceeds it. \Cref{thm:grand-bundle-nec} is a version of \citet[Proposition~3]{ChenGui2025}.\footnote{The proof, in \Cref{app:necessity-grand}, is based on the paper's duality method and differs from theirs.}

Suppose the seller adds the grand bundle to the menu at the lowest price at which no buyer strictly prefers it. No buyer's utility changes. Indifference persists along the bundle cone of the grand bundle, which has a nonempty interior above the lower threshold. An indifferent buyer in the interior of the type space therefore lies in a set of positive measure of indifferent buyers, who already take the grand bundle in the original menu. If the grand bundle has a null cell, every indifferent buyer lies on a top face, where the differential virtual value has sources. A small price cut then raises utility on the top faces and in a thin layer below them, where the sources of the differential virtual value exceed the sinks. The dual rules this out: applied to the new utility, an optimal plan bounds the mass of the layer by the change of the utility increment along its pairs, and along a pair that attains the cost the increment at the source is at most the increment at the sink. The proof is in \Cref{app:necessity-grand}.

\subsection{Finiteness of the upper threshold: necessity of inclusivity}\label{sec:main-incl}

Inclusivity is part of the sufficient condition in \Cref{thm:core-bundle}. I next show its necessity: the upper threshold is finite only under this condition. A seller must therefore serve every buyer who values some good at its maximum: a menu that does not serve such buyers is suboptimal at every premium.

\begin{egblock}
\paragraph{Without inclusivity, the upper threshold is infinite.}\label{sec:eg-incl}
Return to the valuation map $v(x) = (x_1, x_2, \alpha(x_1 + x_2))$ of \Cref{eg:uniform} and the pure-bundling menu $M$ of \Cref{sec:eg-pure} (the grand bundle alone). Replace the uniform density, for which $\hat s = \sqrt{2/3} < 1$, by a density heavier near $(1,1)$, with crossing point $\hat s > 1$. The analysis below holds for every density with $\hat s > 1$.

\paragraph{The required direction no cone covers.} When $\hat s > 1$ the menu is not inclusive: the exclusion cell $\Omega_0$ below the anti-diagonal $\{x_1 + x_2 = \hat s\}$ contains the segments $F_1 \cap \Omega_0 = \{1\} \times [0, \hat s-1)$ and $F_2 \cap \Omega_0 = [0, \hat s-1) \times \{1\}$ of the top faces (\Cref{fig:ex3}a). Inside $\Omega_0$ the active option is exclusion, whose cone covers only coordinatewise-nonpositive displacements. The positive top-face mass on $F_i \cap \Omega_0$ therefore cannot be matched to sinks within $\Omega_0$; it must move into $\Omega_{\{1,2\}}$. A source at $(1, \hat s-1) \in F_1 \cap \overline{\Omega_0}$ paired with a sink $x' = (x'_1, x'_2) \in \Omega_{\{1,2\}}$ has the displacement $h = (1-x'_1,\ \hat s-1-x'_2)$, with $h_1 + h_2 = \hat s - (x'_1+x'_2) \le 0$ and equality only on the anti-diagonal boundary. The required transport points along the negative anti-diagonal --- a $-45^\circ$ direction.

\paragraph{Why no premium helps.} At every $\alpha > 0$ the grand bundle's cone lies in the half-plane $\{h : h_1 + h_2 \ge 0\}$: along a displacement with $h_1 + h_2 < 0$, the grand bundle's value $\alpha(h_1+h_2)$ is negative and therefore lies below the cost (panels (b), (c) of \Cref{fig:ex3}, at $\alpha = 2$ and $\alpha = 10$). The required $-45^\circ$ direction is therefore never covered. The upper threshold is infinite ($\alpha^* = +\infty$): by \Cref{thm:incl-necessity} below, pure bundling is suboptimal at every premium.
\end{egblock}

\begin{figure}[htbp]
\centering
\scalebox{0.85}{\tikzset{every picture/.style={line width=0.75pt}} 

\begin{tikzpicture}[x=0.75pt,y=0.75pt,yscale=-1,xscale=1]

\draw  [fill={rgb, 255:red, 0; green, 122; blue, 255 }  ,fill opacity=0.28 ] (249.25,49.65) -- (249.25,196.25) -- (102.5,49.5) -- cycle ;
\draw  [fill={rgb, 255:red, 0; green, 0; blue, 0 }  ,fill opacity=0.12 ] (249.25,230.55) -- (68.35,230.55) -- (68.35,49.65) -- (102.5,49.5) -- (249.25,196.25) -- cycle ;
\draw   (68.35,49.65) -- (249.25,49.65) -- (249.25,230.55) -- (68.35,230.55) -- cycle ;
\draw    (120.67,183.18) -- (71.79,224.95) ;
\draw [shift={(122.19,181.88)}, rotate = 139.48] [fill={rgb, 255:red, 0; green, 0; blue, 0 }  ][line width=0.08]  [draw opacity=0] (7.2,-1.8) -- (0,0) -- (7.2,1.8) -- cycle    ;
\draw [color={rgb, 255:red, 208; green, 2; blue, 27 }  ,draw opacity=1 ][line width=1.5]    (211.59,154.95) -- (247.88,191.07) ;
\draw [shift={(208.75,152.13)}, rotate = 44.86] [fill={rgb, 255:red, 208; green, 2; blue, 27 }  ,fill opacity=1 ][line width=0.08]  [draw opacity=0] (9.36,-2.34) -- (0,0) -- (9.36,2.34) -- cycle    ;
\draw  [color={rgb, 255:red, 248; green, 231; blue, 28 }  ,draw opacity=1 ][fill={rgb, 255:red, 248; green, 231; blue, 28 }  ,fill opacity=1 ] (314,78.37) .. controls (314,78.2) and (314,78.03) .. (314,77.86) .. controls (314,58.86) and (325.06,42.45) .. (341.09,34.7) -- (353.77,61) .. controls (347.5,64.02) and (343.18,70.44) .. (343.18,77.86) .. controls (343.18,77.93) and (343.18,77.99) .. (343.18,78.06) -- cycle ;
\draw  [color={rgb, 255:red, 0; green, 122; blue, 255 }  ,draw opacity=1 ][fill={rgb, 255:red, 0; green, 122; blue, 255 }  ,fill opacity=1 ] (341.82,34.48) .. controls (347.99,31.58) and (354.88,29.96) .. (362.15,29.96) .. controls (388.6,29.96) and (410.05,51.41) .. (410.05,77.86) .. controls (410.05,85.42) and (408.3,92.57) .. (405.18,98.92) -- (378.69,85.96) .. controls (379.89,83.51) and (380.56,80.77) .. (380.56,77.86) .. controls (380.56,67.69) and (372.32,59.45) .. (362.15,59.45) .. controls (359.35,59.45) and (356.7,60.07) .. (354.33,61.18) -- cycle ;
\draw  [color={rgb, 255:red, 126; green, 211; blue, 33 }  ,draw opacity=1 ][fill={rgb, 255:red, 126; green, 211; blue, 33 }  ,fill opacity=1 ] (405.43,98.12) .. controls (397.89,114.51) and (381.35,125.91) .. (362.15,126.01) -- (362,96.96) .. controls (369.55,96.92) and (376.06,92.43) .. (379.03,85.98) -- cycle ;
\draw  [color={rgb, 255:red, 155; green, 155; blue, 155 }  ,draw opacity=1 ][fill={rgb, 255:red, 155; green, 155; blue, 155 }  ,fill opacity=1 ] (362.43,126.07) .. controls (362.22,126.07) and (362.01,126.07) .. (361.8,126.07) .. controls (335.34,126.07) and (313.9,104.62) .. (313.9,78.17) .. controls (313.9,78.16) and (313.9,78.15) .. (313.9,78.15) -- (343.2,78.16) .. controls (343.2,78.16) and (343.2,78.17) .. (343.2,78.17) .. controls (343.2,88.44) and (351.52,96.77) .. (361.8,96.77) .. controls (361.88,96.77) and (361.96,96.77) .. (362.04,96.77) -- cycle ;
\draw [color={rgb, 255:red, 208; green, 2; blue, 27 }  ,draw opacity=1 ]   (363.21,79.58) -- (395.8,112.17) ;
\draw [shift={(361.8,78.17)}, rotate = 45] [fill={rgb, 255:red, 208; green, 2; blue, 27 }  ,fill opacity=1 ][line width=0.08]  [draw opacity=0] (7.2,-1.8) -- (0,0) -- (7.2,1.8) -- cycle    ;
\draw [color={rgb, 255:red, 208; green, 2; blue, 27 }  ,draw opacity=1 ]   (360.38,76.76) -- (327.83,44.2) ;
\draw [shift={(361.8,78.17)}, rotate = 225] [fill={rgb, 255:red, 208; green, 2; blue, 27 }  ,fill opacity=1 ][line width=0.08]  [draw opacity=0] (7.2,-1.8) -- (0,0) -- (7.2,1.8) -- cycle    ;
\draw [color={rgb, 255:red, 208; green, 2; blue, 27 }  ,draw opacity=1 ][line width=1.5]    (146.15,88.56) -- (109,51.88) ;
\draw [shift={(149,91.38)}, rotate = 224.64] [fill={rgb, 255:red, 208; green, 2; blue, 27 }  ,fill opacity=1 ][line width=0.08]  [draw opacity=0] (9.36,-2.34) -- (0,0) -- (9.36,2.34) -- cycle    ;
\draw  [color={rgb, 255:red, 248; green, 231; blue, 28 }  ,draw opacity=1 ][fill={rgb, 255:red, 248; green, 231; blue, 28 }  ,fill opacity=1 ] (314.4,204.37) .. controls (314.4,204.2) and (314.4,204.03) .. (314.4,203.86) .. controls (314.4,184.86) and (325.46,168.45) .. (341.49,160.7) -- (354.17,187) .. controls (347.9,190.02) and (343.58,196.44) .. (343.58,203.86) .. controls (343.58,203.93) and (343.58,203.99) .. (343.58,204.06) -- cycle ;
\draw  [color={rgb, 255:red, 0; green, 122; blue, 255 }  ,draw opacity=1 ][fill={rgb, 255:red, 0; green, 122; blue, 255 }  ,fill opacity=1 ] (331.16,167.68) .. controls (339.57,160.38) and (350.54,155.96) .. (362.55,155.96) .. controls (389,155.96) and (410.45,177.41) .. (410.45,203.86) .. controls (410.45,215.97) and (405.95,227.03) .. (398.54,235.47) -- (376.6,216.2) .. controls (379.49,212.91) and (381.25,208.59) .. (381.25,203.86) .. controls (381.25,193.53) and (372.88,185.16) .. (362.55,185.16) .. controls (357.86,185.16) and (353.57,186.88) .. (350.29,189.73) -- cycle ;
\draw  [color={rgb, 255:red, 126; green, 211; blue, 33 }  ,draw opacity=1 ][fill={rgb, 255:red, 126; green, 211; blue, 33 }  ,fill opacity=1 ] (398.16,235.86) .. controls (389.44,245.71) and (376.72,251.94) .. (362.55,252.01) -- (362.4,222.96) .. controls (367.97,222.93) and (372.98,220.48) .. (376.41,216.61) -- cycle ;
\draw  [color={rgb, 255:red, 155; green, 155; blue, 155 }  ,draw opacity=1 ][fill={rgb, 255:red, 155; green, 155; blue, 155 }  ,fill opacity=1 ] (362.83,252.07) .. controls (362.62,252.07) and (362.41,252.07) .. (362.2,252.07) .. controls (335.74,252.07) and (314.3,230.62) .. (314.3,204.17) .. controls (314.3,204.16) and (314.3,204.15) .. (314.3,204.15) -- (343.6,204.16) .. controls (343.6,204.16) and (343.6,204.17) .. (343.6,204.17) .. controls (343.6,214.44) and (351.92,222.77) .. (362.2,222.77) .. controls (362.28,222.77) and (362.36,222.77) .. (362.44,222.77) -- cycle ;
\draw [color={rgb, 255:red, 208; green, 2; blue, 27 }  ,draw opacity=1 ]   (363.61,205.58) -- (396.2,238.17) ;
\draw [shift={(362.2,204.17)}, rotate = 45] [fill={rgb, 255:red, 208; green, 2; blue, 27 }  ,fill opacity=1 ][line width=0.08]  [draw opacity=0] (7.2,-1.8) -- (0,0) -- (7.2,1.8) -- cycle    ;
\draw [color={rgb, 255:red, 208; green, 2; blue, 27 }  ,draw opacity=1 ]   (360.78,202.76) -- (328.23,170.2) ;
\draw [shift={(362.2,204.17)}, rotate = 225] [fill={rgb, 255:red, 208; green, 2; blue, 27 }  ,fill opacity=1 ][line width=0.08]  [draw opacity=0] (7.2,-1.8) -- (0,0) -- (7.2,1.8) -- cycle    ;

\draw (137.32,56.81) node [anchor=north west][inner sep=0.75pt]  [font=\small] [align=left] {$\displaystyle \Omega_{\{1,2\}}$};
\draw (86.66,151.32) node [anchor=north west][inner sep=0.75pt]  [font=\small] [align=left] {$\displaystyle \Omega_{0}$};
\draw (257.46,127.15) node [anchor=north west][inner sep=0.75pt]  [font=\small]  {$F_{1}$};
\draw (130.18,25.89) node [anchor=north west][inner sep=0.75pt]  [font=\small]  {$F_{2}$};
\draw (251.25,233.95) node [anchor=north west][inner sep=0.75pt]  [font=\small]  {$x_{1}$};
\draw (49.27,28.02) node [anchor=north west][inner sep=0.75pt]  [font=\small]  {$x_{2}$};
\draw (310.1,134.16) node [anchor=north west][inner sep=0.75pt]  [font=\small] [align=left] {(b) Cones at $\displaystyle \alpha =2$};
\draw (75.92,262) node [anchor=north west][inner sep=0.75pt]  [font=\small] [align=left] {(a) Cells and required directions in $\displaystyle D$};
\draw (310.5,260.16) node [anchor=north west][inner sep=0.75pt]  [font=\small] [align=left] {(c) Cones at $\displaystyle \alpha =10$};

\end{tikzpicture}}

\caption{Pure bundling without inclusivity.}
\label{fig:ex3}
\end{figure}

\Cref{thm:incl-necessity} below extends the example to $N$ goods and to candidate menus on any family of bundles.

\begin{theorem}[Necessity of inclusivity]\label{thm:incl-necessity}
Let $N \ge 2$, let $M_{\mathcal{F}}$ be any finite candidate menu on a restricted family $\mathcal{F} \subseteq \mathcal{B}$, and let $f \in C^2(D)$ be strictly positive on $D$. If $M_{\mathcal{F}}$ is not \hyperref[def:incl]{inclusive}, then it is not optimal in the unrestricted problem, at any $\alpha > 0$.
\end{theorem}

\noindent The buyers who value good $i$ at its maximum are sources of the differential virtual value, and an optimal menu screens them from the buyers below them in good $i$. A menu that excludes some of them does not: exclusion gives these sources and the excluded sinks below them the same utility, zero, and no premium on the bundles of the menu changes that. Selling good $i$ alone at a small discount below its maximum value screens them. It raises utility only near the top face of good $i$, and by at most the discount. Weighted by the differential virtual value, the gain at the excluded sources on the face is of the order of the discount. The loss at the sinks of the thin layer below the face is of the order of its square. For a small discount the gain exceeds the loss at every premium. The proof is in \Cref{app:necessity-incl}.

\subsection{Pure bundling and two-tier menus}\label{sec:apps}

\Cref{thm:core-bundle} assumes strict core regularity --- a condition on the menu and the distribution jointly. This subsection derives sufficient conditions for it from primitives in two cases: pure bundling (\Cref{sec:apps-pure}) and a two-tier menu with one peripheral good, which offers the core alone and the core with one add-on (\Cref{sec:apps-twotier}). Both derivations use Myersonian regularity of the value of the core (the sum of the core coordinates), which \Cref{os:single-cell} converts into the dominance of the sources of a cell over its sinks in the unbundling scores.

\subsubsection{Pure bundling}\label{sec:apps-pure}

For the seller, the question is when one price for the grand bundle is optimal against every menu. At $C = [N]$ the menu offers the grand bundle alone at the price $p^* = \alpha \hat s$ --- the restricted optimum on $\mathcal{F}_{[N]} = \{[N]\}$, with $\alpha$-independent crossing point $\hat s$. Call the bundle-value distribution of $\sum_i x_i$ \emph{strictly Myersonian regular} if its virtual value is strictly positive above the crossing point, $H(r) := r\,\gamma(r)/\bar F(r) > 1$ on the active range ($\gamma, \bar F$ the bundle value's density and survivor). A strictly increasing virtual value together with the price first-order condition implies it. Call the pure-bundling menu \emph{strictly inclusive} if $\hat s < \min_i \bar v_i$. A menu $M$ is \emph{uniquely optimal} if it is optimal and every optimal menu gives almost every type the lottery and the price that $M$ gives it.

\begin{proposition}[Optimality and uniqueness of pure bundling]\label{thm:pure-bundle}
Suppose the bundle-value distribution is strictly Myersonian regular. Then pure bundling at $p^*$ is optimal in the unrestricted problem at some $\alpha > 0$ if and only if it is strictly inclusive, and it is then uniquely optimal at every premium above its upper threshold.
\end{proposition}

\noindent Sufficiency is \Cref{thm:core-bundle} at $C = [N]$: strict Myersonian regularity of the bundle value implies strict core regularity (\Cref{os:single-cell}).\footnote{Weak core regularity of the strictly inclusive pure-bundling menu is equivalent to weak Myersonian regularity on the active range: $H \ge 1$ above the crossing point (\Cref{os:apps-puresec}). Strict core regularity at $C = [N]$ is intermediate between the two conditions: $H > 1$ implies it (\Cref{os:single-cell}), and it implies weak core regularity, hence $H \ge 1$.} Necessity holds because, when $\hat s \ge \bar v_i$ for some good $i$, selling good $i$ alone at a price slightly below $\bar v_i$ raises revenue. If $\hat s > \bar v_i$, the menu is not inclusive, and this is the argument of \Cref{thm:incl-necessity}.\footnote{With additive values ($\alpha = 1$), \citet[Proposition 6]{ChenGui2026} identify the necessity of inclusivity for pure bundling.} If $\hat s = \bar v_i$, the menu is inclusive, and for a small discount\footnote{Yi-Chun Chen suggested this perturbation.} the gain on the top face $F_i$ exceeds the loss at the sinks below it.\footnote{The necessity part of \Cref{thm:pure-bundle} holds for every density that satisfies the standing assumptions of \Cref{sec:main}: if $\hat s \ge \min_i \bar v_i$, pure bundling at $p^*$ is strictly suboptimal at every premium.} Uniqueness also follows from \Cref{thm:core-bundle}: every optimal menu sells only the grand bundle, and $p^*$ is its only optimal price. The proof is in \Cref{os:apps-puresec}.

\paragraph{The iid Beta family.} The crossing point of \Cref{thm:pure-bundle} has a closed form for the iid Beta family. Let $f_\beta(x) = \beta^N \prod_i x_i^{\beta-1}$ on $[0,1]^N$ (iid $\mathrm{Beta}(\beta,1)$). The bundle value $\sum_i x_i$ has the distribution function $G(t) = t^{N\beta}/c_N(\beta)$ on $[0,1]$, with $c_N(\beta) = \Gamma(1+N\beta)/\Gamma(1+\beta)^N$. Writing $A_N(\beta) := (1+N\beta)/c_N(\beta)$, the interior first-order condition for the pure-bundling price has the root $\hat s_{N,\beta} = A_N(\beta)^{-1/(N\beta)}$. This root lies strictly below $1$ if and only if $A_N(\beta) > 1$. The bundle value's virtual value is strictly positive on $(\hat s_{N,\beta}, 1]$. When the bundle-value distribution is strictly Myersonian regular, this root is the crossing point and the menu is strictly inclusive if and only if $A_N(\beta) > 1$.\footnote{\Cref{eg:uniform} is the $N=2$, $\beta=1$ member, whose bundle-value distribution is strictly Myersonian regular. In this benchmark, separate sales is suboptimal above $1/2$ (\Cref{sec:main-grand}), and pure bundling is optimal above $\bar\alpha_{\mathrm{unif}} \approx 1.253$ (\Cref{sec:main-core}). At $\alpha = 1$ the optimum is the mixed-bundling menu of \citet{Pavlov2011}. \citet[Example~8]{ChenGui2026} study the optimal menu in between. Their simulations find the switch to pure bundling near $1.21$.}

\subsubsection{Two-tier menus with one peripheral good}\label{sec:apps-twotier}

The second case is a two-tier menu: it offers the bare core $C$ and the grand bundle $C \cup \{p\}$, which adds a single peripheral good. It is the simplest core-peripheral menu that sells an add-on. For the seller, the question is how to price the core when the add-on is sold together with it. Normalize to $D = [0,1]^N$. The core is $C$ with $|C| = N - 1 \ge 1$ and $P = \{p\}$. Types are $x = (\omega, z)$, with $z := x_p$ and core sum $t(\omega) := \sum_{i \in C} \omega_i$.

The \emph{two-tier setting} consists of three conditions. \textnormal{(S1)} The density is a product $f = f_C \otimes f_p$ with $f_C = \bigotimes_{i \in C} f_i$, each $f_i$ log-concave, $C^1$, and positive on $(0,1]$. The core radial score $\Phi_C(\omega) := |C| + \omega \cdot \nabla \log f_C(\omega)$ is nonnegative. \textnormal{(S2)} The peripheral marginal is $f_p(z) = \beta z^{\beta - 1}$ for some $\beta > 0$. \textnormal{(S3)} The menu $M = \{(C, \alpha \hat s_C),\ (C \cup \{p\}, \alpha \hat s_G)\}$ has crossing points $0 < \hat s_C < \hat s_G < 1$ and \emph{mass-balanced} selling cells: each selling cell has equal source and sink mass. Write $d := \hat s_G - \hat s_C$ and $\tau(z) := (\hat s_G - z)^+$. By (S1), $f_C$ satisfies the conditions that \Cref{os:single-cell} imposes on the core density.\footnote{\Cref{thm:twotier} below extends to densities outside the standing assumptions of \Cref{sec:main}, such as $f$ with $\beta \ne 1$ or with some $f_i$ vanishing at $0$ (\Cref{os:sc-extension}).}

In utilities normalized by $\alpha$, the bare-core option is chosen on $\Omega_1 = \{t \ge \hat s_C,\ z \le d\}$ and the grand-bundle option on $\Omega_2 = \{z \ge d,\ t \ge \tau(z)\}$. The upgrade decision depends on $z$ alone: a buyer adds the peripheral good when $z \ge d$. For $z \ge d$, the participation boundary $t = \tau(z)$ decreases one-for-one in $z$, to zero at $z = \hat s_G$. Inclusivity follows from the crossing points of (S3): every core top face has $t \ge 1 > \hat s_G \ge \tau(z)$, and $\hat s_G < 1$ implies that the peripheral face $\{z = 1\}$ lies in $\Omega_2$. Let $H_C(q) := q\, g_C(q)/\bar G_C(q)$ be the core sum's virtual-value statistic (\Cref{os:single-cell}), where $g_C$ and $\bar G_C$ are the density and the survivor of the core sum under the core density $f_C$ of (S1). Mass balance on the bare-core cell is equivalent to $H_C(\hat s_C) = 1 + \beta$ (\Cref{os:apps-twotier}). This is the two-tier analogue of the pure-bundling first-order condition $H(\hat s) = 1$ (\Cref{os:apps-puresec}), with the constant raised from $1$ to $1 + \beta$.

\begin{proposition}[Optimality and uniqueness of the two-tier menu]\label{thm:twotier}
Under the two-tier setting \textnormal{(S1)--(S3)}, the pair $(f, M)$ is strictly core-regular. Hence the two-tier menu has a finite upper threshold and is uniquely optimal at every premium above it.
\end{proposition}

\noindent In other words, at a large premium a seller with one add-on offers the bare core and the core with the add-on, at prices that scale with the premium. Every optimal menu sells only these two bundles, to the same buyers at the same prices.\footnote{For one-dimensional types, \citet[Theorem~1(iii)]{Yang2023nested} proves the analogous uniqueness under the nesting condition and the other assumptions of that theorem: every optimal mechanism agrees at almost every type with nested bundling, of which two-tier menus are a case.} The condition $H_C(\hat s_C) = 1 + \beta$ sets the price of the core: it is the first-order condition for the price of the bare core, with the price of the grand bundle fixed (\Cref{lem:mb-foc}).

A higher core price has two effects. The first is familiar from the core sold alone: a higher price earns more from the buyers of the bare core and loses its marginal buyers. The second is the upgrade margin: with the price of the grand bundle fixed, a higher core price lowers the price of the upgrade. The buyers whose value of the add-on is near $d$ then upgrade. With the peripheral distribution of (S2), this upgrade margin adds $\beta$ to the constant. The proof is in \Cref{os:apps-twotier}.

\subsection{Three regimes in practice}\label{sec:main-segments}

The two thresholds of \Cref{sec:main-core,sec:main-grand} divide the premia into three regimes. For a seller who serves segments whose premia differ, the regimes illustrate how the menu changes from one segment to the next.

\paragraph{Segments in practice.} AI providers keep two price lists for the same models. Consumers, who value the unified interface and the shared memory of one subscription to every model, buy a tiered bundle. Developers, whose own code performs the routing and the context management, buy per-token access to each model separately.\footnote{Compare \url{https://openai.com/chatgpt/pricing} with \url{https://openai.com/api/pricing/}. Anthropic's consumer and API channels show the same split.} Enterprise software shows the same segmentation: large traditional firms license integrated suites, tech firms with large engineering capacity assemble best-of-breed components, and mid-sized enterprises license a core platform with selected modules, some of them also sold on their own. Each segment's premium lies in one of the three regimes: per-token access and best-of-breed assembly lie below $\alpha_*$, mid-sized licensing lies between the thresholds, and consumer tiers and integrated suites lie above $\alpha^*$.

\paragraph{Measured premia.} The premium $\alpha$ has been measured directly, both below and above one. \citet{CrossWebbODoherty2026} elicit willingness to pay for ordinary consumer goods and for every pair among them. Bundle values scale with the sum of the goods' values at a rate of $0.73$. \citet{PopkowskiHaubl2010} measure the ratio of a stamp collection's valuation to the sum of its parts' valuations across $48$ collections. The ratio averages $1.19$. Both premia lie above the lower threshold of two-good bundles, $\alpha_* = 1/2$. In the benchmark, then, separate sales is suboptimal even for goods whose bundle is worth less than its parts (\Cref{thm:grand-bundle-nec}). The stamp collections' premium of $1.19$ lies above the additive point $\alpha = 1$, where the benchmark optimum is mixed bundling, and below the bound above which pure bundling is certified (\Cref{os:dual-transport}).

\section{Discussion}\label{sec:discussion}

\Cref{sec:main} assumes the core-peripheral form, whose valuation maps are linear. This section explains why the premium family takes the core-peripheral form and extends the preference-based analysis to valuation maps that are not tangent.

\subsection{Why core-peripheral}\label{sec:main-algebra}

By \Cref{thm:core-bundle}, the core-peripheral menu has a finite upper threshold $\alpha^*$, above which it is optimal. The structure shows an asymmetry between large and small bundles. The grand bundle enters the menu while its cone is still narrow: above the lower threshold, which is less than one, every optimal menu sells it to a set of buyers of positive measure (\Cref{thm:grand-bundle-nec}). By contrast, above $\alpha^*$, no optimal menu sells a singleton $\{i\} \ne C$ to a set of buyers of positive measure, although its cone is still wide (\Cref{thm:core-bundle}): along the comparisons between types that the dual uses, the premium makes the bundle that each top-face source buys strictly better than the singleton (\Cref{cor:premium}).

The proof of \Cref{thm:core-bundle} reveals an algebraic structure behind bundling. The premium $\alpha$ multiplies the gradients of the bundles in the premium family --- those that contain the core. The cones computed within the premium family are then $\alpha$-invariant: the gradients of its bundles scale together. Relative to the other bundles, the cones of bundles in the premium family widen as $\alpha$ grows. Closure of the premium family under supersets keeps bundle values monotone: every superset of a bundle in the family also earns the premium and is worth at least as much. The family $\mathcal{F}_C$ of core-containing bundles is the simplest superset-closed family of bundles. Every nonempty superset-closed family of bundles contains $[N]$. This fact is the algebraic counterpart of the necessity of the grand bundle (\Cref{thm:grand-bundle-nec}).

The same algebra extends beyond the core-peripheral form. Let several disjoint cores share one premium, earned by every bundle that contains at least one of them. The premium family is again closed under supersets. The gradients of its bundles scale together, which is the invariance that the proof uses. Hence the certificate of \Cref{prop:saddle} applies as stated. The construction of the plan, by contrast, needs an analogue of core regularity.

\subsection{Curvature without tangency}\label{sec:apps-curvature}

For linear valuation maps, the analysis of \Cref{sec:bundle-to-cone,sec:main} studies how bundling depends on the preference with the distribution held fixed. This subsection extends it to valuation maps that are not tangent and asks, as in \Cref{sec:main}, whether some premium makes pure bundling optimal. With one-dimensional types, the valuation manifold is a curve, which I call the \emph{valuation curve}. Only the grand bundle earns the premium.

Two kinds of types determine the answer. The \emph{critical type} is the marginal buyer of pure bundling at its optimal price. A \emph{tangent type} is a buyer at which the valuation curve is tangent: its tangent line there passes through the bottom valuation vector, here the origin. Neither depends on the premium. With regular values, pure bundling is optimal at some premium if and only if the critical type is a tangent type; it is then optimal at every premium of at least one (\Cref{prop:curvature-1d}). The proof, in \Cref{os:curvature} of the Online Supplement, uses the flow duality with the Myerson field of \Cref{sec:typespace}.

\begin{figure}[tbp]
\centering
\input{figures/examples_v11/ex5_tangency}
\caption{Pure bundling with one-dimensional types.}
\label{fig:ex5}
\end{figure}

\begin{egblock}
\paragraph{A curve tangent nowhere.} If the valuation curve is tangent at no buyer, pure bundling is suboptimal at every premium and for every distribution (\Cref{prop:curvature-1d}(i)). \Cref{fig:ex5}(a) draws such a curve: the value of good $2$ grows faster in the type than the value of good $1$. Its tangent line at the critical type (dashed) does not pass through the origin. Red marks the flow directions that pure bundling does not cover. As in the figures of \Cref{sec:main}, each arrow runs from a source to a sink, opposite to its flow direction. At the critical type the virtual values of the two goods sum to zero without both vanishing. Good $1$ still has a positive virtual value there. Its \emph{cutoff} --- the marginal buyer of good $1$ sold alone at its best price --- therefore lies below the critical type. The flow directions of the buyers between the cutoff of good $1$ and the critical type lie in the cone of good $1$ and outside the grand bundle's cone at every premium (\Cref{fig:ex5}(c)). Pure bundling excludes these buyers and covers none of their directions. Selling good $1$ alone to them gains an amount that does not decrease in the premium. Non-tangency at the critical type is the one-dimensional counterpart of the failure of inclusivity (\Cref{thm:incl-necessity}): in both, a good sold alone to buyers outside the bundle's market gains at every premium.

\paragraph{Tangency at the critical type.} \Cref{fig:ex5}(b) draws a curve that is tangent at one buyer only, under a distribution whose critical type is that buyer. The tangent line there passes through the origin. Pure bundling therefore covers every flow direction. A distribution whose critical type differs from this buyer makes pure bundling suboptimal (\Cref{prop:curvature-1d}(i)), with uncovered directions as in panel (a).
\end{egblock}

The tangency condition unifies earlier results. With one-dimensional types, \citet{Ghili2023} and \citet{Yang2023nested} characterize the optimality of pure bundling by comparing the quantities that each bundle sells alone. In the premium family this comparison reduces to tangency at the critical type. The ratio condition of \citet{haghpanah2021pure} holds in this family only when the curve is tangent everywhere. That is, it holds only when the curve lies on a ray. The worst-case distribution for pure bundling in \citet{CheZhong2025} is supported on a partly tangent curve, which corresponds to a latent one-dimensional type mapped into values.

\section{Conclusion}\label{sec:conclusion}

This paper develops a method for screening with combinatorial preferences: a duality followed by two reductions. The maximal revenue equals the minimum cost of a flow of the differential virtual value across valuation vectors (\Cref{thm:flow-duality}). The first reduction represents the differential virtual value on types. This representation exists if and only if tangency holds (\Cref{prop:tangency-char}). On a convex valuation manifold, the flow then reduces to transport between types (\Cref{thm:strong-duality}). The second reduction makes the differential virtual value on types independent of the preference. For a family of preferences that includes the additive one, this holds if and only if the family's valuation maps are linear (\Cref{prop:separation}). Under this separation, the distribution determines the required directions and the preference determines the covered directions of each bundle. A menu is optimal when the bundle each buyer chooses exercises screening power along the required directions through that buyer (\Cref{prop:saddle}). Raising a common premium on the bundles that an optimal menu sells keeps the menu optimal at proportionally scaled prices. When every good can also be sold alone, every optimal menu at a higher premium sells only these bundles (\Cref{cor:premium}).

The question of the paper is how the optimal menu changes as goods become stronger complements. A premium on every bundle that contains a fixed core widens the covered directions of these bundles, while the required directions stay fixed: by the separation, the differential virtual value on types does not change. Two thresholds in the premium describe the optimum: the grand bundle must be offered above the lower threshold (\Cref{thm:grand-bundle-nec}), and above the upper threshold a core-peripheral menu is optimal and every optimal menu is core-peripheral. Under a regularity condition, the upper threshold of an inclusive core-peripheral menu is finite (\Cref{thm:core-bundle}). Inclusivity is also necessary: no premium makes a menu optimal if it does not serve buyers who value some good at its maximum (\Cref{thm:incl-necessity}). The regularity condition follows from primitives for pure bundling and for two-tier menus (\Cref{thm:pure-bundle,thm:twotier}). \Cref{sec:main-algebra} explains why the premium family takes the core-peripheral form. With one-dimensional types, pure bundling can be optimal only if the valuation curve is tangent at the critical type. A curve tangent nowhere makes pure bundling suboptimal at every premium and for every distribution (\Cref{sec:apps-curvature}).

{\small
\renewcommand{\baselinestretch}{0.95}\selectfont
\setlength{\bibsep}{2pt plus 0.5ex}
\nocite{Strassen1965,Kellerer1984}
\bibliography{references_v11}
}

\appendix
\count\footins=1000 \dimen\footins=8in 
 \titlespacing\section{0pt}{4pt}{2pt}
\titlespacing\subsection{0pt}{2pt}{1pt}
\titlespacing\subsubsection{0pt}{1pt}{1pt}
\titlespacing*{\paragraph}{0pt}{.8ex plus .4ex minus .2ex}{0.5em}
\renewcommand\LayoutDisplaySkips{%
  \setlength{\abovedisplayskip}{6pt plus 2pt minus 3pt}%
  \setlength{\belowdisplayskip}{6pt plus 2pt minus 3pt}%
  \setlength{\abovedisplayshortskip}{0pt plus 2pt}%
  \setlength{\belowdisplayshortskip}{4pt plus 2pt minus 2pt}}
\normalsize

\section{Proofs for the method}\label{app:duality}

\subsection{Proof of \texorpdfstring{\Cref{thm:flow-duality}}{the flow-duality theorem}}\label{app:duality-flow}

I identify an option with the pair $(q,p)$ of its lottery vector and price; a menu is then a compact subset of $\mathcal{Q}\times[0,\bar p]$ containing the outside option $(0,0)$. For a menu $M$ write
\begin{equation}\label{eq:flow-menu}
u_M(y):=\max_{(q,p)\in M} \{y\cdot q-p\},
\qquad
p_M(y):=\max\bigl\{p:(q,p)\in M,\ y\cdot q-p=u_M(y)\bigr\},
\end{equation}
$R(M):=\int_X p_M\,\d \Pi$ and $R^*:=\sup_M R(M)$. An option is \emph{active} at $y$ if it attains $u_M(y)$, and \emph{selected} if it is active with price $p_M(y)$. Without loss of optimality, $u_M(y_0)=0$: raising every price by $k:=u_M(y_0)$, dropping the options priced above $\bar p$, which no type then takes, and adjoining $(0,0)$ lowers the utility by $k$ on $X$, as $u_M\ge k$ by monotonicity, and raises revenue by $k$.

\begin{proof}[Proof of \Cref{thm:flow-duality}]
The proof has three steps: primal attainment by compactness, no gap between menus and smooth utilities by mollification, and no duality gap, the key step.

\emph{Step 1: primal attainment.} By Blaschke's selection theorem, the set $\mathcal{K}$ of menus with the Hausdorff metric is compact. The payment $(M,y)\mapsto p_M(y)$ is upper semicontinuous on $\mathcal{K}\times X$: let $(M_n,y_n)\to(M,y)$ and let $(q_n,p_n)$ be selected at $y_n$; along a subsequence realizing $\limsup_n p_{M_n}(y_n)$, the options converge to some $(q,p)$. Hausdorff convergence gives $(q,p)\in M$ and uniform convergence $u_{M_n}\to u_M$ on $X$, so $y\cdot q-p=\lim_n\,(y_n\cdot q_n-p_n)=\lim_n u_{M_n}(y_n)=u_M(y)$, so $(q,p)$ is active at $y$ and $\limsup_n p_{M_n}(y_n)=p\le p_M(y)$. For a maximizing sequence with subsequential limit $M$, the reverse Fatou lemma, with payments in $[0,\bar p]$, gives $R^*\le\int_X \limsup_n p_{M_n}\,\d \Pi\le R(M)$, so $M$ is an optimal menu.

\emph{Step 2: no gap between menus and smooth utilities.} The equality of the maximal revenue over menus and the value of \eqref{eq:flow-primal} is
\begin{equation}\label{eq:smooth-primal}
R^*=\sup_{u\in\mathcal{U}^1} \langle T_\Pi,u \rangle .
\end{equation}
Every $u\in\mathcal{U}^1$ defines a direct mechanism on $X$ with allocation $\nabla u(y)$ and payment $y\cdot\nabla u(y)-u(y)$; it is incentive compatible because $u$ is convex, individually rational because $u\ge u(y_0)=0$, and earns $\langle T_\Pi,u \rangle$ by \eqref{eq:TP}. The utility of such a mechanism is the restriction of a convex function with subgradients in $\mathcal{Q}$ (the characterization of \Cref{sec:menus}, \citealp{RochetChone1998}). By the taxation principle, the menu of its options with nonnegative payments earns at least as much under seller-favorable ties; hence $\langle T_\Pi,u \rangle\le R^*$.

For the reverse bound, fix a menu $M$ and $0<\eta<1$, and let $u_\eta$, defined by \eqref{eq:flow-menu} on $\R^K$, be the utility of $M$ with every price multiplied by $1-\eta$. For a nonnegative smooth kernel $\varrho_\delta$ of mass one supported in the $\norm{\cdot}_\infty$-ball of radius $\delta$, set
\[
w_{\eta,\delta}:=\varrho_\delta*u_\eta,
\qquad
\widetilde u_{\eta,\delta}:=w_{\eta,\delta}-w_{\eta,\delta}(y_0).
\]
Convolution preserves convexity and averages gradients in the closed convex set $\mathcal{Q}$, so $\widetilde u_{\eta,\delta}\in\mathcal{U}^1$.

Fix $y\in X$ with selected option $(q,p)$ in $M$. By Rademacher's theorem, $u_\eta$, convex and Lipschitz, is differentiable at almost every $z$ with $\norm{z-y}_\infty\le\delta$. At such $z$, choose $(q_z,p_z)\in M$ whose discounted version is active at $z$; then $q_z=\nabla u_\eta(z)$ and $u_\eta(z)=z\cdot q_z-(1-\eta) p_z$. Adding the optimality inequalities $y\cdot q-p\ge y\cdot q_z-p_z$ in $M$ and $z\cdot q_z-(1-\eta) p_z\ge z\cdot q-(1-\eta) p$ in the discounted menu, with $\norm{q_z-q}_1\le2$, gives $\eta\,(p_z-p)\ge(y-z)\cdot(q_z-q)\ge-2\delta$. With $\norm{q_z}_1\le1$ and $p\le\bar p$,
\[
y\cdot q_z-u_\eta(z)
=(1-\eta) p_z+(y-z)\cdot q_z
\ge p-\eta \bar p-\frac{2\delta}{\eta}-\delta .
\]
Averaging against $\varrho_\delta(y-z)\,\d z$, with $w_{\eta,\delta}(y_0)\ge0$ because $u_\eta\ge0$ on $\R^K$, gives
\[
y\cdot\nabla \widetilde u_{\eta,\delta}(y)-\widetilde u_{\eta,\delta}(y)
\ge p_M(y)-\eta \bar p-\frac{2\delta}{\eta}-\delta .
\]
Integrate against $\Pi$ with \eqref{eq:TP}, take $\delta=\eta^2$ and let $\eta\downarrow0$: the supremum over $\mathcal{U}^1$ is at least $R(M)$ for every menu, which proves \eqref{eq:smooth-primal}.

\emph{Step 3: no duality gap.} Let $C^1(X)$ be the functions continuously differentiable in $\operatorname{int}X$ whose values and first derivatives extend continuously to $X$, and let $E:=\{\phi\in C^1(X):\phi(y_0)=0\}$ with its $C^1$ norm. The functional $\phi\mapsto\langle T_\Pi,\phi \rangle$ of \eqref{eq:TP} and the operator $\nabla:E\to C(X;\R^K)$ are continuous and linear. By the Riesz representation theorem, the dual of $C(X;\R^K)$ is the space $\mathcal{M}(X;\R^K)$ of finite $\R^K$-valued measures, and $\langle \nabla^*J,\phi \rangle=\int_X \nabla \phi\cdot\d J$. Let $\mathcal{C}\subset E$ be the closed convex cone of convex, coordinatewise nondecreasing functions, $E_{\mathcal{Q}}$ the set of continuous $\mathcal{Q}$-valued fields on $X$, and define the convex, proper, lower semicontinuous functions
\[
F(\phi):=\iota_{\mathcal{C}}(\phi)-\langle T_\Pi,\phi \rangle,
\qquad
G(g):=\iota_{E_{\mathcal{Q}}}(g),
\]
with $\iota$ the convex indicator.

For $\phi\in\mathcal{C}$ and $0<\theta<1$, the contraction toward the center $y^\circ$ of $X$,
\[
\phi_\theta(y):=\phi\bigl(y^\circ+\theta (y-y^\circ)\bigr)-\phi\bigl(y^\circ+\theta (y_0-y^\circ)\bigr),
\]
is defined on a neighborhood of $X$, is convex and nondecreasing, vanishes at $y_0$, has its gradient in $\theta \mathcal{Q}\subseteq\mathcal{Q}$ wherever $\nabla \phi\in\mathcal{Q}$, and converges to $\phi$ in $C^1(X)$ as $\theta\uparrow1$. First, restrictions of utilities in $\mathcal{U}^1$ are therefore dense in $\{\phi\in\mathcal{C}:\nabla \phi\in E_{\mathcal{Q}}\}$, and by Step 2
\[
\inf_{\phi\in E}\,\{F(\phi)+G(\nabla \phi)\}=-\sup_{u\in\mathcal{U}^1} \langle T_\Pi,u \rangle=-R^*,
\]
which is finite because $0\le R^*\le\bar p$. Second, the test functions of admissibility, less their values at $y_0$, are dense in $\mathcal{C}$; since $\langle T_\Pi,1 \rangle=0$ and both sides of the admissibility inequality are continuous in the $C^1$ norm, a flow $J$ is admissible if and only if $\int_X \nabla \phi\cdot\d J\ge\langle T_\Pi,\phi \rangle$ for every $\phi\in\mathcal{C}$.

For $q^\circ\in\operatorname{int} \mathcal{Q}$, the affine potential $\phi^\circ(y):=q^\circ\cdot(y-y_0)$ lies in $\mathcal{C}$, and $G$ is continuous at $\nabla \phi^\circ=q^\circ$ because $E_{\mathcal{Q}}$ contains a ball around it. The Fenchel--Rockafellar theorem \citep{Rockafellar1974} therefore gives
\begin{equation}\label{eq:flow-fenchel}
\inf_{\phi\in E}\,\{F(\phi)+G(\nabla \phi)\}
=
\max_{J\in\mathcal{M}(X;\R^K)}\,\bigl\{-F^*(-\nabla^*J)-G^*(J)\bigr\},
\end{equation}
with the maximum attained because the left side is finite. Since $\mathcal{C}$ is a cone, the characterization of admissibility gives
\[
F^*(-\nabla^*J)
=\sup_{\phi\in\mathcal{C}} \Bigl\{\langle T_\Pi,\phi \rangle-\int_X \nabla \phi\cdot\d J\Bigr\}
=
\begin{cases}
0,&-\nabla\cdot J\succeq_{\mathrm{mcvx}}T_\Pi,\\
+\infty,&\text{otherwise}.
\end{cases}
\]
With $\zeta=\d J/\d |J|$, the conjugate of $G$ is the cost,
\[
G^*(J)=\sup_{g\in E_{\mathcal{Q}}} \int_X g\cdot\d J=\int_X c(\zeta)\,\d |J|=\mathcal{A}(J).
\]
The upper bound is pointwise: $g(y)\cdot\zeta(y)\le c(\zeta(y))$ whenever $g(y)\in\mathcal{Q}$. For the lower bound, choose measurably a vertex of $\mathcal{Q}$ attaining $c(\zeta)$, by the Kuratowski--Ryll-Nardzewski theorem \citep{KuratowskiRyllNardzewski1965}, and approximate it in $L^1(|J|)$ by continuous fields composed with the $1$-Lipschitz metric projection onto $\mathcal{Q}$, whose integrals against $J$ converge to $\mathcal{A}(J)$. The right side of \eqref{eq:flow-fenchel} is therefore the maximum of $-\mathcal{A}(J)$ over admissible flows; with the left side $-R^*$, this is the equality of the values of \eqref{eq:flow-primal} and \eqref{eq:flow-dual}, and a maximizer is an admissible flow of least cost.
\end{proof}

\subsection{The two reductions}\label{app:reductions}

\begin{lemma}[Revenue formula]\label{lem:revenue-invariance}
Assume \hyperref[def:tangent]{tangency}, let $u$ be convex on $\R^K$ with subgradients in $\mathcal{Q}$, and write $U:=u\circ v$ and $\mu:=\mu_{a_v}$.
\begin{enumerate}[label=\textnormal{(\roman*)}]
\item Every incentive-compatible mechanism implementing $u$ has payment $t(x)=a_v(x)\cdot\nabla U(x)-U(x)$ at $\lambda_D$-almost every $x\in D$. Expected revenue is therefore a function of $U$ alone.
\item Every such mechanism, in particular a compact menu $M$ with $U=U_M:=u_M\circ v$, earns $\int_D U\,\d\mu-U(x_0)$, and $\mu(D)=0$. Hence compact menus $M$ and $M'$ whose \emph{utility increment} $w:=U_{M'}-U_M$ satisfies $w(x_0)=0$ satisfy the \emph{increment identity}
\[
R(M')-R(M)=\int_D w\,\d\mu .
\]
\end{enumerate}
\end{lemma}

\begin{proof}
\emph{Part (i).} As $\mathcal{Q}$ is bounded and $v$ smooth, $U$ is Lipschitz and, by Rademacher's theorem, differentiable at $\lambda_D$-almost every $x\in D$. At such an interior $x$, every $q\in\partial u(v(x))$ and direction $h$ satisfy
\[
U(x+th)-U(x)\ge q\cdot\bigl[v(x+th)-v(x)\bigr].
\]
The limits $t\downarrow0$ and $t\uparrow0$ give $\nabla U(x)\cdot h=q\cdot Jv(x)h$, hence $Jv(x)^\top q=\nabla U(x)$, and tangency gives $q\cdot v(x)=q\cdot Jv(x)a_v(x)=a_v(x)\cdot\nabla U(x)$. The allocation $q(x)$ lies in $\partial u(v(x))$, and incentive compatibility gives $t(x)=q(x)\cdot v(x)-U(x)$, hence the payment.

\emph{Part (ii).} By part (i), the mechanism earns $\int_D(a_v\cdot\nabla U-U)f\,\d x$. The Gauss--Green formula for the Lipschitz $U$ and the $C^1$ field $fa_v$ on the piecewise-smooth $D$ gives
\[
\int_D f\,a_v\cdot\nabla U\,\d x=\int_{\partial D}U f\,a_v\cdot n\,\d\sigma-\int_D U\,\nabla\cdot(fa_v)\,\d x .
\]
By the definition of $\mu_{a_v}$ (\Cref{sec:typespace}), the revenue is $\int_D U\,\d\mu-U(x_0)$; at $U\equiv1$, the Gauss--Green identity gives $\mu(D)=1-\int_D f\,\d x=0$. A compact menu $M$ implements $u_M$ under seller-favorable ties (\Cref{sec:menus}), and subtraction gives the increment identity, as $w(x_0)=0$.
\end{proof}

\begin{proof}[Proof of \Cref{prop:tangency-char}]
After (i) $\Rightarrow$ (iii) $\Rightarrow$ (ii), the proof of (ii) $\Rightarrow$ (i) builds, from a type $\hat x$ and a perturbation $\tau$, two utilities in $\mathcal{U}^1$ that agree on $Y$ and earn different revenue; a final part proves that $\hat x$ exists. The key step builds $\tau$, which vanishes on $Y$ and increases along the valuation vector: its gradient at $\hat y$ is the residual.

\emph{(i) $\Rightarrow$ (iii) $\Rightarrow$ (ii).} Under tangency, $v\cdot\nabla u(v)=a_v\cdot\nabla U$ in \eqref{eq:TP} for $u\in\mathcal{U}^1$ and $U:=u\circ v$, and the Gauss--Green identity of \Cref{lem:revenue-invariance}(ii) gives (iii), which implies (ii).

\emph{(ii) $\Rightarrow$ (i): the type $\hat x$.} Assume (ii) and that tangency fails. It fails where $r_v\neq0$, a nonempty open set of interior types because $r_v$ is continuous and $D=\overline{\operatorname{int}D}$. Choose $\hat x$ in it at which $Y$ is a \emph{single sheet}: for every sufficiently small open ball $\mathcal{G}\ni\hat x$, some open $O\ni\hat y:=v(\hat x)$ has $Y\cap O=v(\mathcal{G})\cap O$. For a small such $\mathcal{G}$ in the failure set, the \emph{sheet} $Y\cap O$ is open in $v(\mathcal{G})$, an embedded submanifold by the constant-rank theorem.

\emph{The perturbation $\tau$.} The least-squares residual $r:=r_v(\hat x)$ is orthogonal to the tangent space $\operatorname{range}Jv(\hat x)$ of the sheet at $\hat y$. In a slice chart at $\hat y$ the Jacobian rows of the last $K-d$ coordinates span the normal space; hence a combination of these coordinates is, after shrinking $O$, a smooth $g$ with
\[
g=0\text{ on }Y\cap O,\qquad\nabla g(\hat y)=r .
\]
The decomposition $\hat y=Jv(\hat x)a_v(\hat x)+r$ gives $\hat y\cdot r=\abs{r}^2>0$, so $y\cdot\nabla g(y)>0$ on a smaller $O$. For a smooth cutoff $\rho\ge0$ with compact support in $O$ and $\rho(\hat y)>0$, let $\tau:=\rho g$, extended by zero outside $O$. Then $\tau$ is smooth, vanishes on $Y$ with $\nabla\tau=0$ on $Y\setminus O$, and on $Y\cap O$ satisfies
\[
y\cdot\nabla\tau(y)=\rho(y)\,y\cdot\nabla g(y)\ge0,
\]
strictly where $\rho>0$. As $Y$ contains $y_0$ and the support of $\Pi$, \eqref{eq:TP} gives
\[
\langle T_\Pi,\tau\rangle=\int_D v(x)\cdot\nabla\tau(v(x))f(x)\,\d x>0,
\]
as the integrand is nonnegative and, since $f>0$, positive near the interior $\hat x$.

\emph{Two utilities.} Choose $q^\circ\in\operatorname{int}\mathcal{Q}$, an open box $W\supseteq X$ and $k_0>0$ so small that the gradients of
\[
u_0(y):=q^\circ\cdot(y-y_0)+\frac{k_0}{2}\abs{y-y_0}^2
\]
on $W$ lie in a compact subset of $\operatorname{int}\mathcal{Q}$. Its Hessian is $k_0I$ and $\tau$ has bounded first and second derivatives, so $u_0$ and $u_\varepsilon:=u_0+\varepsilon\tau$ lie in $\mathcal{U}^1$ for small $\varepsilon>0$. The two utilities agree on $Y$, including at $y_0$, but, contradicting (ii),
\[
\langle T_\Pi,u_\varepsilon\rangle-\langle T_\Pi,u_0\rangle=\varepsilon\langle T_\Pi,\tau\rangle>0 .
\]

\emph{Existence of the type $\hat x$.} By the rank condition $v$ is locally injective, so its fibers in the compact $D$ are finite. Take an open ball $\mathcal{G}_0$ in the failure set on which $v$ is injective, $\hat x\in\mathcal{G}_0$ minimizing the number of preimages of $v(x)$ over $\mathcal{G}_0$, and neighborhoods $\mathcal{N}_j$ of the $m$ preimages of $\hat y$, with disjoint compact closures and $v$ injective on each closure. By compactness, preimages in $D$ of points near $\hat y$ lie in $\bigcup_j\mathcal{N}_j$. For a small ball $\mathcal{G}\ni\hat x$ in $\mathcal{G}_0$, each $v(x)$ with $x\in\mathcal{G}$ has at least $m$ preimages by minimality and at most one in each $\mathcal{N}_j$, hence one, $\pi_j(x)$, in each. The inverse of $v$ on $\overline{\mathcal{N}_j}$ is continuous, so $\pi_j$ is continuous and injective, and $\pi_j(\mathcal{G})$ is open in $\R^d$ by invariance of domain. These sets lie in $D$ and contain all preimages of $\hat y$, so compactness gives an open $O\ni\hat y$ with $v^{-1}(O)\cap D\subseteq\bigcup_j\pi_j(\mathcal{G})$, and $v(\pi_j(\mathcal{G}))=v(\mathcal{G})$ gives $Y\cap O=v(\mathcal{G})\cap O$.

Hence the type $\hat x$, the perturbation $\tau$ and the two utilities contradict (ii), and the three statements are equivalent.
\end{proof}

\begin{proof}[Proof of \Cref{prop:separation}]
The ``if'' directions hold because $\mu_a$ is determined by $f$ and $a$ and every linear map is tangent with radial field $x$. I now prove the ``only if'' directions.

\emph{Part (i).} Let $b:=a_v-a_{\tilde v}$ for tangent $v$ and $\tilde v$ with $\mu_{a_v}=\mu_{a_{\tilde v}}$ for every positive $C^2$ density $f$. The measures differ on $\operatorname{int}D$ by $-\nabla\cdot(fb)\,\d x$; hence $\nabla\cdot(fb)=0$ there. The uniform density gives $\nabla\cdot b=0$, and each density $f_i\propto e^{x_i}$ gives
\[
0=\nabla\cdot(f_ib)=f_i(\nabla\cdot b+b_i)=f_ib_i .
\]
Hence $b_i=0$; continuity extends $b=0$ to $D$.

\emph{Part (ii).} By part (i) and the additive member, every member $v$ has radial field $x$, and tangency $Jv(x)x=v(x)$ gives, for every bundle $B$,
\[
x\cdot\nabla v_B(x)=v_B(x),\qquad x\in D .
\]
At $x=0$, $v_B(0)=0$. As $D$ is a box with bottom $0$, $tx\in D$ for $x\in D$ and $t\in(0,1]$, and
\[
\frac{\d}{\d t}\frac{v_B(tx)}{t}=\frac{tx\cdot\nabla v_B(tx)-v_B(tx)}{t^2}=0,\qquad t\in(0,1].
\]
Hence $v_B(x)=v_B(tx)/t\to\nabla v_B(0)\cdot x$ as $t\downarrow0$, and every member is linear.
\end{proof}

\section{Proof of the core-bundle theorem}\label{app:core}\label{os:align-transfer}

This appendix proves \Cref{thm:core-bundle} and \Cref{lem:restricted-opt}. Under the seller-favorable tie convention each type takes one option, and the cells $\{\Omega_\ell\}_{\ell \ge 0}$ partition $D$ as Borel sets.

\textbf{Nesting.} The dominance $\nu^-_{\ell,\varepsilon} \preceq_{\mathrm{FOSD}} \nu^+_{\ell,\varepsilon}$ at a selling cell $\ell$ and a slack $\varepsilon > 0$ implies the dominance at every slack $\varepsilon' \in [0, \varepsilon]$. Every bundle in the support of $\sigma_\ell$ contains the core. Hence $s_\ell(x) = \sum_{i \in C} x_i + \sum_{j \in P} w_{\ell,j}\, x_j$, and the unbundling scores satisfy
\[
q^{\varepsilon'}_{\ell,T}(x) = \frac{\varepsilon'}{\varepsilon}\, q^{\varepsilon}_{\ell,T}(x) + \Bigl(1 - \frac{\varepsilon'}{\varepsilon}\Bigr)\, s_\ell(x),
\qquad
s_\ell(x) = q^{\varepsilon}_{\ell,C}(x) + \varepsilon \sum_{j \in P} w_{\ell,j}\, x_j ,
\]
where the last sum adds the upward coordinates of $\mathsf{D}_\ell\, x|_P$. Therefore $Q^{\varepsilon'}_\ell$ is a linear function of $Q^{\varepsilon}_\ell$ with nonnegative coefficients. Every bounded Borel coordinatewise nondecreasing function of $Q^{\varepsilon'}_\ell$ is then one of $Q^{\varepsilon}_\ell$, and the dominance at $\varepsilon$ implies the dominance at $\varepsilon'$. At $\varepsilon' = 0$ every unbundling score equals $s_\ell$, and dropping the repeated coordinates gives the option score $Q_\ell$ of \hyperref[def:weak-core]{weak core regularity}. In particular, the dominance at a common positive slack at every selling cell implies weak core regularity.

\subsection{Proof of \texorpdfstring{\Cref{thm:core-bundle}}{the core-bundle theorem}}\label{app:core-L6}

A pair $(y, x)$ of a transport plan lists the source $y$ first, and $h := y - x$ is its \emph{displacement}. For $h \in \R^N$ write $h_i^+ := \max\{h_i, 0\}$ and $h_i^- := \max\{-h_i, 0\}$, and let $R_C(h) := \sum_{i \in C} h_i^-$ be the \emph{core reversal}. With $V^\alpha$ the matrix of the valuation map \eqref{eq:standing-v} at premium $\alpha$, the \emph{gain} of bundle $B$ along $h$ is $v^\alpha_B(h) := \alpha^{\mathbf{1}\{C \subseteq B\}} \sum_{i \in B} h_i$, and that of a selling option $\ell$ is $\alpha\, s_\ell(h)$. The dual cost (\Cref{sec:bundle-to-cone}) is $c_\alpha := c_{V^\alpha}$. The proof builds its plan with Strassen's theorem.

\begin{theorem}[Strassen's theorem \citep{Strassen1965,Kellerer1984}]\label{thm:rel-strassen}
Let $\nu_1, \nu_2$ be finite nonnegative Borel measures of equal mass on a compact subset of $\R^m$. Then $\nu_1 \preceq_{\mathrm{FOSD}} \nu_2$ if and only if there is a coupling of $(\nu_2, \nu_1)$ concentrated on $\{(z, z') : z \ge z' \text{ coordinatewise}\}$.
\end{theorem}

\begin{proof}[Proof of \Cref{thm:core-bundle}]
Let $\varepsilon > 0$ be an unbundling slack of strict core regularity. I construct one plan $\gamma'$, independent of the premium, that satisfies the conditions of \Cref{prop:saddle-D} at every $\alpha \ge 1/\varepsilon$, in three steps: mass balance, transport by Strassen's theorem, and coverage. Step 3 is the key step, because one premium must cover a continuum of displacements, including those with option gain near zero.

\emph{Step 1: mass balance.} A selling cell is \emph{mass-balanced} if
\begin{equation}\label{eq:MB}
\text{(MB)}\qquad \mu(\Omega_\ell) = 0, \quad\text{equivalently}\quad \mu^+(\Omega_\ell) = \mu^-(\Omega_\ell), \qquad \ell \ge 1 .
\end{equation}
Every selling cell is mass-balanced, because testing the dominance at slack $\varepsilon$ against the constants $1$ and $-1$ equates the masses of $\mu^+_\ell$ and $\mu^-_\ell$. The exclusion cell is balanced too. It contains the origin, where every selling option gives the utility $-p_\ell < 0$. The sources $\mu^+$ are concentrated on the origin and the top faces (\Cref{sec:main}), and by inclusivity no top face meets $\Omega_0$; hence $\mu^+_0 = \delta_0$. Global balance $\mu(D) = 0$ (\Cref{lem:revenue-invariance}(ii)) and \eqref{eq:MB} give $\mu(\Omega_0) = 0$, hence $\mu^-(\Omega_0) = 1$.

\emph{Step 2: transport by Strassen's theorem.} Each selling cell $\ell$ has a coupling $\gamma'_\ell$ of $(\mu^+_\ell, \mu^-_\ell)$ whose displacements satisfy, almost everywhere, the \emph{sign inequalities} and the \emph{slack inequality}
\begin{equation}\label{eq:core-budget}
w_{\ell,j}\, h_j \ge 0 \ge (1 - w_{\ell,j})\, h_j \quad (j \in P),
\qquad
s_\ell(h) \ge \varepsilon \bigl[s_\ell(h) + R_C(h)\bigr].
\end{equation}
By Step 1, $\mu^+_\ell$ and $\mu^-_\ell$ have equal mass. If this mass is zero, the zero plan is this coupling. Otherwise Strassen's theorem (\Cref{thm:rel-strassen}) gives a coupling $\pi_\ell$ of $(\nu^+_{\ell,\varepsilon}, \nu^-_{\ell,\varepsilon})$ on the compact set $Q_\ell^{\varepsilon}(\overline{\Omega_\ell})$, concentrated on pairs $(z^+, z^-)$ with $z^+ \ge z^-$. The disintegration theorem on $\overline{\Omega_\ell}$ gives probability kernels $K^\pm(z, \cdot)$ of $\mu^\pm_\ell$ along $Q_\ell^{\varepsilon}$, and
\[
\gamma'_\ell(\d y, \d x) := \int K^+(z^+, \d y)\, K^-(z^-, \d x)\, \pi_\ell(\d z^+, \d z^-)
\]
has marginals $\mu^\pm_\ell$, hence is concentrated on $\Omega_\ell \times \Omega_\ell$, and charges only pairs with ordered scores, $Q_\ell^{\varepsilon}(y) \ge Q_\ell^{\varepsilon}(x)$. Ordered scores imply \eqref{eq:core-budget}. Since $Q_\ell^{\varepsilon}$ is linear, they state $Q_\ell^{\varepsilon}(h) \ge 0$, and its peripheral coordinates $\mathsf{D}_\ell\, h|_P \ge 0$ are the sign inequalities. If some core coordinate of $h$ is negative, the unbundling coordinate at the set $T$ of these coordinates is $(1 - \varepsilon)\, s_\ell(h) - \varepsilon R_C(h) \ge 0$, the slack inequality. If none is negative, then $R_C(h) = 0$, and the sign inequalities give $s_\ell(h) \ge 0$, which implies the slack inequality $s_\ell(h) \ge \varepsilon\, s_\ell(h)$ because $\varepsilon \le 1$.

By Step 1, $\gamma'_0 := \delta_0 \otimes \mu^-_0$ couples $\mu^+_0$ and $\mu^-_0$. Because the cells partition $D$, the plan $\gamma' := \sum_{\ell = 0}^L \gamma'_\ell$ has marginals $\mu^\pm$ with difference $\mu$, which is condition (iii); each summand is concentrated on $\Omega_\ell \times \Omega_\ell$, which is condition (i). Multiplication by $\alpha > 0$ preserves every comparison between options and the seller-favorable order of tied prices, because option $\ell$ gives type $x$ the utility $\alpha\,(s_\ell(x) - \hat s_\ell)$. Hence the cells, and by separation (\Cref{sec:bundle-to-cone}) $\mu^\pm$ and $\gamma'$, do not depend on $\alpha$.

\emph{Step 3: coverage.} Fix $\alpha \ge 1/\varepsilon$, a selling cell $\ell$ and a $\gamma'_\ell$-typical pair $(y, x)$, $h = y - x$. Condition (ii) requires $h$ to lie in the cone of every bundle that option $\ell$ allocates. For each $j \in P$, the sign inequalities give the \emph{sign table}
\[
\begin{array}{c|c}
w_{\ell,j} = 1 & h_j \ge 0\\
w_{\ell,j} = 0 & h_j \le 0\\
0 < w_{\ell,j} < 1 & h_j = 0.
\end{array}
\]
In every case $w_{\ell,j}\, h_j = h_j^+$. With $\sum_{i \in C} h_i = \sum_{i \in C} h_i^+ - R_C(h)$, this gives
\begin{equation}\label{eq:B-max-gain}
\max_{B \in \mathcal{F}_C} \sum_{i \in B} h_i = \sum_{i \in C} h_i + \sum_{j \in P} h_j^+ = s_\ell(h) = \sum_{i \in [N]} h_i^+ - R_C(h).
\end{equation}
Every bundle that option $\ell$ allocates attains the largest gain along $h$. Two cases compare it with the bundles that contain the core and with those that do not.

\textit{Core-containing bundles.} Every bundle in $\operatorname{supp}\sigma_\ell$ attains the maximum in \eqref{eq:B-max-gain}, hence the largest gain $\alpha\, s_\ell(h)$ among the bundles that contain the core: it contains every peripheral good of weight one, omits every peripheral good of weight zero, and the goods of intermediate weight have zero displacement. This gain is at least the gain $0$ of exclusion, because $s_\ell(h) \ge \varepsilon[s_\ell(h) + R_C(h)] \ge 0$ by the slack inequality and \eqref{eq:B-max-gain}.

\textit{Bundles that do not contain the core.} Such a bundle receives no premium, and bounding each coordinate by its positive part gives
\[
v^\alpha_B(h) = \sum_{i \in B} h_i \;\le\; \sum_{i \in [N]} h_i^+ \;=\; s_\ell(h) + R_C(h) \;\le\; \frac{s_\ell(h)}{\varepsilon} \;\le\; \alpha\, s_\ell(h) \qquad (B \notin \mathcal{F}_C),
\]
where the second equality is the last equality of \eqref{eq:B-max-gain}, the second inequality is the slack inequality, and the last holds because $\alpha \varepsilon \ge 1$ and $s_\ell(h) \ge 0$.

The two cases give $c_\alpha(h) = \alpha\, s_\ell(h) = v^\alpha_B(h)$, hence $h \in K_B^{\mathcal{B}}$, for every $B \in \operatorname{supp}\sigma_\ell$. On the exclusion plan every bundle has a nonpositive gain along $h = -x$, and $h \in K_\varnothing^{\mathcal{B}}$.

If the closure of a second cell $\Omega_m$ contains both types of a pair, continuity of $U_M$ and the affine formula of option $m$ on $\overline{\Omega_m}$ give $U_M(y) - U_M(x) = \sum_{B \in \mathcal{B} \cup \{\varnothing\}} \sigma_m(B)\, v^\alpha_B(h)$, an average of bundle gains, with $v^\alpha_\varnothing = 0$. This average equals $c_\alpha(h)$ on the plan: on a selling-cell pair the affine formula of option $\ell$ on $\Omega_\ell$, with gradient $\alpha\, w_\ell$, gives $U_M(y) - U_M(x) = \alpha\, s_\ell(h)$, which is $c_\alpha(h)$ by the two cases, and on an exclusion pair both sides vanish. Each gain is at most $c_\alpha(h)$. Hence every bundle in $\operatorname{supp}\sigma_m$ attains $c_\alpha(h)$, and $h$ lies in its cone.

By Steps 2 and 3, $\gamma'$ satisfies the three conditions of \Cref{prop:saddle-D} at every $\alpha \ge 1/\varepsilon$. The origin lies in $\Omega_0$, where $U_M = 0$; hence $M$ is optimal in the unrestricted problem at every such premium. Every optimal menu at a premium $\alpha > 1/\varepsilon$ is core-peripheral, by \Cref{cor:premium}(ii) applied to $M$ at the premium $1/\varepsilon$ with $\mathcal{S} = \mathcal{F}_C$: raising the premium from $1/\varepsilon$ to $\alpha$ multiplies by $\alpha\varepsilon > 1$ the value of every bundle in $\mathcal{F}_C$, and the box $D$, the singletons in $\mathcal{F} = \mathcal{B}$ and the positive sink density $\Phi f$ satisfy the other conditions of \Cref{cor:premium}(ii). Both conclusions of \Cref{thm:core-bundle} therefore hold at every premium above $1/\varepsilon$.
\end{proof}

\subsection{Restricted optimality}\label{lem:mb-foc}

This subsection shows that mass balance is the first-order condition for each price and proves \Cref{lem:restricted-opt}: inclusivity and weak core regularity imply restricted optimality.

\emph{Mass balance as the first-order condition.} At prices $p$ the indirect utility is $U_p(x) = \max\bigl\{0,\ \max_m(\alpha s_m(x) - p_m)\bigr\}$, with $U_p(0) = 0$ by positive prices, so $R = \int_D U_p\, \d\mu$ by \Cref{lem:revenue-invariance}(ii). The map $\delta \mapsto U_{p + \delta e_\ell}(x)$ is $1$-Lipschitz with derivative $-\mathbf{1}_{\Omega_\ell}(x)$ off the tie set. The tie set is a finite union of hyperplanes: it is $\lambda_D$-null, avoids the origin because prices are positive, and meets each top face in dimension at most $N-2$ when the menu is \emph{face-generic}, that is, when no indifference hyperplane coincides with a top face. Face-genericity is without loss: at such a coincidence one of the two options is weakly dominated on $D$ and can be deleted. The tie set is therefore $\mu^\pm$-null, and dominated convergence gives $\d R/\d p_\ell = -\mu(\Omega_\ell)$: a price with a cell of positive measure is a critical point of revenue if and only if the cell is mass-balanced. Hence restricted optimality implies mass balance for every face-generic menu.

The restricted problem has the \emph{restricted cost}
\[
c^{\mathcal{F}_C}_\alpha(h) := \max\Bigl\{0,\ \alpha \max_{B \in \mathcal{F}_C} \sum_{i \in B} h_i\Bigr\}.
\]

\begin{proof}[Proof of \Cref{lem:restricted-opt}]
Fix $\alpha > 0$. The proof has three parts: the plan, attainment and weak duality. The plan and attainment reuse Steps 1 to 3 of the proof of \Cref{thm:core-bundle} at slack zero. Attainment shows that every bundle of $M$ attains the restricted cost along the plan's displacements, and weak duality compares $M$ with every menu on $\mathcal{F}_C$.

\emph{The plan.} Weak core regularity is the dominance at slack zero, by nesting at $\varepsilon' = 0$. Hence Steps 1 and 2 of the proof of \Cref{thm:core-bundle}, run at slack zero, construct from inclusivity a plan $\gamma' = \delta_0 \otimes \mu^-_0 + \sum_{\ell \ge 1} \gamma'_\ell$ with marginals $\mu^\pm$; each $\gamma'_\ell$ is concentrated on $\Omega_\ell \times \Omega_\ell$, with displacements that satisfy \eqref{eq:core-budget} at $\varepsilon = 0$.

\emph{Attainment.} Fix a selling cell $\ell$ and a displacement $h$ of $\gamma'_\ell$. The sign table, \eqref{eq:B-max-gain} and the first case of Step 3 follow from the sign inequalities and hold along $h$: every bundle in $\operatorname{supp}\sigma_\ell$ attains $\max_{B \in \mathcal{F}_C} \sum_{i \in B} h_i = s_\ell(h)$, which is nonnegative by the slack inequality at $\varepsilon = 0$. Hence $c^{\mathcal{F}_C}_\alpha(h) = \alpha\, s_\ell(h)$. On $\Omega_\ell$ the indirect utility is affine with gradient $(V^\alpha)^\top \sigma_\ell = \alpha\, w_\ell$; hence $U_M(y) - U_M(x) = c^{\mathcal{F}_C}_\alpha(h)$ for $\gamma'_\ell$-almost every pair. On the exclusion plan both sides vanish, because $U_M = 0$ on $\Omega_0$ and $x \ge 0$.

\emph{Weak duality.} Let $M'$ be a compact menu of lotteries on $\mathcal{F}_C$, to which the restricted problem reduces by the taxation principle, with indirect utility $U'$. Type $x$ can take the option $(q, p)$ of type $y$, and $\sum_B q_B \le 1$; hence
\[
U'(y) - U'(x) \;\le\; q \cdot V^\alpha h \;=\; \alpha \sum_{B \in \mathcal{F}_C} q_B \sum_{i \in B} h_i \;\le\; c^{\mathcal{F}_C}_\alpha(h).
\]
Since $U'(0) \ge 0$, \Cref{lem:revenue-invariance}(ii) and the marginals of $\gamma'$ give
\begin{align*}
R(M') \;=\; \int_D U'\, \d\mu - U'(0) &\;\le\; \int_{D \times D} \bigl[U'(y) - U'(x)\bigr]\, \d\gamma'(y, x)\\
&\;\le\; \int_{D \times D} c^{\mathcal{F}_C}_\alpha(y - x)\, \d\gamma'(y, x).
\end{align*}
For $M$, positive prices give $U_M(0) = 0$, and attainment gives equality in both inequalities. Hence $R(M) \ge R(M')$.
\end{proof}

\clearpage
\newgeometry{top=1.25in, bottom=1.25in, left=1.25in, right=1.25in, footskip=30pt}
\pagenumbering{arabic}
\setcounter{page}{1}
\begin{center}
{\Large\bfseries Online Supplement}
\end{center}

\renewcommand{\thesection}{S.\arabic{section}}
\renewcommand{\theHsection}{S.\arabic{section}}
\setcounter{section}{0}

\section{Proof of transport duality}\label{os:transport-proof}

Write $\nu:=v_\#\mu$. By \Cref{lem:revenue-invariance}(ii), $\mu(D)=0$; hence $\nu(Y)=0$. On the valuation space the primal and dual are
\begin{equation}\label{eq:primal}
\sup_{u\in\mathcal{U}}\int_Y u\,\d\nu,
\end{equation}
where $\mathcal{U}$ consists of the nonnegative restrictions to $Y$ of convex functions on $\R^K$ with subgradients in $\mathcal{Q}$, and
\begin{equation}\label{eq:dual}
\inf_{\gamma\in\Gamma_+(Y\times Y)}\int_{Y\times Y}c(y,y')\,\d\gamma(y,y')
\quad\text{s.t.}\quad
\gamma_1-\gamma_2\succeq_{\mathrm{mcvx}}\nu,
\end{equation}
with the dominance tested against $\mathcal{U}$. Here $\Gamma_+(Y \times Y)$ is the set of finite positive Borel measures on $Y \times Y$, and $c(y, y') := c(y - y')$. Every difference $\gamma_1-\gamma_2$ has mass zero, and $\nu(Y)=0$; hence adding a constant to a utility changes the value of neither side of the dominance. Pushforward by $v$ transfers a plan on $D$ to one on $Y$, and a Borel right inverse of the continuous surjection $v:D\to Y$ transfers it back. Both transfers preserve the cost and every tested integral. Therefore \eqref{eq:primal} and \eqref{eq:dual} have the values of \eqref{eq:primal-D} and \eqref{eq:dual-D}, and dual attainment on $Y$ implies dual attainment on $D$.

\begin{proof}[Proof of \Cref{thm:strong-duality}]
I follow the bounded-marginal-value argument of \citet{GOZ2002} and \citet{KM2019} in two steps. Step 1, the envelope bound, is the key step: a utility that satisfies only a relaxed constraint is made feasible while its value changes by at most the relaxation. Step 2, the supporting functional, is the standard sum-rule argument and gives an admissible plan that costs the primal value.

Let $\mathcal{U}_0$ be the cone of continuous, convex, coordinatewise nondecreasing functions on $Y$ that vanish at $y_0$. For $g\in C(Y\times Y)$ define
\[
(Au)(y,y'):=u(y)-u(y'),
\qquad
\mathcal{V}(g):=\sup\Bigl\{\int_Y u\,\d\nu:u\in\mathcal{U}_0,\ Au\le g\Bigr\},
\]
with $\mathcal{V}(g)=-\infty$ when the constraint is infeasible. A feasible utility $u$ is \emph{normalized} if $u(y_0)=0$. Every normalized feasible utility belongs to $\mathcal{U}_0$ and satisfies $Au\le c$; Step 1 proves the converse.

\emph{Step 1: the envelope bound.} Fix $\bar u\in\mathcal{U}_0$ with $A\bar u\le g$, and let $\delta:=\norm{g-c}_\infty$. Its one-sided envelope is
\[
u(y):=\inf_{z\in Y}\{\bar u(z)+c(y,z)\}
=\sup_{q\in\mathcal{Q}}\{q\cdot y-\bar u^*(q)\},
\qquad
\bar u^*(q):=\sup_{z\in Y}\{q\cdot z-\bar u(z)\}.
\]
The equality is the minimax theorem, applied to the continuous function $\bar u(z)+q\cdot(y-z)$, convex in $z$ and affine in $q$, on the compact convex sets $Y$ and $\mathcal{Q}$. The second expression defines a convex function on $\R^K$ with subgradients in $\mathcal{Q}$. In particular, $u$ is nondecreasing and $Au\le c$. Taking $z=y$ gives $u\le\bar u$, and $A\bar u\le g$ gives $u\ge\bar u-\delta$. Thus
\[
0\le\bar u-u\le\delta\quad\text{on }Y.
\]
For $g=c$ this proves $u=\bar u$. Therefore $\mathcal{U}_0\cap\{Au\le c\}$ is the class of normalized feasible utilities.

For general $g$, the function $\widetilde u:=u-u(y_0)$ belongs to this class and has the integral of $u$ with respect to $\nu$, because $\nu(Y)=0$. Applying $0\le\bar u-u\le\delta$ to the objective and taking the supremum over $\bar u$ gives
\begin{equation}\label{eq:os-value-lipschitz}
\mathcal{V}(g)-\mathcal{V}(c)\le\norm{\nu}_{TV}\,\norm{g-c}_\infty
\qquad\text{for every }g\in C(Y\times Y).
\end{equation}
The value $\mathcal{V}(c)$ is finite: $0$ is feasible, and every feasible utility satisfies $0\le u(y)\le c(y,y_0)$ on the compact set $Y$. It is the value of \eqref{eq:primal}, since for every $u\in\mathcal{U}$ the normalized utility $u-u(y_0)$ has the objective of $u$.

\emph{Step 2: the supporting functional.} The function $\mathcal{V}$ is concave. It is finite wherever the constraint is feasible, by \eqref{eq:os-value-lipschitz}, so $-\mathcal{V}$ is proper. By \eqref{eq:os-value-lipschitz}, the convex function $g\mapsto-\mathcal{V}(g)+\norm{\nu}_{TV}\norm{g-c}_\infty$ attains its minimum at $c$. The Moreau--Rockafellar sum rule \citep{Rockafellar1974} applies, as the norm term is continuous, and gives a continuous linear functional $\gamma$ on $C(Y\times Y)$ with
\[
\mathcal{V}(g)\le\mathcal{V}(c)+\langle\gamma,g-c\rangle
\qquad\text{for every }g\in C(Y\times Y).
\]
For $h\ge0$, $\mathcal{V}(c+h)\ge\mathcal{V}(c)$ implies $\langle\gamma,h\rangle\ge0$. By the Riesz representation theorem, $\gamma$ is a finite positive Borel measure on $Y\times Y$. Moreover, $\mathcal{V}(0)=0$, because $Au\le0$ implies that $u$ is constant, hence zero. Also $\mathcal{V}(2c)=2\mathcal{V}(c)$, because $\mathcal{U}_0$ is a cone. The supporting inequality at $g=0$ and at $g=2c$ then gives $\int c\,\d\gamma=\mathcal{V}(c)$. For $u\in\mathcal{U}_0$, the supporting inequality at $g=Au$ gives
\[
\int_Y u\,\d\nu
\le\mathcal{V}(Au)
\le\mathcal{V}(c)+\int(Au-c)\,\d\gamma
=\int_Y u\,\d(\gamma_1-\gamma_2).
\]
By the mass-zero fact of the lead-in, the inequality extends from $u-u(y_0)\in\mathcal{U}_0$ to every $u\in\mathcal{U}$. Hence $\gamma$ is admissible for \eqref{eq:dual}. Conversely, every admissible plan $\widetilde\gamma$ and every $u\in\mathcal{U}$ satisfy weak duality: $\int_Y u\,\d\nu\le\int[u(y)-u(y')]\,\d\widetilde\gamma\le\int c\,\d\widetilde\gamma$.

\emph{Assembly.} By Step 2, $\gamma$ is admissible and costs $\mathcal{V}(c)$, which by Step 1 is the value of \eqref{eq:primal}; by weak duality, no admissible plan costs less. Hence \eqref{eq:primal} and \eqref{eq:dual} have the same value and $\gamma$ attains it. The transfer of the lead-in gives the same conclusion for \eqref{eq:primal-D} and \eqref{eq:dual-D}. By \Cref{lem:revenue-invariance}, this value is the optimal revenue, and \Cref{thm:flow-duality} identifies it with the value of the flow dual \eqref{eq:flow-dual}.
\end{proof}

\section{Proof of the bundle-cone certificate}\label{os:saddle-proofs}

This section proves \Cref{prop:saddle} for linear $v(x) = Vx$, with $Y = VD$ convex and \hyperref[def:tangent]{tangency} at $a_v(x) = x$. The menu $M$ has options $(\sigma_\ell, p_\ell)$, with $\ell = 0$ the outside option, lotteries $q^\ell := (\sigma_\ell(B))_{B \in \mathcal{F}} \in \mathcal{Q}$ and ambient utility $u_M$; its indirect utility $U_M := u_M \circ v$ equals $(V^\top q^\ell) \cdot x - p_\ell$ on the cell $\Omega_\ell$. A remark after the proof shows that, at an optimal menu, every optimal plan with an absolutely continuous sink marginal satisfies condition (i).

\begin{proof}[Proof of \Cref{prop:saddle}]
The proof has three parts: equal values, weak duality and revenue.

\emph{Equal values.} The pair $(U_M - U_M(x_0), \gamma)$ is feasible for the primal and admissible for the dual, and its primal value equals its dual cost. The function $u_M - u_M(y_0)$ is convex with subgradients in $\mathcal{Q}$ (\Cref{sec:menus}) and vanishes at $y_0$; its composition with $v$ is $U_M - U_M(x_0)$, which therefore lies in $\mathcal{U}_D$. Condition (iii), $\gamma_1 - \gamma_2 = \mu$, makes $\gamma$ admissible. Conditions (i) and (ii) hold for $\gamma$-almost every pair $(x, x')$; let $\Omega_\ell$ be a cell whose closure contains both points, and $h := x - x'$. The affine formula for $U_M$ extends to the closure, and by condition (ii) and \eqref{eq:bundle-cone}, including $K_\varnothing^{\mathcal{F}}$, every $B \in \operatorname{supp} \sigma_\ell$ attains $v_B(h) = c_V(h)$. Hence
\[
U_M(x) - U_M(x') = \sum_{B \in \mathcal{F}_0} \sigma_\ell(B)\, v_B(h) = c_V(h).
\]
Integration against $\gamma$, condition (iii) and $\mu(D) = 0$ (\Cref{lem:revenue-invariance}(ii)) give
\[
\int_D \bigl[U_M - U_M(x_0)\bigr]\, \d\mu
= \int_{D \times D} \bigl[U_M(x) - U_M(x')\bigr]\, \d\gamma
= \int_{D \times D} c_V(x - x')\, \d\gamma .
\]

\emph{Weak duality.} For $U \in \mathcal{U}_D$ and $\tilde\gamma$ admissible for \eqref{eq:dual-D}, the dominance constraint and the utility-difference bound give
\[
\int_D U\, \d\mu
\le \int_{D \times D} \bigl[U(x) - U(x')\bigr]\, \d\tilde\gamma
\le \int_{D \times D} c_V(x - x')\, \d\tilde\gamma .
\]
Equal values therefore attain both optima: $U_M - U_M(x_0)$ solves \eqref{eq:primal-D} and $\gamma$ solves \eqref{eq:dual-D}, the first assertion.

\emph{Revenue.} Let $\widehat M$ be $M$ with the price of every option $\ell \ge 1$ raised by $U_M(x_0) = u_M(y_0)$. Deleting the options priced above $\bar p$, which no type takes, and adjoining the outside option give the normalized menu of \Cref{app:duality-flow}. Hence $R(\widehat M) = R(M) + U_M(x_0) = \int_D [U_M - U_M(x_0)]\, \d\mu$ by \Cref{lem:revenue-invariance}(ii) and $\mu(D) = 0$, the optimal primal value. Every incentive-compatible and individually rational mechanism implements a convex $u'$ with subgradients in $\mathcal{Q}$ and $u'(y_0) \ge 0$ (\Cref{sec:menus}). With $U' := u' \circ v$, \Cref{lem:revenue-invariance}(ii) and $\mu(D) = 0$ give its revenue
\[
\int_D U'\, \d\mu - U'(x_0)
= \int_D \bigl[U' - U'(x_0)\bigr]\, \d\mu - U'(x_0)
\le \int_D \bigl[U' - U'(x_0)\bigr]\, \d\mu .
\]
As for $U_M$, $U' - U'(x_0) \in \mathcal{U}_D$; hence this revenue is at most the optimal primal value, and $\widehat M$ is revenue-optimal, the second assertion.
\end{proof}

\begin{remark}[Within-cell support]\label{rmk:within-cell-support}
In the linear setting of \Cref{prop:saddle-D}, identify options that give the same affine indirect utility, and hence the same payment. If $U_M-U_M(x_0)$ solves \eqref{eq:primal-D}, then every optimal plan $\gamma$ of \eqref{eq:dual-D} whose sink marginal $\gamma_2$ is absolutely continuous satisfies condition (i) of \Cref{prop:saddle-D}.

\emph{Argument.} By \Cref{thm:strong-duality}, the two ends of the weak-duality chain are equal. Hence $U_M(x)-U_M(x')=c_V(x-x')$ for $\gamma$-almost every pair $(x,x')$. If option $k$ is selected at $x$, its supporting affine function and $q^k\cdot V(x-x')\le c_V(x-x')$ imply $q^k\cdot Vx'-p_k\ge U_M(x')$. The reverse inequality always holds; hence option $k$ is active at $x'$. At every differentiability point of $U_M$, all active options have the same affine piece. After the identification, $x'$ therefore belongs to $\Omega_k$. The nondifferentiability set of $U_M$ is Lebesgue-null by Rademacher's theorem, hence $\gamma_2$-null. Therefore $\gamma$-almost every transported pair lies in one cell.
\end{remark}

\section{Proof of optimality and selection under a premium}\label{os:premium-proof}

This section proves \Cref{cor:premium}.

\begin{proof}[Proof of \Cref{cor:premium}]
Part (i) shows that an optimal plan of the dual \eqref{eq:dual-D} under $v$ is admissible under $v^\alpha$ with cost $\alpha R(M)$, which $\alpha M$ earns. Its key step is that the premium scales the cost of the plan by $\alpha$ pair by pair. Part (ii) uses the barycentric dual, whose plans have the sinks of $\mu$ as their second marginal. Its four steps are the barycentric plan, the scaling, the strict gap at the sources, and the selection at almost every type. The selection is the key step: it extends the restriction on the bundles from the sources of the plan to almost every type.

\emph{Part (i).} Let $q^\ell \in \mathcal{Q}$ be the lottery of option $\ell$, let $U_M := u_M \circ v$, and let $A$ be the diagonal matrix with entry $\alpha$ at each bundle in $\mathcal{S}$ and $1$ elsewhere; then $V^\alpha = AV$. Optimality implies $U_M(x_0) = 0$ (the normalization of \Cref{app:duality-flow}). Hence $U_M \in \mathcal{U}_D$ and $R(M) = \int_D U_M\, \d\mu$ for the differential virtual value $\mu$ on types of $v$ (\Cref{lem:revenue-invariance} and \eqref{eq:primal-D}). The linear map $v$ is tangent with radial field $a(x) = x$, and $Y$ is convex. Hence, by \Cref{thm:strong-duality} and \Cref{thm:flow-duality}, the dual \eqref{eq:dual-D} has an optimal plan $\gamma$ with cost $\int c\bigl(V(x - x')\bigr)\, \d\gamma = R(M)$. Under $v^\alpha$, option $\ell$ of $\alpha M$ gives type $x$ the surplus $q^\ell \cdot V^\alpha x - \alpha p_\ell = \alpha\,(q^\ell \cdot Vx - p_\ell)$, because the support of $q^\ell$ lies in $\mathcal{S}$ and hence $A q^\ell = \alpha\, q^\ell$. Hence every type makes the same choice as in $M$ under the seller-favorable tie rule and pays $\alpha$ times its price. The menu $\alpha M$ therefore earns $\alpha R(M)$. The plan $\gamma$ is admissible under $v^\alpha$, for two reasons. By \Cref{prop:separation}, $\mu$ is also the differential virtual value on types of $v^\alpha$, which is linear with the same radial field. The constraint of \eqref{eq:dual-D} under $v^\alpha$ tests the functions $w \circ v^\alpha = (w \circ A) \circ v$ that vanish at $x_0$, where $w$ is convex with subgradients in $\mathcal{Q}$. Each is $\alpha$ times such a function under $v$, because $A \mathcal{Q} / \alpha \subseteq \mathcal{Q}$ for $\alpha \ge 1$ implies that $(w \circ A)/\alpha$ is convex with subgradients in $\mathcal{Q}$. The key step shows that $c(V^\alpha h) = \alpha\, c(Vh)$ for $\gamma$-almost every $(x, x')$, where $h := x - x'$. Let $q(x)$ be the lottery that type $x$ chooses in $M$, which is a subgradient of $u_M$ at $Vx$. The constraint of \eqref{eq:dual-D} at $U_M$, the subgradient inequality and $q(x) \in \mathcal{Q}$ give
\begin{align*}
R(M) = \int_D U_M\, \d\mu
&\le \int \bigl[U_M(x) - U_M(x')\bigr]\, \d\gamma \\
&\le \int q(x) \cdot V(x - x')\, \d\gamma
\le \int c\bigl(V(x - x')\bigr)\, \d\gamma = R(M).
\end{align*}
The second and third inequalities hold pair by pair, and the chain starts and ends with $R(M)$. Hence $q(x) \cdot Vh = c(Vh)$ for $\gamma$-almost every $(x, x')$. Fix such a pair and let $m := c(Vh)$. If $m > 0$, the inequalities $q(x) \cdot Vh \le (\mathbf{1}^\top q(x))\, m \le m$ hold with equality. Hence $\mathbf{1}^\top q(x) = 1$, and $(Vh)_B = m$ on the support of $q(x)$, which lies in $\mathcal{S}$. Then $(V^\alpha h)_B = \alpha m$ on that support, $(V^\alpha h)_B = \alpha (Vh)_B \le \alpha m$ on the rest of $\mathcal{S}$, and $(V^\alpha h)_B = (Vh)_B \le m \le \alpha m$ outside $\mathcal{S}$. Hence $c(V^\alpha h) = \alpha m$. If $m = 0$, every coordinate of $Vh$ is at most $0$, and so is every coordinate of $V^\alpha h = AVh$. Then $c(V^\alpha h) = 0$. In both cases $c(V^\alpha h) = \alpha\, c(Vh)$. The cost of $\gamma$ under $v^\alpha$ is therefore $\int c\bigl(V^\alpha (x - x')\bigr)\, \d\gamma = \alpha R(M)$. By weak duality under $v^\alpha$ (\Cref{os:saddle-proofs}), no menu earns more than the cost of the admissible plan $\gamma$. Hence $\alpha M$, which earns $\alpha R(M)$, is optimal under $v^\alpha$.

\emph{Part (ii).} Under the conditions of (ii), $D$ is convex, because $Y = VD$ is convex and $V$ is injective by the rank condition. The bottom type is $x_0 = 0$ and the radial field is $a_v(x) = x$. Hence the differential virtual value on types is $\mu = \mu^+ - \mu^-$ with
\[
\mu^+ := \delta_0 + f\,(x \cdot n)\,\sigma,
\qquad
\mu^- := \Phi f\, \lambda_D ,
\]
where $\Phi(x) := N + 1 + x \cdot \nabla \log f(x)$, because $\nabla \cdot (f x) + f = \Phi f$. Both parts are nonnegative: for an outward normal $n$ at $x \in \partial D$, the normal inequality $n \cdot (y - x) \le 0$ at $y = 0 \in D$ gives $x \cdot n \ge 0$, and $\Phi f > 0$ almost everywhere, because the sink density of $\mu$ is positive. They are the positive and negative parts of $\mu$, because $\mu^+$ is concentrated on the Lebesgue-null set $\{0\} \cup \partial D$ and $\mu^- \ll \lambda_D$. The measure $\mu^+$ charges only the origin and the boundary types with $x \cdot n > 0$, and both parts have the same mass, because $\mu(D) = 0$. The same holds under $v^\alpha$, because $\mu$ does not depend on the premium.

\emph{Step 1: the barycentric plan.} For a coupling $\pi$ of $(\mu^+, \mu^-)$, disintegrated as $\pi(\d x, \d x') = \mu^+(\d x)\, \pi_x(\d x')$, write $b_\pi(x) := \int x'\, \pi_x(\d x')$ for the barycenter of the sinks matched with the source $x$, and $h_\pi(x) := x - b_\pi(x)$. By the convex Kantorovich--Rubinstein formula \citep[Corollary~3.5 and Theorem~3.1]{BeiglbockPammerRiessSchrott2025}, applied to $\mu^+$ and $\mu^-$ divided by their common mass and to the support function $h \mapsto c(Vh)$ of $V^\top \mathcal{Q}$, the least barycentric cost $\int_D c(V h_\pi(x))\, \mu^+(\d x)$ over couplings is attained. It equals the largest $\int_D U\, \d\mu$ over convex functions $U$ on $\R^N$ with $U(z) - U(z') \le c(V(z - z'))$ for all $z, z'$, which are the convex functions with subgradients in $V^\top \mathcal{Q}$. Up to a constant, their restrictions to $D$ are the indirect utilities of $\mathcal{U}_D$. Indeed, for such a $U$, let $u$ be the supremum of the affine functions $y \mapsto q \cdot y - (V^\top q) \cdot x + U(x) - U(0)$ over $x \in D$ and $q \in \mathcal{Q}$ with $V^\top q \in \partial U(x)$. Their slopes lie in $\mathcal{Q}$, and their intercepts are bounded because $D$ is compact. Hence $u$ is finite and convex, with subgradients in $\mathcal{Q}$. At $Vz$ with $z \in D$, each of these functions is at most $U(z) - U(0)$ by the subgradient inequality, with equality at $x = z$, because $\partial U(z)$ is nonempty and lies in $V^\top \mathcal{Q}$. Hence $U = U(0) + u \circ v$ on $D$. Conversely, each $u \circ v \in \mathcal{U}_D$ extends to the convex function $z \mapsto u(Vz)$, whose subgradients lie in $V^\top \mathcal{Q}$. Hence the least barycentric cost is the optimal revenue $R(M)$, by \eqref{eq:primal-D} and \Cref{lem:revenue-invariance}(ii). Let $\pi$ attain it. The plan $\gamma := (\mathrm{id}, b_\pi)_\#\mu^+$ of \eqref{eq:dual-D} has the same cost. It is admissible: for $U \in \mathcal{U}_D$, Jensen's inequality gives $\int U\, \d\gamma_2 \le \int U\, \d\mu^-$, hence $\int U\, \d(\gamma_1 - \gamma_2) \ge \int U\, \d\mu$.

\emph{Step 2: the scaling.} The argument of part (i), applied to the plan $\gamma$, which is admissible and costs $R(M)$, gives $c(V^\alpha h_\pi(x)) = \alpha\, c(V h_\pi(x))$ for $\mu^+$-almost every $x$. Hence the barycentric cost of $\pi$ under $v^\alpha$ is $\alpha R(M)$, which $\alpha M$ earns (part (i)). By Step 1 under $v^\alpha$, the least barycentric cost under $v^\alpha$ equals the optimal revenue under $v^\alpha$. The first is at most the cost $\alpha R(M)$ of $\pi$, and the second is at least the revenue $\alpha R(M)$ of $\alpha M$. Hence both equal $\alpha R(M)$, and $\pi$ attains the least barycentric cost under $v^\alpha$.

\emph{Step 3: the strict gap at the sources.} For $\mu^+$-almost every source $x$, $\pi_x$ is a probability measure on the open convex set $\operatorname{int} D$, because $\mu^-(\partial D) = 0$. Its barycenter $b_\pi(x)$ therefore lies in $\operatorname{int} D$. At such a source $x \ne 0$, which has $x \cdot n > 0$ for its outward normal $n$, some good $i$ satisfies $x_i > b_{\pi,i}(x)$. Otherwise $x \le b_\pi(x)$; for a small $\varepsilon > 0$ the point $b_\pi(x) + \varepsilon\mathbf{1}$ lies in $D$, and $s := 1 + \varepsilon/\max_j x_j$ satisfies $0 \le s x \le b_\pi(x) + \varepsilon\mathbf{1}$. Then $s x \in D$ because $D$ is downward closed, and the normal inequality at $y = s x$ gives $(s - 1)\, x \cdot n \le 0$, a contradiction. The singleton $\{i\}$ therefore gives $c(V h_\pi(x)) \ge a_i\,(x_i - b_{\pi,i}(x)) > 0$. For $\alpha > 1$, every bundle $B \notin \mathcal{S}$, the empty bundle included, then satisfies
\[
v^\alpha_B\bigl(h_\pi(x)\bigr) = v_B\bigl(h_\pi(x)\bigr) \le c\bigl(V h_\pi(x)\bigr) < \alpha\, c\bigl(V h_\pi(x)\bigr) = c\bigl(V^\alpha h_\pi(x)\bigr),
\]
the equality on the right by Step 2. At the origin, $b_\pi(0) \in \operatorname{int} D \subseteq (0, \infty)^N$. Hence every coordinate of $h_\pi(0) = -b_\pi(0)$ is negative. Every nonempty bundle then has a negative value along $h_\pi(0)$, because its gradient is nonnegative and nonzero: its valuation is nondecreasing and nonconstant. The empty bundle alone therefore attains $c(V^\alpha h_\pi(0)) = 0$. Hence, at $\mu^+$-almost every source $x$, only bundles in $\mathcal{S}$ attain $c(V^\alpha h_\pi(x))$ if $x \ne 0$, and only the empty bundle attains it if $x = 0$.

\emph{Step 4: the selection at almost every type.} Let $M'$ be a menu that is optimal under $v^\alpha$, with indirect utility $W$. Then $W(0) = 0$ by the normalization of \Cref{app:duality-flow}. Hence, by \Cref{lem:revenue-invariance}(ii) and Step 2, $\int_D W\, \d\mu = \alpha R(M)$, the barycentric cost of $\pi$ under $v^\alpha$. In the chain
\begin{align*}
\int_D W\, \d\mu &= \int_D \Bigl[W(x) - \int W\, \d\pi_x\Bigr]\, \mu^+(\d x) \le \int_D \bigl[W(x) - W(b_\pi(x))\bigr]\, \mu^+(\d x)\\
&\le \int_D c\bigl(V^\alpha h_\pi(x)\bigr)\, \mu^+(\d x),
\end{align*}
the inequalities are Jensen's inequality and the bound on utility differences. The two ends both equal $\alpha R(M)$, and both slacks are nonnegative pointwise. Hence, for $\mu^+$-almost every source $x$, $W(x) - W(b_\pi(x)) = c(V^\alpha h_\pi(x))$ and $\int W\, \d\pi_x = W(b_\pi(x))$. Let $\mathcal{E}$ be the Lebesgue-null set of types in $\partial D$ or at which $W$ is not differentiable. Then $\mu^-(\mathcal{E}) = 0$, because $\mu^- \ll \lambda_D$. Hence $\pi_x(\mathcal{E}) = 0$ for $\mu^+$-almost every source $x$. Fix such a source and an option of $M'$ active at $x$, with affine utility $L(z) = \xi \cdot z - p$. As $W \ge L$ and $W(x) = L(x)$, $W(x) - W(b_\pi(x)) \le \xi \cdot h_\pi(x) \le c(V^\alpha h_\pi(x))$; both inequalities are equalities, and $W(b_\pi(x)) = L(b_\pi(x))$. Hence $\int (W - L)\, \d\pi_x = 0$ with $W - L \ge 0$, and $L$ is active at $\pi_x$-almost every sink. At such a sink $x' \notin \mathcal{E}$, every option active at $x'$ has the gradient $\nabla W(x') = \xi$. The lottery $\sigma'$ of the option that $x'$ takes therefore satisfies $\sum_{B \in \mathcal{F}_0} \sigma'(B)\, v^\alpha_B(h_\pi(x)) = \xi \cdot h_\pi(x) = c(V^\alpha h_\pi(x))$. Each $v^\alpha_B(h_\pi(x))$ is at most $c(V^\alpha h_\pi(x))$, and the weights sum to one. Hence every bundle in the support of $\sigma'$ attains $c(V^\alpha h_\pi(x))$. By Step 3, $\sigma'(\mathcal{S}) = 1$ if $x \ne 0$, and $\sigma' = \delta_\varnothing$ if $x = 0$. In the second case the option is the outside option, because $W \ge 0$ and prices are nonnegative. The sinks at which this conclusion fails therefore form a $\pi$-null set, hence a $\mu^-$-null set. This set is also Lebesgue-null, because the sinks of $\mu$ have a positive density almost everywhere.

Steps 1 to 4 prove part (ii). Steps 1 and 2 give the plan $\pi$, which attains the least barycentric cost at both premia. Step 3 shows that, along the direction $h_\pi(x)$ of a source $x$, only bundles in $\mathcal{S}$ attain its cost if $x \ne 0$, and only the empty bundle if $x = 0$. Step 4 shows that, in every menu that is optimal under $v^\alpha$, almost every type buys only bundles that attain this cost for a source matched with it. Every such menu therefore sells only bundles in $\mathcal{S}$.
\end{proof}

\section{Necessity through the dual}\label{os:S7-necessity}

In the setting of \Cref{sec:main}, the necessity theorems concern menus that are optimal in the unrestricted problem. Each proof below applies the dominance constraint of the transport dual \eqref{eq:dual-D} to the utility of one added option. For a menu $M$ with indirect utility $U$ and a lottery $q \in \mathcal{Q}$, define the \emph{supporting price} and the \emph{gap} of $q$ by
\[
p_q := \max_{x \in D}\bigl\{(V^\top q) \cdot x - U(x)\bigr\}, \qquad g_q := U - (V^\top q) \cdot x + p_q \;\ge\; 0 .
\]
The option $(q, p_q)$ is the cheapest that no buyer strictly prefers to $M$, and $g_q(x)$ is the utility that buyer $x$ loses by switching to it.

\begin{lemma}[Dual bound on the gap]\label{lem:S7-gap}
Let $M$ be optimal in the unrestricted problem, with indirect utility $U$, and let $q \in \mathcal{Q}$. Then
\[
\int_D (\varepsilon - g_q)^+\, \d\mu \;\le\; 0 \qquad \text{for every } \varepsilon \in (0, p_q] .
\]
\end{lemma}

\begin{proof}
By the normalization of \Cref{app:duality-flow}, $U(0) = 0$, so by \Cref{lem:revenue-invariance}(ii) the optimal revenue is $\int_D U\, \d\mu$. By \Cref{thm:strong-duality}, some transport plan $\gamma$ with $\gamma_1 - \gamma_2 \succeq_{\mathrm{mcvx}} \mu$ costs this revenue. The chain
\[
\int_D U\, \d\mu \;\le\; \int_D U\, \d(\gamma_1 - \gamma_2) \;=\; \int \bigl[U(x) - U(x')\bigr]\, \d\gamma(x, x') \;\le\; \int c_V(x - x')\, \d\gamma(x, x')
\]
holds with equality throughout: its first inequality is the dominance tested at $U$, its second is the bound $U(x) - U(x') \le c_V(x - x')$ of \Cref{sec:bundle-to-cone}, and its two ends are equal (the argument of \Cref{rmk:within-cell-support}). Hence
\begin{equation}\label{eq:S7-slack}
\int_D U\, \d(\gamma_1 - \gamma_2) = \int_D U\, \d\mu,
\qquad
U(x) - U(x') = c_V(x - x') \quad \text{for $\gamma$-almost every } (x, x') .
\end{equation}
Fix $\varepsilon \in (0, p_q]$ and let $w := (\varepsilon - g_q)^+$. The function $U + w = \max\{U,\ (V^\top q) \cdot x - p_q + \varepsilon\}$ is the indirect utility of $M$ with the option $(q, p_q - \varepsilon)$ added. It lies in $\mathcal{U}_D$: it is $u' \circ v$ for the convex function $u' := \max\{u_M,\ q \cdot y - p_q + \varepsilon\}$ on $\R^K$, whose subgradients lie in $\mathcal{Q}$, and it vanishes at $0$ because $g_q(0) = p_q \ge \varepsilon$. The dominance tested at $U + w$, less the first equality of \eqref{eq:S7-slack}, gives
\[
\int_D w\, \d\mu \;\le\; \int_D w\, \d(\gamma_1 - \gamma_2) \;=\; \int \bigl[w(x) - w(x')\bigr]\, \d\gamma(x, x') .
\]
At a pair with the second equality of \eqref{eq:S7-slack}, $g_q(x) - g_q(x') = c_V(x - x') - (V^\top q) \cdot (x - x') \ge 0$, because $c_V$ is the support function of $V^\top \mathcal{Q} \ni V^\top q$. Hence $w(x) \le w(x')$ at such a pair, and the last integral is nonpositive.
\end{proof}

By the decomposition of $\mu$ in \Cref{sec:main}, every $w$ with $w(0) = 0$ satisfies
\begin{equation}\label{eq:S7-increment}
\int_D w\, \d\mu = \sum_{i \in [N]} \bar v_i \int_{F_i} w f\, \d\sigma_i - \int_D w\, \Phi f\, \d x,
\end{equation}
with $\sigma_i$ the surface measure on the top face $F_i$. Each proof below takes an option at a price just below its supporting price and shows that the increment $w = (\varepsilon - g_q)^+$ has a positive integral with respect to $\mu$: its integral over the top faces, where $\mu$ has the density $\bar v_i f > 0$, exceeds its integral over the interior, where $\mu$ has the density $-\Phi f$. By \Cref{lem:S7-gap}, the menu is then not optimal. Throughout, $V_{\varnothing} := 0$, and option $\ell$ with lottery $q^\ell \in \mathcal{Q}$ has the gradient $V^\top q^\ell$.

\subsection{Proof of \texorpdfstring{\Cref{thm:grand-bundle-nec}}{the grand-bundle necessity theorem}}\label{app:necessity-grand}

The \emph{cell of the deterministic grand bundle} is the set of types that take an option whose lottery assigns the grand bundle with probability one. Above the lower threshold, I show that a menu in which this cell is null violates \Cref{lem:S7-gap} for the grand bundle. Mechanisms reduce to menus by the taxation principle.

\begin{proof}[Proof of \Cref{thm:grand-bundle-nec}]
Suppose, for contradiction, that $M$ is optimal in the unrestricted problem and that the cell is null, and let $U$ be its indirect utility. The five steps are the exposing direction, the supporting price and the null contact set, the contact set on the top faces (the key step, where the null cell becomes a geometric constraint), the vanishing thickness and the layer bound, and the dual obstruction.

\emph{Step 1: the exposing direction.} Some $a \in (0, \infty)^N$ satisfies
\begin{equation}\label{eq:S7-expose}
a \cdot V_{[N]} > a \cdot V_B \qquad \text{for every bundle } B \ne [N].
\end{equation}
For $\alpha \ge 1$, take $a := \mathbf{1}$. For $\alpha < 1$, write $k := \abs{C}$, set $a_i := 1$ on $C$, and choose positive peripheral coordinates with sum $r$ ($r := 0$ when $P$ is empty) such that $\alpha(k + r) > k - 1 + r$, which is possible because $\alpha k > k - 1$: by the formula for $\mathcal{Z}_\alpha$ in \Cref{sec:main-grand}, the inequality $\min_{i \in C} x_i \ge (1 - \alpha) \sum_{i \in [N]} x_i$ holds on an open set of interior types near the core diagonal when $\alpha k > k - 1$ and on a null set otherwise, so the hypothesis of the theorem is equivalent to $\alpha k > k - 1$. A bundle missing a core good has value at most $k - 1 + r$ at $a$; a proper bundle containing the core omits a positive peripheral coordinate and has value below $\alpha(k + r)$. The finitely many strict inequalities \eqref{eq:S7-expose} imply that $a$ lies in the interior of the bundle cone $K^{\mathcal{B}}_{[N]}$ of the grand bundle \eqref{eq:bundle-cone} and that $V_{[N]}$ is an exposed vertex of $\operatorname{conv}\{V_B : B \in \mathcal{B} \cup \{\varnothing\}\}$; hence a lottery $q$ with $V^\top q = V_{[N]}^\top$ is the deterministic grand bundle.

\emph{Step 2: the supporting price and the null contact set.} For the deterministic grand bundle $q := \delta_{[N]}$, write
\begin{gather*}
p^* := p_q = \max_{x \in D} \bigl\{V_{[N]} \cdot x - U(x)\bigr\},\\
g := g_q = U - V_{[N]} \cdot x + p^* \ge 0, \qquad \mathcal{K} := \{x \in D : g(x) = 0\}.
\end{gather*}
The maximum is attained, so $\mathcal{K}$ is nonempty and compact. The contact set is Lebesgue-null. At a differentiability point of $U$ in $\mathcal{K} \cap \operatorname{int} D$, $g$ attains its minimum, hence $\nabla U = V_{[N]}^\top$ and every option optimal there is the deterministic grand bundle by Step 1. As $U$ is differentiable almost everywhere (Rademacher's theorem), a contact set of positive measure would give that option a cell of positive measure. For $h \in K^{\mathcal{B}}_{[N]}$, every option satisfies $V^\top q^\ell \cdot h \le V_{[N]} \cdot h$; taking maxima gives, for $x, x + h \in D$,
\begin{equation}\label{eq:S7-monotone}
U(x + h) \le U(x) + V_{[N]} \cdot h, \qquad g(x + h) \le g(x).
\end{equation}
In particular $(\mathcal{K} + K^{\mathcal{B}}_{[N]}) \cap D \subseteq \mathcal{K}$.

\emph{Step 3: the contact set on the top faces.} The travel time to the boundary in direction $a$,
\[
\tau(x) := \min_{i \in [N]} \frac{\bar v_i - x_i}{a_i},
\]
is continuous and vanishes if and only if $x$ lies on a top face. I claim $\tau = 0$ on $\mathcal{K}$. If $x \in \mathcal{K}$ and $0 < t < \tau(x)$, then $x + ta \in \operatorname{int} D$ and $ta \in \operatorname{int} K^{\mathcal{B}}_{[N]}$. A small open ball around $ta$ then lies in $K^{\mathcal{B}}_{[N]}$ with its translate by $x$ in $D$, and by \eqref{eq:S7-monotone} this translate lies in $\mathcal{K}$, contradicting Step 2. Hence $\mathcal{K} \subseteq \bigcup_i F_i$, so $0 \notin \mathcal{K}$ and $p^* = g(0) > 0$.

\emph{Step 4: the vanishing thickness and the layer bound.} The thickness
\[
\delta_\varepsilon := \max\bigl\{\tau(x) : x \in D,\ g(x) \le \varepsilon\bigr\}
\]
tends to zero as $\varepsilon \downarrow 0$: otherwise points $x_n$ with $g(x_n) \le \varepsilon_n \downarrow 0$ and $\tau(x_n)$ bounded away from zero have a limit point in $\mathcal{K}$ with $\tau > 0$, contradicting Step 3. For $0 < \varepsilon < p^*$, the increment $w_\varepsilon := (\varepsilon - g)^+$ satisfies $w_\varepsilon(0) = 0$ and is nondecreasing along lines in direction $a$ by \eqref{eq:S7-monotone}. Up to a null set of lines, each $x \in D$ is uniquely $x = y - ta$ with $y \in F_i$ the terminal point in direction $a$ and $0 \le t \le \ell(y) := \min_j y_j/a_j$; this parameterization has the Jacobian $a_i$ on $F_i$, and $t = \tau(x)$, so $w_\varepsilon(x) > 0$ implies $t \le \delta_\varepsilon$. With the face mass $J_\varepsilon := \sum_i a_i \int_{F_i} w_\varepsilon\, \d\sigma_i$,
\begin{equation}\label{eq:S7-layer}
\int_D w_\varepsilon\, \d x = \sum_{i \in [N]} a_i \int_{F_i} \int_0^{\ell(y)} w_\varepsilon(y - ta)\, \d t\, \d\sigma_i(y) \le \delta_\varepsilon J_\varepsilon .
\end{equation}

\emph{Step 5: the dual obstruction.} Let
\[
m := \min_D f > 0, \qquad L := \sup_D \abs{\Phi f} < \infty, \qquad c := m \min_{i \in [N]} \frac{\bar v_i}{a_i} > 0 .
\]
By \eqref{eq:S7-increment} and \eqref{eq:S7-layer},
\[
\int_D w_\varepsilon\, \d\mu \;\ge\; (c - L\delta_\varepsilon)\, J_\varepsilon .
\]
Moreover $J_\varepsilon > 0$: a point of $\mathcal{K}$ lies on some $F_i$ and has $w_\varepsilon = \varepsilon$, so by continuity $w_\varepsilon > 0$ on a relative neighborhood of positive surface measure in $F_i$. By Step 4, $c - L\delta_\varepsilon > 0$ for every small $\varepsilon$, so $\int_D w_\varepsilon\, \d\mu > 0$, contradicting \Cref{lem:S7-gap} at $q = \delta_{[N]}$ and $\varepsilon < p^*$.

By Steps 1 to 4, the gap of the grand bundle is small only in a layer of vanishing thickness below the top faces, where the face sources of $\mu$ exceed its interior sinks. By Step 5, this contradicts the dual bound. Hence in every optimal menu the deterministic grand bundle has a cell of positive measure.
\end{proof}

\subsection{Proof of \texorpdfstring{\Cref{thm:incl-necessity}}{the inclusivity necessity theorem}}\label{app:necessity-incl}

\begin{proof}[Proof of \Cref{thm:incl-necessity}]
Fix $\alpha > 0$ and suppose, for contradiction, that $M_{\mathcal{F}}$ is optimal in the unrestricted problem at $\alpha$, with indirect utility $U$. Fix a good $i$ whose top face meets the exclusion cell $\Omega_0$, which exists because $M_{\mathcal{F}}$ is not \hyperref[def:incl]{inclusive}. Let $V_i := \alpha \bar v_i$ if $C \subseteq \{i\}$ and $V_i := \bar v_i$ otherwise, and let $q := \delta_{\{i\}}$, with gradient $(V_i/\bar v_i)\, e_i$. The supporting price of $q$ is $p_q = V_i$, because $U \ge 0$ and $U = 0$ on $F_i \cap \Omega_0 \ne \varnothing$; its gap is $g_q = U - V_i x_i/\bar v_i + V_i$. For $0 < \varepsilon < V_i$, the increment $w := (\varepsilon - g_q)^+$ satisfies $0 \le w \le \varepsilon$ and $w(0) = 0$. It is positive only on the slab $\{x_i > \bar v_i - \varepsilon \bar v_i/V_i\}$, of volume $\varepsilon (\bar v_i/V_i) \prod_{j \ne i} \bar v_j$. Under seller-favorable ties and positive prices, the exclusion cell
\[
\Omega_0 = \bigl\{x \in D : (V^\top q^\ell) \cdot x - p_\ell < 0 \text{ for every option } \ell \ge 1\bigr\}
\]
is relatively open in $D$, because the menu is finite. Hence $F_i \cap \Omega_0$ is a nonempty relatively open subset of $F_i$, of positive surface measure, where $U = 0$ and $w = \varepsilon$. The other face terms of \eqref{eq:S7-increment} are nonnegative. Hence
\[
\int_D w\, \d\mu \;\ge\; \bar v_i \Bigl(\min_D f\Bigr)\, \sigma_i(F_i \cap \Omega_0)\, \varepsilon \;-\; \Bigl(\sup_D \abs{\Phi f}\Bigr)\, \frac{\bar v_i}{V_i} \prod_{j \ne i} \bar v_j\, \varepsilon^2 ,
\]
which is positive for small $\varepsilon > 0$ and contradicts \Cref{lem:S7-gap}.
\end{proof}

\section{Strict core regularity from Myersonian regularity}\label{os:single-cell}

This section derives the order of \hyperref[def:align-reg]{strict core regularity}, the hypothesis of \Cref{thm:core-bundle}, from one-dimensional Myersonian regularity of the core sum. The setting is the core $C$, a density $\rho$ on its box and a cutoff $q$ that stands for the crossing point. \Cref{os:sc-cert} shows that above the cutoff the top-face sources dominate the sinks in the unbundling scores of \Cref{sec:main-core}, computed for the core alone. Pure bundling applies the theorem once, at the crossing point (\Cref{os:apps-puresec}), and the two-tier menu applies it at every cutoff up to the crossing point of the core (\Cref{os:apps-twotier}). In both, nesting (\Cref{app:core}) then gives strict core regularity.

Throughout, $D_C:=\prod_{i\in C}[0,\bar v_i]$ is the box of the core, and $F_c^C:=\{\omega\in D_C:\omega_c=\bar v_c\}$ for $c\in C$ are its top faces, with surface measures $\sigma_c$. The \emph{core sum} $t(\omega):=\sum_{i\in C}\omega_i$ has value $\bar t_C:=\sum_{i\in C}\bar v_i$ at the top corner and least value $\underline t_C:=\min_{i\in C}\bar v_i$ on the top faces. For a positive $C^1$ function $\rho$ on $\operatorname{int}D_C$, the \emph{core radial score} is $\Phi_\rho:=|C|+\omega\cdot\nabla\log\rho$, and $\Phi_\rho\rho=\sum_{c\in C}\partial_c(\omega_c\rho)$.
\begin{definition}[Admissible core density]\label{def:adm-core}
A density $\rho$ on $D_C$ is \emph{admissible} if:
\begin{enumerate}[label=\textnormal{(\roman*)},nosep]
\item $\rho\in C^1(\operatorname{int}D_C)$, $\rho>0$ on $\operatorname{int}D_C$, and $\rho$ and every $\omega_c\,\partial_c\rho$ are integrable;
\item $\rho$ extends continuously to $D_C$ off the lower faces, its top-face traces are integrable, and $\omega_c\,\rho(\omega)\to0$ as $\omega_c\downarrow0$ for each $c\in C$ and each fixed choice of positive values of the other coordinates;
\item the core sum has a continuous, strictly positive density $\gamma$ on $(0,\bar t_C)$, with $r\,\gamma(r)\to0$ as $r\downarrow0$;
\item $\rho$ is strictly positive and $\Phi_\rho\,\rho$ is bounded on a neighborhood of the top corner.
\end{enumerate}
\end{definition}

For admissible $\rho$, $\Phi_\rho\rho$ is integrable by condition (i). With $\bar F(r):=\int_{\{t\ge r\}}\rho\,\d\omega$ the survivor of the core sum, its \emph{virtual-value statistic}, as in \Cref{sec:apps-pure}, is
\[
H(r):=\frac{r\,\gamma(r)}{\bar F(r)}\quad(0<r<\bar t_C),\qquad H(0):=0,
\]
continuous on $[0,\bar t_C)$ by condition (iii). For a cutoff $q\in[0,\underline t_C)$, set $\Omega_q:=\{t\ge q\}$ and define the \emph{cutoff measures}
\[
\nu^+_q:=\sum_{c\in C}\bar v_c\,\rho\big|_{F_c^C},
\qquad
\nu^-_q:=\bigl(\Phi_\rho+H(q)\bigr)\rho\,\mathbf{1}_{\Omega_q}\,\d\omega .
\]
The source measure is the same for every $q$, because $t\ge\underline t_C>q$ on every top face. At slack $\lambda\in[0,1]$, the \emph{unbundling scores of the core} are the coordinates of
\[
Q^\lambda_C(\omega):=\Bigl(\lambda\sum_{i\in T}\omega_i+(1-\lambda)\,t(\omega)\Bigr)_{\varnothing\ne T\subseteq C} .
\]
With no peripheral goods, $Q^\lambda_C$ is the option score of \Cref{sec:main-core} at slack $\lambda$.

\begin{theorem}[Single-cell dominance]\label{os:sc-cert}
Let $\rho$ be admissible.
\begin{enumerate}[label=\textnormal{(\roman*)}]
\item Suppose that $q\in[0,\underline t_C)$ satisfies
\[
H(r)>H(q)\quad(q<r<\bar t_C),
\qquad
\Phi_\rho+H(q)\ge0\quad\text{on }\operatorname{int}D_C\cap\Omega_q .
\]
Then $(Q^\lambda_C)_\#\nu^-_q\preceq_{\mathrm{FOSD}}(Q^\lambda_C)_\#\nu^+_q$ for some $\lambda\in(0,1]$.
\item Let $q_*<\underline t_C$ satisfy the hypotheses of (i). If
\begin{equation}\label{eq:sc-family-hypotheses}
\begin{aligned}
H(q)&\le H(q_*) &&(0\le q\le q_*),\\
\Phi_\rho+H(q)&\ge0\quad\text{on }\operatorname{int}D_C\cap\Omega_q &&(0\le q\le q_*),
\end{aligned}
\end{equation}
then the conclusion of (i) holds at every $q\in[0,q_*]$ with one $\lambda\in(0,1]$.
\end{enumerate}
\end{theorem}

\begin{proof}
The proof constructs, through the quantiles of the core sum, a coupling of $(\nu^+_q,\nu^-_q)$ along which the unbundling scores of the core are ordered. It has four steps: the tail identity; the uniform shift, the key step, because one positive shift must hold up to the top corner, where both tails vanish; the coupling by quantiles and the lift, which gives the dominance; and one shift for all lower cutoffs. Steps 1 to 3 prove (i), and Step 4 with Step 3 at each lower cutoff proves (ii).

\emph{Step 1: the tail identity.} Interpret $r\,\gamma(r)$ as zero at $r=0$. The tails
\[
B(r):=\sum_{c\in C}\bar v_c\int_{F_c^C\cap\{t\ge r\}}\rho\,\d\sigma_c ,
\qquad
I_q(r):=\nu^-_q\{t\ge r\}
\]
satisfy
\begin{equation}\label{eq:sc-tail-gap}
B(r)-I_q(r)=r\,\gamma(r)-H(q)\,\bar F(r)=\bar F(r)\,\bigl[H(r)-H(q)\bigr]\qquad(q\le r<\bar t_C).
\end{equation}
At $r=q$ the gap vanishes, so the cutoff measures have equal mass $I_q(q)=B(0)$, positive by condition (iv), and under the hypotheses of (i) the source tail strictly exceeds the sink tail on $(q,\bar t_C)$. The identity follows from integration by parts along each coordinate: for every Lipschitz $\psi$ on $D_C$,
\begin{equation}\label{eq:admissible-ibp}
\int_{D_C}(\omega\cdot\nabla\psi)\,\rho\,\d\omega
=\sum_{c\in C}\bar v_c\int_{F_c^C}\psi\,\rho\,\d\sigma_c
-\int_{D_C}\psi\,\Phi_\rho\,\rho\,\d\omega .
\end{equation}
For almost every $\omega_{-c}$, $\omega_c\mapsto\omega_c\rho$ is absolutely continuous on the closed fiber, by condition (i), Fubini's theorem and the endpoint limits of condition (ii); by the integrability in conditions (i) and (ii), Fubini's theorem sums the one-dimensional integrations by parts. Applying \eqref{eq:admissible-ibp} to $\psi=(t-r)^+$ and differentiating in $r\in(0,\bar t_C)$ gives
\begin{equation}\label{eq:sc-divergence}
\int_{\{t\ge r\}}\Phi_\rho\,\rho\,\d\omega=B(r)-r\,\gamma(r).
\end{equation}
Dominated convergence applies because no level set $\{t=r\}$ with $0<r<\bar t_C$ is charged by Lebesgue measure or by a top-face trace, which is an integrable density on a face of dimension $|C|-1$, or the point mass at $\bar t_C$ when $|C|=1$. At $r=0$, \eqref{eq:sc-divergence} is \eqref{eq:admissible-ibp} at $\psi=1$. For $r\ge q$, adding $H(q)\bar F(r)$ and using $r\gamma(r)=H(r)\bar F(r)$ gives \eqref{eq:sc-tail-gap}.

\emph{Step 2: the uniform shift.} With the shift $T_\lambda(r):=(1-\lambda)\,r+\lambda\,\bar t_C$, for some $\lambda\in(0,1/2]$,
\begin{equation}\label{eq:sc-shifted-tail}
I_q(r)\le B\bigl(T_\lambda(r)\bigr)\qquad(q\le r\le\bar t_C).
\end{equation}
Set $q_0:=(q+\underline t_C)/2$; three intervals each give a bound on $\lambda$. By conditions (ii) and (iv), the sink density is bounded above and the face densities are bounded below by positive constants on a closed corner neighborhood. For $s\le\min_i\bar v_i$, the set $\{t\ge\bar t_C-s\}$ is a simplex of volume $s^{|C|}/|C|!$ that meets each top face in volume $s^{|C|-1}/(|C|-1)!$. Hence, for some $k_1,k_2>0$, all small $s>0$ and every $\lambda\le1/2$,
\[
I_q(\bar t_C-s)\le k_1s^{|C|}\le k_2s^{|C|-1}\le B(\bar t_C-s/2)\le B\bigl(T_\lambda(\bar t_C-s)\bigr),
\]
the last inequality because $T_\lambda(\bar t_C-s)\le\bar t_C-s/2$ and $B$ is nonincreasing; this is \eqref{eq:sc-shifted-tail} on $[\bar t_C-r_0,\bar t_C]$ for some $r_0\in(0,\bar t_C-q_0)$. On $[q,q_0]$, any $\lambda>0$ with $T_\lambda(q_0)\le\underline t_C$ gives $B(T_\lambda(r))=B(0)=I_q(q)\ge I_q(r)$, because no source lies below $\underline t_C$ and the sink measure is nonnegative. On $[q_0,\bar t_C-r_0]$, the gap in \eqref{eq:sc-tail-gap} is continuous and positive, with a positive minimum $m$. $B$ is uniformly continuous, because the top-face traces charge no level set of $t$ (Step 1), and $T_\lambda(r)-r\le\lambda\bar t_C$; a smaller $\lambda$ therefore gives $B(T_\lambda(r))>B(r)-m\ge I_q(r)$ there. The least of the three bounds on $\lambda$ gives \eqref{eq:sc-shifted-tail}.

\emph{Step 3: the coupling by quantiles and the lift.} By \eqref{eq:sc-shifted-tail}, the core sum under $\nu^+_q$ first-order stochastically dominates $T_\lambda$ of the core sum under $\nu^-_q$, and both measures have mass $B(0)$. The disintegration theorem gives a probability kernel $K^+(r,\d y)$ on $\{t=r\}$ that disintegrates $\nu^+_q$ along $t$, and a kernel $K$ with $\rho\,\d\omega=\int_0^{\bar t_C}K(r,\d\omega)\,\gamma(r)\,\d r$ that is concentrated on $\{t=r\}$ and integrates $\abs{\Phi_\rho}$ for almost every $r$. Let $a(q,r):=\int(\Phi_\rho+H(q))^+K(r,\d\omega)$, and let $K^-_q(r,\cdot)$ be $K(r,\cdot)$ reweighted by $(\Phi_\rho+H(q))^+/a(q,r)$; on the exceptional null set and where $a(q,r)=0$, the kernels are the point mass at $(r/\bar t_C)(\bar v_i)_{i\in C}$. Under the hypothesis $\Phi_\rho+H(q)\ge0$ on $\operatorname{int}D_C\cap\Omega_q$, $\nu^-_q=\int_q^{\bar t_C}K^-_q(r,\cdot)\,a(q,r)\,\gamma(r)\,\d r$. With $F_+^{-1}$ and $F_q^{-1}$ the quantile functions of the core sum under $\nu^+_q/B(0)$ and $\nu^-_q/B(0)$, the measure
\[
\pi^q(\d y,\d x):=B(0)\int_0^1K^+\bigl(F_+^{-1}(u),\d y\bigr)\,K^-_q\bigl(F_q^{-1}(u),\d x\bigr)\,\d u
\]
has marginals $\nu^+_q$ and $\nu^-_q$, and $t(y)\ge T_\lambda(t(x))$ for $\pi^q$-almost every $(y,x)$, because $F_+^{-1}(u)\ge T_\lambda(F_q^{-1}(u))$ for almost every $u$ and the kernels are concentrated on level sets of $t$. For such a pair, $h:=y-x$ and $\delta:=\bar t_C-t(x)$ satisfy $t(h)\ge\lambda\delta$ and, because $x_i\le\bar v_i$, $\sum_ih_i^-\le\sum_i(\bar v_i-y_i)=\bar t_C-t(y)\le(1-\lambda)\delta$. Hence $\lambda\sum_ih_i^-\le(1-\lambda)t(h)$, and every unbundling coordinate is ordered:
\[
\lambda\sum_{i\in T}h_i+(1-\lambda)\,t(h)\;\ge\;(1-\lambda)\,t(h)-\lambda\sum_ih_i^-\;\ge\;0\qquad(\varnothing\ne T\subseteq C).
\]
That is, $Q^\lambda_C(y)\ge Q^\lambda_C(x)$ for $\pi^q$-almost every pair. Integrating a bounded Borel coordinatewise nondecreasing function of $Q^\lambda_C$ along $\pi^q$ gives the dominance in (i).

\emph{Step 4: one shift for all lower cutoffs.} Under \eqref{eq:sc-family-hypotheses}, the $\lambda$ of Step 2 at $q_*$ satisfies \eqref{eq:sc-shifted-tail} at every $q\in[0,q_*]$. The source measures coincide and, by the second line of \eqref{eq:sc-family-hypotheses}, the sink measures are nonnegative with mass $B(0)$. For $r\ge q_*$, \eqref{eq:sc-tail-gap} and $H(q)\le H(q_*)$ imply $I_q(r)\le I_{q_*}(r)\le B(T_\lambda(r))$. At $r=q_*$, \eqref{eq:sc-shifted-tail} and $B\le B(0)=I_{q_*}(q_*)$ give $B(T_\lambda(q_*))=B(0)$. As $I_q$ is nonincreasing, for $q\le r<q_*$,
\[
I_q(r)\le I_q(q)=B(0)=B\bigl(T_\lambda(q_*)\bigr)\le B\bigl(T_\lambda(r)\bigr).
\]

With this $\lambda$, Step 3 at each $q\in[0,q_*]$ gives (ii).
\end{proof}

\begin{remark}[Admissible singular densities]\label{os:sc-extension}
Smooth positive densities and the iid $\mathrm{Beta}(\beta,1)$ densities are admissible. Admissibility of the full type density gives the revenue identity with finite sources and sinks. With this identity and a positive radial score, the sufficiency argument holds for every core, and the selection of \Cref{cor:premium}(ii) holds.

\emph{Argument.} A smooth positive density under the standing assumptions satisfies every condition. The iid $\mathrm{Beta}(\beta,1)$ density on $[0,1]^C$, for $\beta>0$, has $\Phi_\rho=|C|\beta$, $\omega_c\,\partial_c\rho=(\beta-1)\rho$, $\omega_c\rho\to0$ on every lower face and integrable top-face traces, and is smooth and positive near the top corner. The density $\gamma$ of its sum is positive and continuous on $(0,|C|)$, by splitting the convolution integrals away from their integrable endpoint singularities, and proportional to $r^{|C|\beta-1}$ near zero, so $r\gamma(r)\to0$. Let the full type density $f$ be admissible with $[N]$ in place of $C$, with $\int_Df=1$ and the radial score $\Phi=N+1+x\cdot\nabla\log f$ positive on $\operatorname{int}D$, and let the core $C\subseteq[N]$ be arbitrary. For linear valuations and any density, the envelope formula gives $x\cdot\nabla U=q(x)\cdot v(x)$ at every differentiability point of the indirect utility $U$, so revenue is $\int_D(x\cdot\nabla U-U)f\,\d x$. Applying \eqref{eq:admissible-ibp} on $D$ to the Lipschitz function $U$ gives
\[
\int_D(x\cdot\nabla U-U)f\,\d x
=\sum_j\bar v_j\int_{F_j}Uf\,\d\sigma_j-\int_DU\,\Phi f\,\d x
=\int_DU\,\d\mu-U(0),
\]
where $\mu:=\delta_0+\sum_j\bar v_j\,f|_{F_j}-\Phi f\cdot\lambda_D$. Its sources and sinks are finite, because the top-face traces are integrable by condition (ii) and $\Phi f=\sum_j\partial_j(x_jf)+f$ is integrable by condition (i); \eqref{eq:admissible-ibp} at $U=1$ gives $\mu(D)=0$. With this identity, the bundle-cone certificate (\Cref{prop:saddle}) and the proofs of \Cref{cor:premium}(ii) and \Cref{thm:core-bundle} hold for $f$ and every core $C$. The applications (\Cref{os:apps-puresec}, \Cref{os:apps-twotier}) use it for the Beta family and the singular marginal of the peripheral good.
\end{remark}

\section{Proofs for the applications}\label{os:apps}

\subsection{Pure bundling}\label{os:apps-puresec}

Here $C=[N]$ and $t(x):=\sum_ix_i$. With $\rho:=f$, $\gamma$, $\bar F$ and $H(r):=r\gamma(r)/\bar F(r)$ are the density, the survivor and the virtual-value statistic of $t$ (\Cref{os:single-cell}), and $\Phi=\Phi_\rho+1$. The menu $M$ sells the grand bundle at $p^*=\alpha\hat s$, where $\hat s$ maximizes $s\bar F(s)$, the revenue of the price $\alpha s$ divided by $\alpha$, and its selling cell under seller-favorable ties is $\Omega_1=\{t\ge\hat s\}$. It is strictly inclusive if $\hat s<\underline t_C:=\min_i\bar v_i$; write $\bar t_C:=\sum_i\bar v_i$. The maximizer $\hat s$ is interior, because $s\bar F(s)$ vanishes at $0$ and at $\bar t_C$ and is positive between them. It satisfies the first-order condition $\hat s\,\gamma(\hat s)=\bar F(\hat s)$, equivalent to $H(\hat s)=1$, because $\gamma$ is continuous and strictly positive on $(0,\bar t_C)$.

\begin{proof}[Proof of \Cref{thm:pure-bundle}]
The proof has four parts: sufficiency, necessity above the least top value, necessity at the boundary $\hat s=\underline t_C$, and uniqueness above the upper threshold. Both necessity parts go through the dual: Part 2 is a case of \Cref{thm:incl-necessity}, and Part 3 applies \Cref{lem:S7-gap} to a good sold alone. The boundary part is the key one, because there the utility increment is positive only near a corner of a top face and the gain, of order $\varepsilon^N$, is found from its leading term.

\emph{Part 1: sufficiency.} Suppose $\hat s<\underline t_C$. Every top face lies in $\Omega_1$, because $t\ge\bar v_j>\hat s$ on $F_j$, and the origin lies in the exclusion cell. With $q=\hat s$, the cutoff measures of \Cref{os:single-cell} are therefore the cell's sources and sinks, the sink identity by $H(\hat s)=1$:
\[
\nu^+_{\hat s}=\sum_j\bar v_j\,f|_{F_j}=\mu^+_1,
\qquad
\nu^-_{\hat s}=\bigl(\Phi_\rho+H(\hat s)\bigr)f\,\mathbf{1}_{\Omega_1}\,\d x=\mu^-_1.
\]
The hypotheses of \Cref{os:sc-cert} hold: $f$ is admissible (\Cref{os:sc-extension}), $\hat s\in[0,\underline t_C)$, $H(r)>H(\hat s)$ for $\hat s<r<\bar t_C$ by strict Myersonian regularity, and $\Phi_\rho+H(\hat s)=\Phi>0$. \Cref{os:sc-cert}(i) is the dominance $\nu^-_{1,\lambda}\preceq_{\mathrm{FOSD}}\nu^+_{1,\lambda}$ at some slack $\lambda\in(0,1]$, because $Q^\lambda_C$ is the option score of the grand bundle at slack $\lambda$ (\Cref{os:single-cell}). By nesting (\Cref{app:core}), the pair is therefore \hyperref[def:align-reg]{strictly core-regular}. The menu is finite, has a positive price, and is \hyperref[def:incl]{inclusive} because every top face lies in $\Omega_1$. \Cref{thm:core-bundle} gives optimality at every sufficiently large premium.

\emph{Part 2: necessity above the least top value.} Suppose $\bar v_i<\hat s$ for some good $i$. The corner $\bar v_ie_i$ of the top face $F_i$ has $t=\bar v_i<\hat s$ and therefore lies in the exclusion cell. Hence pure bundling at $p^*$ is not \hyperref[def:incl]{inclusive}, and by \Cref{thm:incl-necessity} it is not optimal at any premium.

\emph{Part 3: necessity at the boundary.} Suppose $\bar v_i=\hat s=\underline t_C$ for some good $i$, fix $\alpha>0$, and suppose that pure bundling at $p^*$ is optimal at $\alpha$, with the indirect utility $U_\alpha=\alpha(t-\hat s)^+$. Let $x^*:=\hat s e_i$ be the corner of $F_i$ at which the other coordinates vanish, and let $q:=\delta_{\{i\}}$ be good $i$ alone, with the gradient $e_i$. Its supporting price is $p_q=\bar v_i$, because $x_i-U_\alpha(x)\le\bar v_i$ on $D$, with equality at $x^*$. For $0<\varepsilon<\bar v_i$, the increment of \Cref{lem:S7-gap},
\[
w_\varepsilon:=(\varepsilon-g_q)^+=\bigl[x_i-\bar v_i+\varepsilon-\alpha\bigl(t(x)-\hat s\bigr)^+\bigr]^+,
\]
is the utility gain from adding good $i$ alone at the price $\bar v_i-\varepsilon$, and $0\le w_\varepsilon\le\varepsilon$. I show that $\int_D w_\varepsilon\,\d\mu>0$ for small $\varepsilon$, which contradicts the lemma. With $\delta:=\hat s-x_i$ and $z:=\sum_{j\ne i}x_j$,
\[
U_\alpha(x)=\alpha(z-\delta)^+,
\qquad
w_\varepsilon(x)=\bigl[\varepsilon-\delta-\alpha(z-\delta)^+\bigr]^+,
\]
and $w_\varepsilon(x)>0$ implies
\[
0\le\delta<\varepsilon,
\qquad
0\le z<\delta+\frac{\varepsilon-\delta}{\alpha}\le\varepsilon\max\{1,\alpha^{-1}\}.
\]
For fixed $\alpha$, the set $\{w_\varepsilon>0\}$ is thus a neighborhood of $x^*$ of volume $O_\alpha(\varepsilon^N)$. The interior term of \eqref{eq:S7-increment} is then $O_\alpha(\varepsilon^{N+1})$, because $0\le w_\varepsilon\le\varepsilon$ and $\Phi f$ is bounded, and for small $\varepsilon$ the set meets no top face other than $F_i$. On $F_i$, $w_\varepsilon(\hat s,x_{-i})=\bigl(\varepsilon-\alpha\sum_{j\ne i}x_j\bigr)^+$, which vanishes outside a simplex contained in $F_i$ for small $\varepsilon$. Substituting $x_{-i}=(\varepsilon/\alpha)y$ and expanding $f$ at $x^*$ gives
\begin{align*}
\int_D w_\varepsilon\,\d\mu
&=\hat s\int_{\{x_{-i}\ge0,\ \sum_{j\ne i}x_j\le\varepsilon/\alpha\}}\Bigl(\varepsilon-\alpha\sum_{j\ne i}x_j\Bigr)f(\hat s,x_{-i})\,\d x_{-i}+O_\alpha(\varepsilon^{N+1})\\
&=\frac{\hat s\,f(x^*)\,\varepsilon^N}{\alpha^{N-1}}\int_{S_{N-1}}\Bigl(1-\sum_jy_j\Bigr)\d y+O_\alpha(\varepsilon^{N+1})
=\frac{\hat s\,f(x^*)}{N!\,\alpha^{N-1}}\,\varepsilon^N+O_\alpha(\varepsilon^{N+1}),
\end{align*}
where $S_{N-1}:=\{y\in\R_+^{N-1}:\sum_jy_j\le1\}$ and the simplex integral, by its level-set volumes, is $\int_0^1(1-u)^{N-1}\d u/(N-1)!=1/N!$. The leading coefficient is positive because $f(x^*)>0$. Hence $\int_D w_\varepsilon\,\d\mu>0$ for small $\varepsilon$, which contradicts \Cref{lem:S7-gap}.

\emph{Part 4: uniqueness above the upper threshold.} Suppose $\hat s<\underline t_C$, and let $\alpha$ exceed the upper threshold of pure bundling, which is finite by Part 1. By \Cref{thm:core-bundle} at $C=[N]$, where $\mathcal{F}_C=\{[N]\}$, every optimal menu at $\alpha$ is core-peripheral: almost every type takes the outside option or an option that allocates the grand bundle with probability one. Some option allocates the grand bundle with probability one: otherwise almost every type takes the outside option, and the menu earns zero, less than the optimal revenue. Let $\alpha s$ be the least price of such an option, which exists because the menu is compact. By incentive compatibility, every buyer of the grand bundle pays $\alpha s$: almost every type with $t>s$ buys it, and almost every type with $t<s$ does not. The tie set $\{t=s\}$ is null. Hence the menu earns $\alpha s\bar F(s)$, and $s$ maximizes $s\bar F(s)$. Its only maximizer is $\hat s$. The ranges $(\hat s,\infty)$ and $[0,\underline t_C]$ cover $[0,\infty)$, because $\hat s<\underline t_C$. On $(\hat s,\infty)$, $s\bar F(s)<\hat s\bar F(\hat s)$: strict Myersonian regularity gives $(s\bar F)'=\bar F\,(1-H)<0$ on $(\hat s,\bar t_C)$, and $\bar F=0$ on $[\bar t_C,\infty)$. On $[0,\underline t_C]$, \eqref{eq:S7-increment} at $w=(t-s)^+$, whose left side is $s\bar F(s)$ by \Cref{lem:revenue-invariance}(ii), gives
\[
s\bar F(s)=\sum_j\bar v_j\int_{F_j}(t-s)\,f\,\d\sigma_j-\int_D(t-s)^+\,\Phi f\,\d x,
\]
because $t\ge\bar v_j\ge s$ on $F_j$. The first term is affine in $s$, and the subtracted integral is strictly convex in $s$ because $\Phi f>0$. Hence $s\bar F(s)$ is strictly concave on $[0,\underline t_C]$, and its only maximizer there is the critical point $\hat s$. Therefore $s=\hat s$, and the menu gives almost every type the lottery and the price of pure bundling at $p^*$: the grand bundle at $p^*$ if $t>\hat s$, and nothing otherwise. Hence pure bundling at $p^*$ is uniquely optimal.

Parts 1 to 3 prove the equivalence in \Cref{thm:pure-bundle}: Part 1 proves sufficiency, and Parts 2 and 3 prove that pure bundling at $p^*$ is strictly suboptimal at every premium when $\hat s\ge\underline t_C$. Part 4 proves that it is then uniquely optimal above its upper threshold.
\end{proof}

\begin{remark}[Weak core regularity in pure bundling]
Let the pure-bundling menu be strictly inclusive. Then $(f,M)$ is \hyperref[def:weak-core]{weakly core-regular} if and only if $H(r)\ge1$ for $\hat s\le r<\bar t_C$.

At $C=[N]$ the option score is $t$, and by Part 1 the cutoff measures at $q=\hat s$ are $\mu^+_1$ and $\mu^-_1$, with $H(\hat s)=1$. The tail identity \eqref{eq:sc-tail-gap} becomes
\[
\mu^+_1(\{t\ge r\})-\mu^-_1(\{t\ge r\})=\bar F(r)\,[H(r)-1]
\qquad(\hat s\le r<\bar t_C).
\]
The two measures have equal mass and their tails agree below $\hat s$. The tail comparison of weak core regularity therefore holds if and only if $H\ge1$ on $[\hat s,\bar t_C)$, where $\bar F>0$.
\end{remark}

\subsection{Two-tier menus}\label{os:apps-twotier}

This subsection proves \Cref{thm:twotier} in the setting (S1)--(S3) of \Cref{sec:apps-twotier}, with the notation of \Cref{os:single-cell} at $\bar v_i=1$; here $\bar t_C=|C|$ and $\underline t_C=1$. In the coordinates $(\omega,z)$ of \Cref{sec:apps-twotier}, the bare-core option has the marginal-inclusion vector $(\mathbf{1}_C,0)$, and the grand-bundle option has $(\mathbf{1}_C,1)$. At slack $\varepsilon$ their option scores are therefore
\[
Q^\varepsilon_1(\omega,z)=\bigl(Q^\varepsilon_C(\omega),\,0,\,-z\bigr),
\qquad
Q^\varepsilon_2(\omega,z)=\bigl(Q^\varepsilon_C(\omega)+(1-\varepsilon)\,z\,\mathbf{1},\,z,\,0\bigr),
\]
with $Q^\varepsilon_C$ the vector of unbundling scores of the core (\Cref{os:single-cell}) and $\mathbf{1}$ the vector of ones. In a slice of fixed $z$, the unbundling coordinates of both scores are $Q^\varepsilon_C(\omega)$ plus a constant, and their peripheral coordinates are constant.

Write $G_C=1-\bar G_C$ for the distribution function of the core sum under $f_C$, $F_i^C:=\{\omega_i=1\}$ for the top faces of $D_C=[0,1]^C$, and $\Gamma_C:=\sum_{i\in C} f_i(1)$. The paragraph after the proof derives four facts from (S1) and (S2): (i) $f_C$ is an \hyperref[def:adm-core]{admissible core density} with $\Phi_\rho=\Phi_C$; (ii) $H_C$ is continuous on $[0,\bar t_C)$, vanishes at zero and is strictly increasing on $(0,\bar t_C)$; (iii) the core sum is new better than used: $\bar G_C(q+r)\le\bar G_C(q)\,\bar G_C(r)$ for $q,r\ge0$, with $\bar G_C=0$ beyond $\bar t_C$; (iv) $f=f_C\otimes f_p$ is admissible in the sense of \Cref{os:sc-extension}. The boundaries of the cells $\Omega_0,\Omega_1,\Omega_2$ of \Cref{sec:apps-twotier} lie in $\{t=\hat s_C\}$, $\{z=d\}$ and $\{t=\tau(z)\}$, which are null for Lebesgue measure and for the surface measure of every top face. Hence the cell measures $\mu^\pm_\ell:=\mu^\pm|_{\Omega_\ell}$ do not depend on ties and sum to $\mu^\pm$.

\begin{proof}[Proof of \Cref{thm:twotier}]
The proof has two parts. The first part proves strict core regularity: for one slack $\varepsilon>0$, the sinks of each selling cell are dominated by its sources in the option score at slack $\varepsilon$. \Cref{thm:core-bundle} then gives a finite upper threshold. This part has four steps: the price identities and the split of the grand-cell sinks, the core faces, the peripheral face, and the combination. Step 3 is the key step: it couples sources on the face $z=1$ to sinks in every slice $z\ge d$, with a core reversal $R_C(h)$ (\Cref{app:core-L6}) at most $\hat\kappa(1-z)$ for a fixed $\hat\kappa<1$. The second part proves uniqueness above the upper threshold in two steps, the common exclusion cell and the lotteries and prices. Step 5 is its key step: one optimal plan of the barycentric dual identifies the excluded types of every optimal menu.

\emph{Step 1: the price identities and the split of the grand-cell sinks.}
Let $\nu^+:=\sum_{i\in C} f_C|_{F_i^C}$, of mass $\Gamma_C$. Every core face has $t\ge 1>\hat s_G$. Hence $\nu^+$ is the source measure $\nu^+_q$ of \Cref{os:single-cell} at $\rho=f_C$ for every $q\in [0,\hat s_C]$, and the divergence identity \eqref{eq:sc-divergence} becomes
\begin{equation}\label{eq:tt-div}
\int_{\{t\ge q\}}\Phi_C f_C\,\d\omega
=\Gamma_C-q\,g_C(q),
\qquad q\in[0,\hat s_C],
\end{equation}
with $q\,g_C(q)=0$ at $q=0$. By (S2), $z(\log f_p)'(z)=\beta-1$; hence $\Phi=\Phi_C+1+\beta$. The sources of $\Omega_1=\{t\ge \hat s_C,\ z\le d\}$ have mass $\Gamma_C d^\beta$, and \eqref{eq:tt-div} at $q=\hat s_C$ gives its sink mass. By $q\,g_C(q)=H_C(q)\,\bar G_C(q)$, mass balance on $\Omega_1$,
\[
\Gamma_C d^\beta
=d^\beta\left[\Gamma_C-\hat s_C\,g_C(\hat s_C)
 +(1+\beta)\bar G_C(\hat s_C)\right],
\]
is equivalent to
\begin{equation}\label{eq:tt-coreprice}
H_C(\hat s_C)=1+\beta.
\end{equation}

On the grand cell, define
\begin{gather*}
\varphi(z):=1+\beta-H_C(\tau(z)),\\
\mu^-_W:=\bigl[\Phi_C+H_C(\tau(z))\bigr]f\,\mathbf{1}_{\Omega_2}\,\d\omega\,\d z,
\qquad
\mu^-_G:=\varphi(z)\,f\,\mathbf{1}_{\Omega_2}\,\d\omega\,\d z.
\end{gather*}
Both are nonnegative: $\Phi_C\ge 0$ by (S1), and $\tau(z)\le \hat s_C$ for $z\ge d$, which with (ii) and \eqref{eq:tt-coreprice} gives $\varphi\ge 0$. Their sum is $\mu^-_2$. By \eqref{eq:tt-div}, the first part balances the core-face mass in each slice:
\[
\int_{\{t\ge q\}}\bigl[\Phi_C+H_C(q)\bigr]f_C\,\d\omega
=\Gamma_C-q\,g_C(q)+H_C(q)\,\bar G_C(q)=\Gamma_C,
\qquad q\in[0,\hat s_C].
\]
The sources of $\Omega_2$ are the core-face sources $\mu^+_{2,C}$ over $z\ge d$ and the sources $\mu^+_p:=\beta\,f_C(\omega)\,\d\omega\otimes\delta_{z=1}$ on the peripheral face $F_p=\{z=1\}$, which lies in $\Omega_2$ because $\hat s_G<1$. The slice identity gives $\mu^-_W(\Omega_2)=\Gamma_C(1-d^\beta)=\mu^+_{2,C}(\Omega_2)$. Hence mass balance on $\Omega_2$ is equivalent to
\begin{equation}\label{eq:tt-gbooks}
\int_d^1\varphi(z)\,f_p(z)\,\bar G_C(\tau(z))\,\d z
=\mu^-_G(\Omega_2)=\beta=\mu^+_p(F_p).
\end{equation}
Both identities hold, because (S3) assumes mass balance on both selling cells. The grand cell thus has two balanced pairs: $(\mu^+_{2,C},\mu^-_W)$ in each slice of $z$, and $(\mu^+_p,\mu^-_G)$, of common mass $\beta$.

\emph{Step 2: the core faces, slice by slice.}
For $q\in[0,\hat s_C]$, let $\nu^-_q:=[\Phi_C+H_C(q)]f_C\,\mathbf{1}_{\{t\ge q\}}\,\d\omega$. By (i), (ii), $\hat s_C<1=\underline t_C$ and $\Phi_C+H_C(q)\ge 0$, \Cref{os:sc-cert}(ii) at the largest cutoff $q_*=\hat s_C$ gives one $\lambda\in(0,1]$ with $(Q^\lambda_C)_\#\nu^-_q\preceq_{\mathrm{FOSD}}(Q^\lambda_C)_\#\nu^+$ for every $q\in[0,\hat s_C]$.

By \eqref{eq:tt-coreprice}, $\mu^-_1=[\Phi_C+H_C(\hat s_C)] f\,\mathbf{1}_{\Omega_1}$. Disintegrated over $z$ with weight $f_p$, $(\mu^+_1,\mu^-_1)$ has the slices $(\nu^+,\nu^-_{\hat s_C})$ for $z\le d$, and $(\mu^+_{2,C},\mu^-_W)$ has the slices $(\nu^+,\nu^-_{\tau(z)})$ for $z\ge d$. Fix $\varepsilon\in(0,\lambda]$ and a bounded Borel coordinatewise nondecreasing $\psi$. In each slice, $\psi\circ Q^\varepsilon_\ell(\cdot,z)$ is a bounded Borel coordinatewise nondecreasing function of $Q^\varepsilon_C$, hence of $Q^\lambda_C$ by nesting (\Cref{app:core}), and the slice dominance applies to it. Integrating over $z$ with weight $f_p$ gives, by Fubini's theorem,
\begin{equation}\label{eq:tt-slices}
\int\psi\circ Q^\varepsilon_1\,\d\mu^-_1\le\int\psi\circ Q^\varepsilon_1\,\d\mu^+_1,
\qquad
\int\psi\circ Q^\varepsilon_2\,\d\mu^-_W\le\int\psi\circ Q^\varepsilon_2\,\d\mu^+_{2,C}.
\end{equation}

\emph{Step 3: coupling the peripheral face to the remaining sinks.}
I couple $(\mu^+_p,\mu^-_G)$ with a core reversal at most $\hat\kappa\,h_p$ along every displacement, where $\hat\kappa:=\hat s_C/(1-d)$. Given $z$, the core law of $\mu^-_G$ is $f_C$ conditioned on $t\ge \tau(z)$; the coupling obtains it by increasing the core coordinates of a source drawn from $f_C$.

Let $u$ be uniform on $(0,1)$, $q\in [0,\hat s_C]$, $t=G_C^{-1}(u)$ and $t'=G_C^{-1}(G_C(q)+\bar G_C(q)\,u)$. Then $t$ has density $g_C$, $t'$ has the law of the core sum given $t\ge q$, and $0\le t'-t\le q$; hence this conditional law lies, in first-order stochastic dominance, between the law of $t$ and that law shifted by $q$. The lower bound holds because $G_C(q)+\bar G_C(q)\,u\ge u$; the upper bound holds because $G_C(t)=u$ and (iii) give
\[
G_C(t+q)=1-\bar G_C(t+q)\ge1-\bar G_C(q)\,\bar G_C(t)=G_C(q)+\bar G_C(q)\,u .
\]

Draw $z\in[d,1]$ with density $\varphi(z)\,f_p(z)\,\bar G_C(\tau(z))/\beta$, which integrates to one by \eqref{eq:tt-gbooks}, draw $u$ independently, and set $q=\tau(z)$. Let $c_t$ be the conditional law of $\omega$ under $f_C$ given $t(\omega)=t$. Efron's theorem \citep{Efron1965} gives $c_t\preceq_{\mathrm{FOSD}} c_{t'}$ for $t\le t'$, because the core coordinates are independent with log-concave densities. The monotone-coupling statement of Strassen's theorem (\Cref{thm:rel-strassen}) couples $(c_t,c_{t'})$ on $\{\omega\le \omega'\}$, through a Borel kernel $\Theta_{t,t'}$ (last paragraph). Draw $(\omega,\omega')$ from $\Theta_{t,t'}$, and let $\gamma_G$ be $\beta$ times the law of $((\omega,1),(\omega',z))$. As $t$ has density $g_C$ independently of $z$ and $\omega$ has law $c_t$ given $t$, $\omega$ has density $f_C$ and the first marginal is $\mu^+_p$. Given $z$, $\omega'$ has density $f_C(\omega')\mathbf{1}_{\{t(\omega')\ge\tau(z)\}}/\bar G_C(\tau(z))$, and multiplying by the density of $z$ and by $\beta$ gives $\mu^-_G$.

The participation boundary satisfies $\tau(z)\le\hat\kappa(1-z)$ on $[d,1]$, because on $[d,\hat s_G]$ the ratio $(\hat s_G-z)/(1-z)$ is decreasing with value $\hat s_C/(1-d)$ at $z=d$, and $\tau=0$ on $[\hat s_G,1]$. Also $\hat\kappa<1$, because $\hat s_C<1-d$ is $\hat s_G<1$. Every displacement $h=(\omega-\omega',1-z)$ of $\gamma_G$ has $h_p=1-z\ge0$ and $h_i\le0$ for $i\in C$. Hence $\sum_{i\in C}h_i^-=t'-t\le\tau(z)\le\hat\kappa\, h_p$, and for $\varepsilon\le1-\hat\kappa$ every coordinate of $Q^\varepsilon_2(h)$ is nonnegative: its peripheral coordinates are $(h_p,0)$, and for every nonempty $T\subseteq C$,
\[
\varepsilon\sum_{i\in T}h_i+(1-\varepsilon)\bigl(t(h)+h_p\bigr)\;\ge\;(1-\varepsilon)\,h_p-\sum_{i\in C}h_i^-\;\ge\;(1-\varepsilon-\hat\kappa)\,h_p\;\ge\;0 .
\]
Integrating a bounded Borel coordinatewise nondecreasing function of $Q^\varepsilon_2$ along $\gamma_G$ therefore gives
\begin{equation}\label{eq:tt-face}
\int\psi\circ Q^\varepsilon_2\,\d\mu^-_G\le\int\psi\circ Q^\varepsilon_2\,\d\mu^+_p .
\end{equation}

\emph{Step 4: the combination.}
Set $\varepsilon:=\min\{\lambda,1-\hat\kappa\}>0$. The first inequality of \eqref{eq:tt-slices} is the dominance on the bare-core cell at slack $\varepsilon$. On the grand cell, $\mu^+_2=\mu^+_{2,C}+\mu^+_p$ and $\mu^-_2=\mu^-_W+\mu^-_G$, and the sum of the second inequality of \eqref{eq:tt-slices} and \eqref{eq:tt-face} is the dominance at slack $\varepsilon$. By nesting (\Cref{app:core}), $(f,M)$ is \hyperref[def:align-reg]{strictly core-regular}.

The menu $M$ is finite on $\mathcal{F}_C$ with the positive prices $\alpha\hat s_C$ and $\alpha\hat s_G$, and it is inclusive (shown before \Cref{thm:twotier}); $f$ is admissible by (iv), with the radial score $\Phi = \Phi_C + 1 + \beta$, positive by (S1) and (S2). By \Cref{os:sc-extension}, \Cref{thm:core-bundle} holds for $f$. Hence the upper threshold $\alpha^*(C,M)$ is finite, and at every premium $\alpha>\alpha^*(C,M)$ the two-tier menu is optimal in the unrestricted problem and every optimal menu is core-peripheral. Fix such a premium $\alpha$ for Steps 5 and 6.

\emph{Step 5: the common exclusion cell.}
Every optimal menu at $\alpha$ has the exclusion cell $\Omega_0$, up to a null set.\footnote{\citet{DDT2017} use a fixed optimal dual plan in the same way to show that the optimal mechanism of one of their examples is unique.} Steps 1, 3 and 4 of the proof of \Cref{cor:premium}(ii) apply to $v^\alpha$, by (iv) and \Cref{os:sc-extension}. Let $\pi$ attain the least barycentric cost under $v^\alpha$ (Step 1), and let $\pi_0:=\pi(\{0\}\times\cdot)$. By Steps 3 and 4, in every optimal menu at $\alpha$, $\pi$-almost every pair $(x,x')$ has $x'$ taking the outside option if $x=0$ and a lottery without the empty bundle if $x\ne0$. The exclusion cell $A'$ of every optimal menu at $\alpha$ therefore satisfies $\mu^-|_{A'}=\pi_0$, because $\mu^-$ is the second marginal of $\pi$. As $\mu^-$ has a positive density, all these cells agree with the exclusion cell $\Omega_0$ of $M$ up to a null set.

\emph{Step 6: the lotteries and prices.}
Let $M'$ be an optimal menu at $\alpha$. It sells only bundles in $\mathcal{F}_C=\{C,C\cup\{p\}\}$. Its indirect utility is therefore $W=\max\{0,\alpha t+\chi(z)\}$ for a convex $\chi$ with slopes in $[0,\alpha]$. By Step 5, $\{t<-\chi(z)/\alpha\}=\Omega_0$ up to a null set. Hence $\chi(z)=-\alpha\min\{\hat s_C,\tau(z)\}$ on $[0,\hat s_G]$, with slope $\alpha$ on $[d,\hat s_G]$. Convexity and the slope bound then give $\chi(z)=\alpha(z-\hat s_G)$ on $[\hat s_G,1]$, and $W=U_M$. The option that almost every type takes therefore has the gradient and the price of the option of $M$. It also has its lottery, because the outside option has the gradient $0$ and a lottery $\sigma$ on $\mathcal{F}_C$ has the gradient $\alpha(\mathbf{1}_C,\sigma(C\cup\{p\}))$.

The two-tier menu is therefore uniquely optimal at every premium above its upper threshold: Steps 1 to 4 and \Cref{thm:core-bundle} give its optimality there, and Steps 5 and 6 its uniqueness.
\end{proof}

\paragraph{Regularity and measurable couplings used in the proof.}
Each log-concave $f_i$ is bounded and unimodal. Hence $\int_0^1|f_i'|\le2\sup f_i$ and $\omega_i f_i(\omega_i)\to0$ at zero. The density $g_C$ is log-concave by Pr\'ekopa's theorem \citep{Prekopa1973}, bounded, and strictly positive and continuous on $(0,\bar t_C)$, by (S1) for one core coordinate and by convolution for two or more. These properties, the positivity and continuous differentiability of $f_i$ near $1$ and Fubini's theorem give (i). Log-concavity and positivity of $g_C$ make its failure rate $g_C/\bar G_C$ positive and nondecreasing \citep{BagnoliBergstrom2005}, which gives (ii), with continuity at zero from the boundedness of $g_C$. Integrating the nondecreasing failure rate over $[q,q+r]$ and over $[0,r]$ gives the new-better-than-used property, which is (iii). For (iv), the further conditions on $f$ follow from $zf_p'=(\beta-1)f_p\in L^1$, $zf_p(z)=\beta z^\beta\to0$, the boundedness, positivity on $(0,N)$ and continuity of $g_C*f_p$, and the positive $C^1$ factors near the top corner. That continuity follows from the $L^1$ translation continuity of $f_p$ and the boundedness of $g_C$. The radial score $\Phi=\Phi_C+1+\beta$ of $f$ is positive.

For the kernel of Step 3, when $|C|\ge2$, fix $c\in C$ and let $c_t$ have the density $f_C/g_C(t)$ on the level set $\{t(\omega)=t\}$, parametrized by $\omega_{-c}$ through $\omega_c=t-\sum_{i\in C\setminus\{c\}}\omega_i$. By dominated convergence with the bounded marginals, $t\mapsto c_t$ is weakly continuous on $(0,\bar t_C)$, because the indicators in the density jump only on two Lebesgue-null hyperplanes. It extends weakly continuously by $c_0=\delta_0$ and $c_{\bar t_C}=\delta_{\mathbf{1}}$, because the sum levels shrink to the two corners. For $|C|=1$, $c_t=\delta_t$. Weak continuity extends Efron's order to all $0\le t\le t'\le \bar t_C$, and Strassen's theorem gives the couplings for every such pair. Their correspondence has nonempty compact values and a closed graph on this compact triangle, and the Kuratowski--Ryll-Nardzewski theorem \citep{KuratowskiRyllNardzewski1965} gives a Borel kernel $\Theta_{t,t'}$.

\section{Pure bundling with one-dimensional types}\label{os:curvature}

The type $t\in[0,1]$ has a density $f$, continuous and positive on $(0,1)$, and every bundle of the $N\ge2$ goods is for sale. Each $v_i$ is strictly increasing, vanishes at $0$, is absolutely continuous on $[0,1]$ and $C^1$ on $(0,1)$. Only the grand bundle earns the premium, $v_{[N]}=\alpha\sum_iv_i$; other bundles are valued additively. The Myerson field $a$, the virtual values and the revenue of selling $B$ alone to the types above $t$ are
\[
a(t)=\frac{1-F(t)}{f(t)},\qquad \psi=v-av',\qquad
\mathrm{Rev}_B(t)=v_B(t)[1-F(t)]=\int_t^1\psi_B\,\d F.
\]
Thus $\psi_{[N]}=\alpha\sum_i\psi_i$ and $\mathrm{Rev}_B'=-f\psi_B$ on $(0,1)$. Pure bundling at price $v_{[N]}(\hat t)$ sells the grand bundle to the types above its \emph{marginal type} $\hat t$. A \emph{critical type} maximizes $\mathrm{Rev}_{[N]}=\alpha\sum_i\mathrm{Rev}_i$. It is independent of $\alpha$ and lies in $(0,1)$, as $\mathrm{Rev}_{[N]}$ is positive there and zero at $0$ and $1$, so $\sum_i\psi_i=0$ at it. The valuation map is \emph{\hyperref[def:tangent]{tangent} at} $t$ when $v(t)$ lies on the line spanned by $v'(t)$. For $t\in(0,1)$, equivalently, $v'(t)\ne0$ and the tangent line to $v$ at $t$ passes through $v(0)=0$.

\begin{proposition}[Pure bundling with one-dimensional types]\label{prop:curvature-1d}
\begin{enumerate}[label=\textnormal{(\roman*)}]
\item If pure bundling with marginal type $\hat t$ is optimal at some premium $\alpha>0$, then $\psi_i(\hat t)=0$ for every good $i$, and $v$ is tangent at $\hat t$. In particular, if $v$ is tangent at no type in $(0,1)$, pure bundling is suboptimal at every premium.
\item If every $\psi_i$ is nonpositive below some $t^\circ\in(0,1)$ and nonnegative above it, pure bundling at price $v_{[N]}(t^\circ)$ is optimal at every premium $\alpha\ge1$.
\item If each $\psi_i$ is negative below some $t_i^\circ\in(0,1)$ and positive above it, pure bundling is optimal at some premium if and only if $t_1^\circ=\dots=t_N^\circ$, if and only if $v$ is tangent at a critical type; it is then optimal at every premium $\alpha\ge1$.
\end{enumerate}
\end{proposition}

\begin{proof}
Part (i), the key step, perturbs pure bundling by a lottery on one good; part (iii) combines it with the bound of part (ii).

\emph{(i) Every good's virtual value vanishes at the marginal type.}
Let pure bundling with marginal type $\hat t$ be optimal at $\alpha$; as every other price is feasible, $\hat t$ is a critical type. For a good $i$ and $\tilde t\in(0,\hat t)$, add a lottery that gives good $i$ with probability $\alpha\wedge1$ at price $(\alpha\wedge1)v_i(\tilde t)$, and lower the grand bundle's price by $(\alpha\wedge1)[v_i(\hat t)-v_i(\tilde t)]$; both prices are positive. The lottery's surplus has the sign of $t-\tilde t$. The grand bundle's value net of the lottery's, $(\alpha-\alpha\wedge1)v_i+\alpha\sum_{j\ne i}v_j$, is strictly increasing, and the two surpluses are equal at $\hat t$. Hence only the types in $(\tilde t,\hat t)$ change their choice, to the lottery, and revenue changes by $(\alpha\wedge1)[\mathrm{Rev}_i(\tilde t)-\mathrm{Rev}_i(\hat t)]\le0$. As this holds for every $\tilde t<\hat t$, $-f(\hat t)\psi_i(\hat t)=\mathrm{Rev}_i'(\hat t)\ge0$. With $\sum_i\psi_i(\hat t)=0$, every $\psi_i(\hat t)$ vanishes, hence so does $\psi(\hat t)$, and $v(\hat t)=a(\hat t)v'(\hat t)$ is tangency at $\hat t$.

\emph{(ii) Pure bundling attains the virtual-surplus bound.}
For an incentive-compatible, individually rational mechanism with allocation probabilities $\sigma_B(t)$ and indirect utility $U(t)$, the one-dimensional envelope formula and integration by parts, valid for absolutely continuous values, give
\[
\mathrm{Rev}=\int_0^1\sum_B\sigma_B\psi_B\,\d F-U(0)\le\int_0^1c(\psi)\,\d F,\qquad c(\psi)=\max\Bigl\{0,\max_B\psi_B\Bigr\},
\]
because $\sigma(t)$ lies in the lottery simplex and $U(0)\ge0$. This is the bound of \Cref{thm:flow-duality} for the Myerson flow $J:=v_\#(f\psi\,\d t)$ of \Cref{sec:typespace}, a finite measure as $fv$ and $(1-F)v'$ are integrable: integration by parts with the absolutely continuous $\phi\circ v$ gives $\int\nabla\phi\cdot\d J=\langle T_\Pi,\phi\rangle$ for $\phi$ of class $C^1$ near $X$, so $J$ is admissible, with cost $\int_0^1c(\psi)\,\d F$ as $v$ is injective.

Under part (ii)'s hypothesis with $\alpha\ge1$, $c(\psi)=0$ below $t^\circ$, where every $\psi_B\le0$, and $c(\psi)=\psi_{[N]}$ above $t^\circ$, where every $B\ne[N]$ has $\psi_B\le\sum_i\psi_i\le\alpha\sum_i\psi_i=\psi_{[N]}$. The bound is therefore $\int_{t^\circ}^1\psi_{[N]}\,\d F=\mathrm{Rev}_{[N]}(t^\circ)$, which pure bundling at price $v_{[N]}(t^\circ)$ earns.

\emph{(iii) The three conditions are equivalent.}
First, if pure bundling is optimal at some premium, every $\psi_i$ vanishes at its marginal type by part (i), so the strict sign changes make every $t_i^\circ$ equal to it. Second, if every $t_i^\circ$ equals $t^\circ$, part (ii) gives optimality at price $v_{[N]}(t^\circ)$ at every $\alpha\ge1$, and $t^\circ$ is the unique critical type, as $\mathrm{Rev}_{[N]}'=-\alpha f\sum_i\psi_i$ changes sign from positive to negative there. Each $\psi_i$ is continuous on $(0,1)$, as $f$, $F$ and $v'$ are, so $\psi_i(t^\circ)=0$, and $v$ is tangent at $t^\circ$ as in part (i). Third, let $v(t)=\lambda v'(t)$ at a critical type $t$. Then $\sum_iv_i'(t)\ne0$, as $\sum_iv_i(t)=\lambda\sum_iv_i'(t)$ is positive, and
\[
0=\sum_i\psi_i(t)=\bigl(\lambda-a(t)\bigr)\sum_iv_i'(t).
\]
Hence $\lambda=a(t)$, every $\psi_i(t)=(\lambda-a(t))v_i'(t)$ vanishes, and $t_1^\circ=\dots=t_N^\circ=t$. Implications one and two give the first equivalence and optimality at every $\alpha\ge1$; two and three give the second.
\end{proof}

\section{An explicit upper-threshold bound in the uniform benchmark}\label{os:dual-transport}

In the two-good iid uniform benchmark, pure bundling is optimal at every premium $\alpha \ge \bar\alpha_{\mathrm{unif}} = (9 + 4\sqrt 6)/15 \approx 1.253$. The plan that certifies it through \Cref{prop:saddle-D} couples the origin atom to the sinks of the exclusion cell along nonpositive displacements, the top face $F_1$ to the half $T := \Omega_1 \cap \{x_2 \le x_1\}$ of the selling cell, and $F_2$ to the other half by reflection. The grand bundle covers the displacement from a source $(1, t) \in F_1$ to a sink if and only if the sink lies below the line through $(1, t)$ with slope $-(\alpha - 1)/\alpha$. These lines are parallel, so the reachable sets are nested. Hence a coupling along covered displacements exists if and only if the Hall condition $3\, \lambda_D(R(t)) \ge t$ holds for every $t \in (0,1)$, where $R(t)$ is the part of $T$ below the line through $(1, t)$ \citep{Strassen1965, Kellerer1984}. A direct computation shows that the condition holds for every $t$ if and only if $\alpha \ge \bar\alpha_{\mathrm{unif}}$. At that premium the condition binds at one interior $t$ and at $t = 1$.

\end{document}